\documentclass[reqno,12pt,a4paper]{article}

\usepackage{amsfonts,amssymb,amsthm, 
amsmath,amscd, bbm, bm, mathabx,
mathrsfs,delarray,subfigure,framed,graphicx}

\usepackage[dvipsnames]{xcolor}[svgnames,x11names]

\usepackage{hyperref,enumitem}
\usepackage{cite}
\usepackage{xypic}
\usepackage{pstricks}
\usepackage{sectsty}
\usepackage{relsize} 

\definecolor{light-gray}{gray}{0.96}
\definecolor{shadecolor}{named}{light-gray}

\usepackage[scr=rsfs,cal=dutchcal]{mathalfa}

\usepackage{qcircuit}

\usepackage{scalerel}

\usepackage{fancyhdr}
\fancyheadoffset[RE,LO]{0.\textwidth} 

\hypersetup{
  colorlinks,
  linkcolor=blue,
  linktoc=all,
  citecolor=blue,
   urlcolor=blue
}

\usepackage{sectsty}

\colorlet{sectitlecolor}{blue}
\colorlet{sectboxcolor}{white}
\colorlet{secnumcolor}{blue}

\sectionfont{\color{sectitlecolor}}

\makeatletter
\renewcommand\@seccntformat[1]{%
  \colorbox{sectboxcolor}{\textcolor{secnumcolor}{\csname the#1\endcsname}}%
  \quad
}
\makeatother

\usepackage{etoolbox}
\patchcmd{\thebibliography}{\section*{\refname}}{}{}{}

\newcommand{\rmD}{{\mathrm{D}}}
\newcommand{\rmE}{{\mathrm{E}}}

\DeclareSymbolFont{sfletters}{OML}{cmbrm}{m}{it}  

\DeclareMathSymbol{\sfGamma}{\mathord}{sfletters}{"00}
\DeclareMathSymbol{\sfDelta}{\mathord}{sfletters}{"01}
\DeclareMathSymbol{\sfTheta}{\mathord}{sfletters}{"02}
\DeclareMathSymbol{\sfLambda}{\mathord}{sfletters}{"03}
\DeclareMathSymbol{\sfXi}{\mathord}{sfletters}{"04}
\DeclareMathSymbol{\sfPi}{\mathord}{sfletters}{"05}
\DeclareMathSymbol{\sfSigma}{\mathord}{sfletters}{"06}
\DeclareMathSymbol{\sfUpsilon}{\mathord}{sfletters}{"07}
\DeclareMathSymbol{\sfPhi}{\mathord}{sfletters}{"08}
\DeclareMathSymbol{\sfPsi}{\mathord}{sfletters}{"09}
\DeclareMathSymbol{\sfOmega}{\mathord}{sfletters}{"0A}
\DeclareMathSymbol{\sfalpha}{\mathord}{sfletters}{"0B}
\DeclareMathSymbol{\sfbeta}{\mathord}{sfletters}{"0C}
\DeclareMathSymbol{\sfgamma}{\mathord}{sfletters}{"0D}
\DeclareMathSymbol{\sfdelta}{\mathord}{sfletters}{"0E}
\DeclareMathSymbol{\sfepsilon}{\mathord}{sfletters}{"0F}
\DeclareMathSymbol{\sfzeta}{\mathord}{sfletters}{"10}
\DeclareMathSymbol{\sfeta}{\mathord}{sfletters}{"11}
\DeclareMathSymbol{\sftheta}{\mathord}{sfletters}{"12}
\DeclareMathSymbol{\sfiota}{\mathord}{sfletters}{"13}
\DeclareMathSymbol{\sfkappa}{\mathord}{sfletters}{"14}
\DeclareMathSymbol{\sflambda}{\mathord}{sfletters}{"15}
\DeclareMathSymbol{\sfmu}{\mathord}{sfletters}{"16}
\DeclareMathSymbol{\sfnu}{\mathord}{sfletters}{"17}
\DeclareMathSymbol{\sfxi}{\mathord}{sfletters}{"18}
\DeclareMathSymbol{\sfpi}{\mathord}{sfletters}{"19}
\DeclareMathSymbol{\sfrho}{\mathord}{sfletters}{"1A}
\DeclareMathSymbol{\sfsigma}{\mathord}{sfletters}{"1B}
\DeclareMathSymbol{\sftau}{\mathord}{sfletters}{"1C}
\DeclareMathSymbol{\sfupsilon}{\mathord}{sfletters}{"1D}
\DeclareMathSymbol{\sfphi}{\mathord}{sfletters}{"1E}
\DeclareMathSymbol{\sfchi}{\mathord}{sfletters}{"1F}
\DeclareMathSymbol{\sfpsi}{\mathord}{sfletters}{"20}
\DeclareMathSymbol{\sfomega}{\mathord}{sfletters}{"21}
\DeclareMathSymbol{\sfvarepsilon}{\mathord}{sfletters}{"22}
\DeclareMathSymbol{\sfvartheta}{\mathord}{sfletters}{"23}
\DeclareMathSymbol{\sfvarpi}{\mathord}{sfletters}{"24}
\DeclareMathSymbol{\sfvarrho}{\mathord}{sfletters}{"25}
\DeclareMathSymbol{\sfvarsigma}{\mathord}{sfletters}{"26}
\DeclareMathSymbol{\sfvarphi}{\mathord}{sfletters}{"27}

\DeclareMathSymbol{\spartial}{\mathord}{sfletters}{"40}

\DeclareMathSymbol{\sfA}{\mathord}{sfletters}{"41}
\DeclareMathSymbol{\sfB}{\mathord}{sfletters}{"42}
\DeclareMathSymbol{\sfC}{\mathord}{sfletters}{"43}
\DeclareMathSymbol{\sfD}{\mathord}{sfletters}{"44}
\DeclareMathSymbol{\sfE}{\mathord}{sfletters}{"45}
\DeclareMathSymbol{\sfF}{\mathord}{sfletters}{"46}
\DeclareMathSymbol{\sfG}{\mathord}{sfletters}{"47}
\DeclareMathSymbol{\sfH}{\mathord}{sfletters}{"48}
\DeclareMathSymbol{\sfI}{\mathord}{sfletters}{"49}
\DeclareMathSymbol{\sfJ}{\mathord}{sfletters}{"4A}
\DeclareMathSymbol{\sfK}{\mathord}{sfletters}{"4B}
\DeclareMathSymbol{\sfL}{\mathord}{sfletters}{"4C}
\DeclareMathSymbol{\sfM}{\mathord}{sfletters}{"4D}
\DeclareMathSymbol{\sfN}{\mathord}{sfletters}{"4E}
\DeclareMathSymbol{\sfO}{\mathord}{sfletters}{"4F}
\DeclareMathSymbol{\sfP}{\mathord}{sfletters}{"50}
\DeclareMathSymbol{\sfQ}{\mathord}{sfletters}{"51}
\DeclareMathSymbol{\sfR}{\mathord}{sfletters}{"52}
\DeclareMathSymbol{\sfS}{\mathord}{sfletters}{"53}
\DeclareMathSymbol{\sfT}{\mathord}{sfletters}{"54}
\DeclareMathSymbol{\sfU}{\mathord}{sfletters}{"55}
\DeclareMathSymbol{\sfV}{\mathord}{sfletters}{"56}
\DeclareMathSymbol{\sfW}{\mathord}{sfletters}{"57}
\DeclareMathSymbol{\sfX}{\mathord}{sfletters}{"58}
\DeclareMathSymbol{\sfY}{\mathord}{sfletters}{"59}
\DeclareMathSymbol{\sfZ}{\mathord}{sfletters}{"5A}
\DeclareMathSymbol{\sfa}{\mathord}{sfletters}{"61}
\DeclareMathSymbol{\sfb}{\mathord}{sfletters}{"62}
\DeclareMathSymbol{\sfc}{\mathord}{sfletters}{"63}
\DeclareMathSymbol{\sfd}{\mathord}{sfletters}{"64}
\DeclareMathSymbol{\sfe}{\mathord}{sfletters}{"65}
\DeclareMathSymbol{\sff}{\mathord}{sfletters}{"66}
\DeclareMathSymbol{\sfg}{\mathord}{sfletters}{"67}
\DeclareMathSymbol{\sfh}{\mathord}{sfletters}{"68}
\DeclareMathSymbol{\sfi}{\mathord}{sfletters}{"69}
\DeclareMathSymbol{\sfj}{\mathord}{sfletters}{"6A}
\DeclareMathSymbol{\sfk}{\mathord}{sfletters}{"6B}
\DeclareMathSymbol{\sfl}{\mathord}{sfletters}{"6C}
\DeclareMathSymbol{\sfm}{\mathord}{sfletters}{"6D}
\DeclareMathSymbol{\sfn}{\mathord}{sfletters}{"6E}
\DeclareMathSymbol{\sfo}{\mathord}{sfletters}{"6F}
\DeclareMathSymbol{\sfp}{\mathord}{sfletters}{"70}
\DeclareMathSymbol{\sfq}{\mathord}{sfletters}{"71}
\DeclareMathSymbol{\sfr}{\mathord}{sfletters}{"72}
\DeclareMathSymbol{\sfs}{\mathord}{sfletters}{"73}
\DeclareMathSymbol{\sft}{\mathord}{sfletters}{"74}
\DeclareMathSymbol{\sfu}{\mathord}{sfletters}{"75}
\DeclareMathSymbol{\sfv}{\mathord}{sfletters}{"76}
\DeclareMathSymbol{\sfw}{\mathord}{sfletters}{"77}
\DeclareMathSymbol{\sfx}{\mathord}{sfletters}{"78}
\DeclareMathSymbol{\sfy}{\mathord}{sfletters}{"79}
\DeclareMathSymbol{\sfz}{\mathord}{sfletters}{"7A}

\usepackage[LGR,T1]{fontenc}
\usepackage{amsmath,etoolbox}

\newcommand{\declarebsfgreek}[2]{%
  \protected\csdef{bsf#1}{\mathord{\text{\bsfgreekfont#2}}}%
}
\newcommand{\bsfgreekfont}{\usefont{LGR}{cmss}{bx}{it}}

\declarebsfgreek{alpha}{a}
\declarebsfgreek{beta}{b}
\declarebsfgreek{gamma}{g}
\declarebsfgreek{delta}{d}
\declarebsfgreek{epsilon}{e}
\declarebsfgreek{zeta}{z}
\declarebsfgreek{eta}{h}
\declarebsfgreek{theta}{j}
\declarebsfgreek{iota}{i}
\declarebsfgreek{kappa}{k}
\declarebsfgreek{lambda}{l}
\declarebsfgreek{mu}{m}
\declarebsfgreek{nu}{n}
\declarebsfgreek{xi}{x}
\declarebsfgreek{omicron}{o}
\declarebsfgreek{pi}{p}
\declarebsfgreek{rho}{r}
\declarebsfgreek{sigma}{s}
\declarebsfgreek{tau}{t}
\declarebsfgreek{upsilon}{u}
\declarebsfgreek{phi}{f}
\declarebsfgreek{chi}{q}
\declarebsfgreek{psi}{y}
\declarebsfgreek{omega}{w}
\declarebsfgreek{varsigma}{c}

\declarebsfgreek{Gamma}{G}
\declarebsfgreek{Delta}{D}
\declarebsfgreek{Upsilon}{U}
\declarebsfgreek{Omega}{W}
\declarebsfgreek{Theta}{J}
\declarebsfgreek{Lambda}{L}
\declarebsfgreek{Xi}{X}
\declarebsfgreek{Pi}{P}
\declarebsfgreek{Sigma}{S}
\declarebsfgreek{Phi}{F}
\declarebsfgreek{Psi}{Y}

\newcommand{\declarebsfitalic}[2]{%
  \protected\csdef{bsf#1}{\mathord{\text{\bsfitalicfont#2}}}%
}
\newcommand{\bsfitalicfont}{\usefont{T1}{cmss}{bx}{it}}

\declarebsfitalic{a}{a}
\declarebsfitalic{b}{b}
\declarebsfitalic{c}{c}
\declarebsfitalic{d}{d}
\declarebsfitalic{e}{e}
\declarebsfitalic{f}{f}
\declarebsfitalic{g}{g}
\declarebsfitalic{h}{h}
\declarebsfitalic{i}{i}
\declarebsfitalic{j}{j}
\declarebsfitalic{k}{k}
\declarebsfitalic{l}{l}
\declarebsfitalic{m}{m}
\declarebsfitalic{n}{n}
\declarebsfitalic{o}{o}
\declarebsfitalic{p}{p}
\declarebsfitalic{q}{q}
\declarebsfitalic{r}{r}
\declarebsfitalic{s}{s}
\declarebsfitalic{t}{t}
\declarebsfitalic{u}{u}
\declarebsfitalic{v}{v}
\declarebsfitalic{x}{x}
\declarebsfitalic{y}{y}
\declarebsfitalic{z}{z}

\declarebsfitalic{A}{A}
\declarebsfitalic{B}{B}
\declarebsfitalic{C}{C}
\declarebsfitalic{D}{D}
\declarebsfitalic{E}{E}
\declarebsfitalic{F}{F}
\declarebsfitalic{G}{G}
\declarebsfitalic{H}{H}
\declarebsfitalic{I}{I}
\declarebsfitalic{J}{J}
\declarebsfitalic{K}{K}
\declarebsfitalic{L}{L}
\declarebsfitalic{M}{M}
\declarebsfitalic{N}{N}
\declarebsfitalic{O}{O}
\declarebsfitalic{P}{P}
\declarebsfitalic{Q}{Q}
\declarebsfitalic{R}{R}
\declarebsfitalic{S}{S}
\declarebsfitalic{T}{T}
\declarebsfitalic{U}{U}
\declarebsfitalic{V}{V}
\declarebsfitalic{X}{X}
\declarebsfitalic{Y}{Y}
\declarebsfitalic{Z}{Z}

\newcommand{\bfsfH}{\sfH\!\!\!\!\!\sfH}

\newcommand{\msA}{{\sf{A}}}

\newcommand{\msF}{{\sf{F}}}
\newcommand{\msG}{{\sf{G}}}
\newcommand{\msH}{{\sf{H}}}
\newcommand{\msI}{{\sf{I}}}

\newcommand{\msK}{{\sf{K}}}

\newcommand{\msM}{{\sf{M}}}

\newcommand{\msP}{{\sf{P}}}

\newcommand{\msR}{{\sf{R}}}

\newcommand{\msU}{{\sf{U}}}

\newcommand{\msZ}{{\sf{Z}}}

\usepackage{yfonts}

\DeclareFontFamily{U}{yfrak}{}
\DeclareFontShape{U}{yfrak}{m}{n}{<->yfrak}{}

\newcommand{\mathgoth}[1]{\text{\usefont{U}{yfrak}{m}{n}#1}}

\newcommand{\fkQ}{{\mathgoth{Q}}}

\newcommand{\bbC}{{\mathbb{C}}}

\newcommand{\bbF}{{\mathbb{F}}}

\newcommand{\bbH}{{\mathbb{H}}}

\newcommand{\bbN}{{\mathbb{N}}}

\newcommand{\bbZ}{{\mathbb{Z}}}

\newcommand{\scH}{{\matheul{H}}}

\newcommand{\scK}{{\matheul{K}}}
\newcommand{\scL}{{\matheul{L}}}

\newcommand{\clA}{{\mathcal{A}}}
\newcommand{\clB}{{\mathcal{B}}}
\newcommand{\clC}{{\mathcal{C}}}
\newcommand{\clD}{{\mathcal{D}}}
\newcommand{\clE}{{\mathcal{E}}}
\newcommand{\clF}{{\mathcal{F}}}

\newcommand{\clK}{{\mathcal{K}}}
\newcommand{\clL}{{\mathcal{L}}}
\newcommand{\clM}{{\mathcal{M}}}

\newcommand{\clP}{{\mathcal{P}}}
\newcommand{\clQ}{{\mathcal{Q}}}
\newcommand{\clR}{{\mathcal{R}}}
\newcommand{\clS}{{\mathcal{S}}}
\newcommand{\clT}{{\mathcal{T}}}
\newcommand{\clU}{{\mathcal{U}}}

\newcommand{\clX}{{\mathcal{X}}}
\newcommand{\clY}{{\mathcal{Y}}}
\newcommand{\clZ}{{\mathcal{Z}}}

\newcommand{\bra}[1]{{\langle{#1}|}}

\newcommand{\ket}[1]{{|{#1}\rangle}}

\usepackage{accents}

\sectionfont{\fontsize{12}{15}\selectfont}
\subsectionfont{\fontsize{12}{15}\selectfont}

\usepackage{enumitem}


\DeclareMathOperator{\id}{id}

\DeclareMathOperator{\Hom}{Hom}

\DeclareMathOperator{\Ver}{Ver}
\DeclareMathOperator{\Edg}{Edg}
\DeclareMathOperator{\Obj}{Obj}

\DeclareMathOperator{\Funct}{Funct}
\DeclareMathOperator{\SMFunct}{SMFunct}
\DeclareMathOperator{\SMEnd}{SMEnd}

\DeclareMathOperator{\End}{End}

\DeclareMathOperator{\tr}{tr}

\numberwithin{equation}{subsection} 
\numberwithin{subsection}{section} 

\newcommand{\ceqref}[1]{{\textcolor{blue}{\eqref{#1}}}}
\newcommand{\cref}[1]{{\textcolor{blue}{\ref{#1}}}}
\newcommand{\ccite}[1]{{\textcolor{blue}{\hspace{-.40pt}\cite{#1}}}}

\newcommand{\sss}{{\hbox{\large $\sum$}}}
\newcommand{\ppp}{{\hbox{\large $\prod$}}}

\newcommand{\uuu}{{\hbox{\large $\bigcup$}}}

\newcommand{\ul}[1]{{\underline{#1}}}
\newcommand{\ol}[1]{{\overline{#1}}}

\newcommand{\hfpt}{\hspace{.75pt}}
\newcommand{\mhfpt}{\hspace{-.75pt}}

\newcommand{\smallsss}{{\hbox{$\sum$}}}

\newcommand\mycom[2]{\genfrac{}{}{-3pt}{}{#1}{#2}}

\font\euler=eusm10 at 12.8 truept
\font\scripteuler=eusm7
\font\scriptscripteuler=eusm5 
\def\eul{\fam=12}
\newcommand{\matheul}[1]{{{\eul #1}}}

\newtheorem{defi}{{\sf Definition}}[subsection]
\newtheorem{prop}{{\sf Proposition}}[subsection]
\newtheorem{theor}{{\sf Theorem}}[subsection]

\newtheorem{exa}{{\sf Example}}[subsection]
\newtheorem{rem}{{\sf Remark}}[subsection]
\newtheorem{cor}{{\sf Corollcorary}}[subsection]

\DeclareMathSymbol{*}{\mathbin}{symbols}{"03} 

\makeatletter
\renewcommand{\eqref}[1]{\tagform@{\ref{#1}}}
\def\maketag@@@#1{\hbox{#1}}
\makeatother

\DeclareMathVersion{normal2}

\begin{document}

\thispagestyle{empty} 

\vskip1.5cm
\begin{large} 
{\flushleft\textcolor{blue}{\sffamily\bfseries Calibrated hypergraph states:}}
{\flushleft\textcolor{blue}{\sffamily\bfseries I calibrated hypergraph and multi qudit state monads}} 
\end{large}
\vskip1.2cm
\hrule height 1.5pt
\vskip1.2cm
{\flushleft{\sffamily \bfseries Roberto Zucchini}\\
\it Department of Physics and Astronomy,\\
University of Bologna,\\
I.N.F.N., Bologna division,\\
viale Berti Pichat, 6/2\\
Bologna, Italy\\
Email: \textcolor{blue}{\tt \href{mailto:roberto.zucchini@unibo.it}{roberto.zucchini@unibo.it}}, 
\textcolor{blue}{\tt \href{mailto:zucchinir@bo.infn.it}{zucchinir@bo.infn.it}}}


\vskip.9cm 
{\flushleft\sc
Abstract:} 
Hypergraph states are a special kind of multipartite states encoded by
hypergraphs. They play a significant role in 
quantum error correction, measurement--based quantum computation, quantum non locality and entanglement.
In a series of two papers, we introduce and study calibrated hypergraph states,
a broad generalization of weighted hypergraph states codified by 
hypergraphs equipped with calibrations, an ample extension of weightings. 
We propose as a guiding principle that a constructive theory of hypergraph
states must be based on a categorical framework for
hypergraphs on one hand and multi qudit states on the other constraining hypergraph
states enough to render the determination of their general structure possible. 
In this first paper, we introduce graded $\varOmega$ monads, concrete Pro
categories isomorphic to the Pro category $\varOmega$ of finite von Neumann ordinals and equipped with
an associative and unital
graded multiplication, and their morphisms, maps of $\varOmega$ monads compatible with their monadic structure. 
We then show that both calibrated hypergraphs and multi qudit states 
naturally organize in graded $\varOmega$ monads. In this way, we lay the foundation for the construction of calibrated
hypergraph state map as a special morphism of these $\varOmega$ monads in the companion paper. 
\vspace{2mm}
\par\noindent
MSC: 05C65 81P99 81Q99

\vfill\eject

{\color{blue}\tableofcontents}

\vfill\eject

\renewcommand{\sectionmark}[1]{\markright{\thesection\ ~~#1}}

\section{\textcolor{blue}{\sffamily Introduction}}\label{sec:intro}

Graphs and hypergraphs are usefully employed to model binary or higher-order relations of any kind of
objects \ccite{Berge:1973gth,Ouvrard:2020hir}. 
Graphs consist of a set of vertices and one of edges connecting pairs of distinct vertices.
Hypergraphs are a broad generalization of graphs featuring vertices and hyperedges joining
any number of vertices rather than just two. Graphs and hypergraphs can be enriched with edge and hyperedge
weight data, yielding multigraphs and weighted hypergraphs. 

Graphs and hypergraphs find many applications
in the analysis of the data sets occurring in many scientific disciplines.
In quantum information and computation theory, they can be utilised to address a number
of fundamental issues as reviewed in the next subsection.


\subsection{\textcolor{blue}{\sffamily Graph and hypergraph states}}\label{subsec:ghgstrev}

Graph states are a special kind of multipartite states  encoded by graphs. 
They were originally introduced by Schlingemann and Hein, Eisert and Briegel 
in refs. \!\!\ccite{Schlingemann:2003cag,Hein:2003meg} over two decades ago.
Since then, they have been shown to be relevant in quantum error correction \ccite{Schlingemann:2000ecg,Bausch:2019tpn},
measurement-based quantum computation \ccite{Raussendorf:2001oqc},
quantum secret sharing \ccite{Markham:2008gss}
and analysis of Bell non locality \ccite{Scarani:2005ncs,Guehne:2004big,Baccari:2020bst}
and entanglement \ccite{Kruszynska:2006epp,Toth:2005dsf,Jungnitsch:2011ewg}. 
Ref. \!\!\ccite{Hein:2006ega} provides a comprehensive introduction to the subject. 

Hypergraph states are a wider class of multipartite states encoded by hypergraphs. 
They were first considered by Qu {\it et al.} and Rossi {\it et al.} in refs. \!\!\ccite{Qu:2012eqs,Rossi:2012qhs}
as a broad extension of graph states.
As these latter, they enter foremostly in quantum error correction \ccite{Wagner:2017shs},
measurement-based quantum computation \ccite{Gachechiladze:2019mbh,Takeuchi:2018qcu},
state verification and self-testing \ccite{Morimae:2017vmq,Zhu:2018evh},
analysis of Bell non locality \ccite{Gachechiladze:2015evr}   
and study of entanglement \ccite{Guehne:2014nch,Ghio:2017med}.  
Background material on hypergraph states can be found in ref. \!\!\ccite{Gachechiladze:2019phd}.

Qudits are higher level computational units extending the scope of normal 2--level qubits. Compared to these latter,
qudits have a larger state space that can be used to process and store more information and allow in principle
for a diminution of circuit complexity together with an increase of algorithmic efficiency.
They were first investigated in some depth in refs. \!\!\ccite{Ashikhmin:2000nbc,Gheorghiu:2011qsg}. Ref.
\!\!\ccite{Wang:2020qhd} reviews these qudits and their application to higher dimensional quantum computing. 

Qudit graph states have been explored in a number of works, e.g. \!\!\ccite{Helwig:2013amq,Keet:2010qss}. 
The study of qudit hypergraph states was initiated in refs. \!\!\ccite{Steinhoff:2016:qhs,Xiong:2017qhp}. 
These states were considered for quantum error correction in ref. \!\!\ccite{Looi:2007ecq}.

In spite of the breadth of the variety of the graph and hypergraph states which have been extensively studied in the literature
reviewed above, such states share some common universal features. First they are stabilizer states, 
i.e. one-dimensional error correcting codes \ccite{Gottesman:1997scq,Garcia:2014gss}. 
Second, they are locally maximally entangleable states \ccite{Kruszynska:2008lem}. 

A fundamental issue in quantum information theory is the classification of multipartite states 
in entanglement local equivalence classes. The main classification schemes used, as is well--known,
are local unitary/local Clifford, separable and stochastic local operations with classical communications
equivalence. Graph and hypergraph states lend themselves particularly well to this kind of investigation
because of the richness of their mathematical properties. This matter has been investigated in 
refs. \ccite{Nest:2003lct,Nest:2004lce,Nest:2004uce,Bravyi:2005ghz} for graph states
and \ccite{Lyons:2014luh} for hypergraph states and has been considered also for qudit graph
states in \ccite{Hostens:2004qma,Bahramgiri:2006cnb}.


\subsection{\textcolor{blue}{\sffamily Plan of the endeavour}}\label{subsec:wsproject}

In the present endeavour, we introduce and study calibrated hypergraph states,
a broad generalization of weighted hypergraph states encoded by 
hypergraphs endowed with calibrations, a non trivial extension of weightings.
The guiding principle informing our approach is that a working theory
of hypergraph states should be formulated in a categorical framework for both 
hypergraphs and multi qudit states and, making reference to this, postulate conditions on 
hypergraph states strong enough to allow for the determination of their general structure.
It leads to a set--up general enough to cover most of the graph state constructions
surveyed in the previous subsection. Our method is comparable to that of
Ionicioiu and Spiller in ref. \!\!\ccite{Ionicioiu:2012egq} whom we are indebted to for inspiration.
However, unlike theirs, our theory cannot be properly characterized as axiomatic but rather as constructive. 

The work is naturally divided in two parts, henceforth referred to as I and II,
which we outline next. 

In I, we introduce graded $\varOmega$ monads, concrete Pro
categories isomorphic to the Pro category $\varOmega$ of finite von Neumann ordinals and endowed with an associative and unital
graded multiplication, and their morphisms, mappings of $\varOmega$ monads compatible with their monadic structure. 
We then prove that calibrated hypergraphs on one hand and multi qudit states on the other 
naturally form graded $\varOmega$ monads. In this way, we pave the way to the construction of the calibrated
hypergraph state map as a special morphism of these $\varOmega$ monads carried out in II. 

In II, relying on the graded $\varOmega$ monadic framework elaborated in I and focusing on 
qudits over a generic Galois ring, 
we explicitly construct a calibrated hypergraph state map as a special morphism of the
calibrated hypergraph and multi qudit state $\varOmega$ monads.
We show furthermore that calibrated hypergraph states are locally maximally entangleable
stabilizer states, elucidate their relationship to weighted hypergraph states, prove that
they reduce to the weighted ones in the customary qubit
case and demonstrate through examples that this is no longer the case for higher qudits. 
We finally illustrate a number of special technical results and applications.  

The present paper covers Part I. Its contents are outlined in more detail in subsect. \cref{subsec:plan} below.




\subsection{\textcolor{blue}{\sffamily The graded monadic framework of hypergraph states}}\label{subsec:plan}

In this paper, we elaborate a categorical framework for the analysis and computation 
of qudit hypergraph states. 
Our work has a somewhat abstract outlook and is rather technical. To make its reading more straightforward,
in this subsection we outline the main ideas underlying its development. 
The discussion will be informal. Precise definitions of the concepts presented and rigorous statements
of the properties of the set--up constructed are to be found in the main body of the paper. 

The basic principle on which the theory we propose is based can be formulated as follows.
Hypergraphs and their data on one hand and multipartite quantum states on the other are naturally 
described by Pro categories isomorphic to the Pro category of finite von Neumann ordinals and
equipped with a graded monadic multiplicative structure. In subsect. \cref{subsec:design} below, 
we shall justify this principle by showing how it is implemented in and analysing
its implications for hypergraph state theory. Refs. \!\!\ccite{MacLane:1978cwm,Fong:2019act}
and app. \cref{app:cat} provide background on category theory.

Pro and the related Prop categories
were introduced long ago having in mind applications to algebraic topology
\,\ccite{Maclane:1965cta,Boardman:1968hes,May:1972ils,Markl:2002otp}, but with time
through the development of operads they found usage also in universal algebra and graph theory. 
A Pro category is a strict monoidal category all of whose objects are monoidal powers
of a single generating object. 
The category of finite von Neumann ordinals $\varOmega$ is the prototypical example of Pro category.
It is a concrete category featuring the ordinal sets $[l]$ as objects 
and the ordinal functions $f:[l]\rightarrow[m]$ as morphisms. Its
monoidal product formalizes the set theoretic operation of disjoint union of ordinal sets and 
functions through the usual ordinal sum.



The basic building blocks of hypergraphs are vertices. The fundamental components of multipartite 
systems are dits. The ordinal Pro category $\varOmega$ provides a
combinatorial model which describes such units avoiding the large redundancy
associated with the different ways in which they are concretely realized. Each set of $l$ units is identified with
the ordinal set $[l]$ viewed as a standard unit set. The functions mapping a set of $l$ units into one
of $m$ units are accordingly identified with the ordinal functions $f:[l]\rightarrow[m]$ 
construed as standard functions of standard unit sets. The operation of disjoint union of sets and functions of units
are then described by the monoidal products of the ordinal sets and functions which represent them.



The categories considered in this paper, such as the various hypergraph categories
and the multi dit mode categories,
are instances of $\varOmega$ categories, that is categories `modelled' on the ordinal Pro category $\varOmega$.
An $\varOmega$ category is a Pro category $D\varOmega$ to\-gether with a strict monoidal isofunctor
$D:\varOmega\rightarrow D\varOmega$. 
From the perspective of category theory, therefore, an $\varOmega$ category $D\varOmega$ is indistinguishable from
$\varOmega$. However, when $D\varOmega$ is a concrete category, as is always the case 
in the applications of our theory, the explicit incarnations of the objects $D[l]$ 
and morphisms $Df:D[l]\rightarrow D[m]$ of $D\varOmega$
as sets and functions do matter.

A graded $\varOmega$ monad is a concrete $\varOmega$ category $D\varOmega$ equipped with an
associative and unital graded monadic multiplication. The monoidal and monadic multiplicative structures of $D\varOmega$
are compatible but distinct and should not be confused. The former acts on the objects and morphisms
of $D\varOmega$, the latter on the elements forming the objects of $D\varOmega$. Indeed, monoidal and monadic 
multiplication stand in a relationship analogous to that occurring between tensor multiplication of
Hilbert spaces and vectors belonging to those spaces. Graded $\varOmega$ monad morphisms can be defined:
they are suitable maps of graded $\varOmega$ monads compatible with their monadic structures.


Graded $\varOmega$ monads and monad morphisms form a category $\ul{\rm GM}_\varOmega$.
This category bears a close relationship to the category $\ul{\rm GM}$
of graded monads. Graded monads were introduced in category theory relatively recently
\ccite{Mellies:2012mea,Katsumata:2014pms,Mellies:2017pcm,Fujii:2019gim}. They have found applications
in computer science, bounded linear logic and probability theory. In our work, another application has been found.

\begin{theor} \label{theor:ocjs2gm}
For any small subcategory $\clS$ of the graded $\varOmega$ monad category $\ul{\rm GM}_\varOmega$, 
there is a functor $\msH_{\hfpt\clS}:\clS\rightarrow\ul{\rm GM}$ injective on objects and morphisms
mapping $\clS$ into the graded monad category $\ul{\rm GM}$.
\end{theor}

\noindent
In this way, graded $\varOmega$ monads and monad morphisms can be identified
with genuine graded monads and monad morphisms. This is what justifies their name. 

The finite ordinal category $\varOmega$, $\varOmega$ categories and graded $\varOmega$ monads
are studied in detail in sect. \cref{sec:omegactg}. The results presented there
provide the formal categorical underpinning for the analysis of later sections. However, we do not insist
on the abstract categorical aspects of our theory beyond what is strictly necessary for its quantum theoretic
applications.

We outline next the main graded $\varOmega$ monads entering our theory.
The design of these categories, however arbitrary it may seem at a first reading, is justified ultimately
by the function they fulfill in the qudit hypergraph state construction and the eventual insight 
in hypergraph state theory which accrues to this latter.


Hypergraphs are collections of hyperedges and hyperedges are non empty sets of vertices. 
Hypergraphs are organized in an appropriate graded $\varOmega$ monad,
the hypergraph $\varOmega$ monad $G\varOmega$.
For each ordinal $[l]$, $G[l]$ is the set of all hypergraphs over the vertex set $[l]$;  
in similar fashion, for each ordinal function $f:[l]\rightarrow[m]$, $Gf:G[l]\rightarrow G[m]$ is the function
of hypergraph sets induced by $f$ as a vertex set function. The monadic multiplicative structure of
$G\varOmega$ is shaped in a manner roughly patterned after the disjoint union of hypergraphs as collections of hyperedges.

Hyperedge calibration data are a generalization of hyperedge weights, whose nature we shall make more
precise later. The construction of calibrated hypergraph states requires endowing the hyperedges
of the underlying hypergraphs with calibration data. The augmented hypergraphs obtained in this manner are called
calibrated. Calibrated hypergraphs are described by another graded $\varOmega$ monad, the calibrated
hypergraph $\varOmega$ monad $G_C\varOmega$. The $\varOmega$ monadic structure of $G_C\varOmega$ extends 
compatibly that of $G\varOmega$ and therefore $G_C\varOmega$ can be regarded as an enhancement of
$G\varOmega$. $G_C\varOmega$ has a richer content than $G\varOmega$, but it has
otherwise a very similar layout.


Hypergraphs augmented with ordinary hyperedge weights,
or weighted hypergraphs, organize in yet another graded $\varOmega$ monad, the weighted hypergraph $\varOmega$ monad 
$G_W\varOmega$. Unlike the $\varOmega$ monads $G\varOmega$, $G_C\varOmega$ introduced above,
which are essential to our analysis, $G_W\varOmega$ serves mainly the purpose of comparing our approach to hypergraphs
states to others existing in the literature. The relationship between $G_C\varOmega$ and $G_W\varOmega$
is encoded formally by a special morphism relating them that can be explicitly described. 

The hypergraph $\varOmega$ monad $G\varOmega$ and its calibrated and weighted enhancements $G_C\varOmega$ and 
$G_W\varOmega$ are studied in depth in sect. \cref{sec:graphctg}.

In classical theory, the distinct configurations of a single cdit are indexed by some finite set $\msR$. As long as one
aims only to a mere cataloguing of the cdit's configurations, $\msR$ can be assumed to be just a collection of labels.
However, the theory of computation invariably requires $\msR$ to be endowed with some kind of
algebraic structure depending on the approach followed and the goals pursued. At the very least, $\msR$ should
be a commutative monoid.
In quantum theory, the independent states of a single qudit are indexed by the same monoid $\msR$
as its cdit counterpart by the principle of base encoding.
The independent states of a single qudit therefore span a state Hilbert space $\scH_1$ of dimension $|\msR|$. 


Multi cdit configurations are described by a specific graded $\varOmega$ monad,
the multi cdit configuration $\varOmega$ monad $E\varOmega$. For each ordinal
$[l]$, $E[l]$ is the set of all configurations of the cdit set $[l]$; 
accordingly, for each ordinal function $f:[l]\rightarrow[m]$, $Ef:E[l]\rightarrow E[m]$ \linebreak is the
function of configuration sets induced by $f$ as a function of cdit sets, a map intuitively
describing the configuration `expansion/compression' implemented by $f$. The monadic multiplicative
structure of $E\varOmega$ is defined as concatenation of configurations
regarded as strings of $\msR$ labels.



Multi qudit states are organized in a graded $\varOmega$ monad that may be regarded
as the basis encoding of $E\varOmega$, the multi qudit state $\varOmega$ monad $\scH_E\varOmega$.
The design of $\scH_E\varOmega$ mirrors therefore that of $E\varOmega$ with multi qudit state Hilbert spaces 
replacing multi cdit configuration sets and linear operators between those spaces
taking the place of the functions between those sets. 
Tensor multiplication of state vectors further substitutes concatenation of configuration strings. 

The multi cdit configuration and multi qudit state $\varOmega$ monads $E\varOmega$ are 
$\scH_E\varOmega$ are defined and studied in detail in sect. \cref{sec:grsttsctg}.


\subsection{\textcolor{blue}{\sffamily Preview of the calibrated hypergraph state construction}}\label{subsec:design}

The calibrated hypergraph state construction is at the same time the final goal and the
ultimate justification of the elaboration of the categorical set--up delineated in subsect. \cref{subsec:plan}. 
In this subsection, we preview the construction of the calibrated hypergraph state map carried out in 
part II of our endeavour and hopefully convince the reader that the $\varOmega$ monadic framework 
is indeed an appropriate one for the analysis of hypergraph states, as claimed. 

Before proceeding to showing how the calibrated hypergraph state construction can be formulated in an
$\varOmega$ monadic framework, it is necessary to specify the notation used.
A calibrated hypergraph formed by a hypergraph $H$ and a hyperedge calibration data set $\varrho$ over $H$
is indicated by a pair $(H,\varrho)$. The monadic multiplication of calibrated hypergraphs is expressed by the symbol
$\smallsmile$ to remind its analogy to disjoint union.
The calibrated hypergraph monadic unit is denoted by $(O,\varepsilon)$ to recall its being empty. A multi qudit state vector
is denoted by a Dirac ket $\ket{\varPsi}$. The monadic multiplication of state vectors is indicated by
the tensor multiplication symbol $\otimes$, in conformity with its mathematical meaning. The state vector monadic unit
is denoted by $\ket{0}$ signifying just the complex unit.



The calibrated hypergraph state construction associates with any calibrated hypergraph
$(H,\varrho)\in G_C[l]$ a hypergraph state $\ket{(H,\varrho)}\in\scH_E[l]$. 
It is covariant, since 
\begin{equation}
\label{i/whgsts16}
\scH_Ef\ket{(H,\varrho)}=\ket{G_Cf(H,\varrho)}
\end{equation}
for any ordinal function $f:[l]\rightarrow[m]$. It is further compatible with the monadic multiplication
in the sense that \hphantom{xxxxxxxxxxxx} 
\begin{equation}
\label{i/whgsts18}
\ket{(H,\varrho)\smallsmile(K,\varsigma)}=\ket{(H,\varrho)}\otimes\ket{(K,\varsigma)}
\end{equation}
for $(H,\varrho)\in G_C[l]$, $(K,\varsigma)\in G_C[m]$ and moreover satisfies 
\begin{equation}
\label{i/whgsts19}
\ket{(O,\varepsilon)}=\ket{0}.
\end{equation}
By these properties, the maps 
$\ket{-}:G_C[l]\rightarrow \scH_E[l]$ assigning the 
hypergraph state $\ket{(H,\varrho)}\in\scH_E[l]$ to each calibrated hypergraph $(H,\varrho)\in G_C[l]$
for the different values of $l$ specify a distinguished morphism of the graded $\varOmega$ monads 
$G_C\varOmega$, $\scH_E\varOmega$. 
An $\varOmega$ monad theoretic interpretation of the calibrated hypergraph state map $\ket{-}$ such as this,
although interesting in its own, hides however a far more relevant point:
relations \ceqref{i/whgsts16}--\ceqref{i/whgsts19} provide valuable information about
the nature of the hypergraph calibration data, which we have not specified yet, 
and essentially determine the map $\ket{-}$. 
This, ultimately, is the reason why our categorical set--up is not an end in itself 
but an indispensable reference grid for the analysis of hypergraph states.



As recalled earlier, the set $\msR$ labelling a cdit's configurations
must have some kind of algebraic structure in computation theory.
$\msR$ is often supposed to be a Galois field. However, this assumption
may sometimes be relaxed. Our construction of hypergraph states works also when $\msR$ is
more generally a Galois ring. Background on Galois ring theory can be found in refs.
\ccite{Wan:2011gfr,Bini:2002gfr,Kibler:2017gre} and app. A of II. 

A Galois ring $\msR$ is determined up to isomorphism by its characteristic $p^r$,
a prime power, and degree $d$. It is a finite ring with $q=p^{rd}$ elements altogether.
It further includes a prime subring $\msP$, a minimal Galois subring isomorphic to the ring $\bbZ_{p^r}$ of integers
modulo $p^r$. It finally features a trace function $\tr:\msR\rightarrow\msP$ with important surjectivity and
non singularity properties.

Since a Galois ring $\msR$ is finite, for any fixed $x\in\msR$ the exponents $u\in\bbN$ giving distinct powers
$x^u$ constitute a finite set $\msZ_x$ of exponents depending on $x$. $\msZ_x$ is in fact
a cyclic monoid. The monoids $\msZ_x$ for varying $x\in\msR$
can be assembled into the direct sum monoid $\msZ=\bigoplus_x\msZ_x$, called
the cyclicity monoid of $\msR$. The power of a ring element $x\in\msR$ to the exponent $u\in\msZ$ is defined to be 
$x^u=x^{u_x}$ if $u=(u_x)_{x\in\msR}$.

Sect. 2 of II reviews the properties of Galois rings most relevant in calibrated hypergraph state theory
and other topics such as the qudit Pauli group
and quantum Fourier transformation, which enter as essential elements in its elaboration. 

We illustrate next the calibrated hypergraph state construction in some detail. 
An exponent function of a given vertex set $X$ is any function $w:X\rightarrow\msZ$.
Since $\msZ$ is a monoid, the exponent functions of $X$ form a monoid, the exponent
monoid $\msZ^X$. A calibration of $X$ is any function $\varrho_X:\msZ^X\rightarrow\msP$.
A calibration $\varrho$  over a hypergraph $H\in G[l]$ is an assignment to each hyperedge $X\in H$ of a calibration
$\varrho_X$ of $X$.

The hypergraph state $\ket{(H,\varrho)}\in\scH_E[l]$ associated with a calibrated hypergraph
$(H,\varrho)\in G_C[l]$ is given by\hphantom{xxxxxxxxxx}
\begin{equation}
\label{i/whgsts10}
\ket{(H,\varrho)}=D_{(H,\varrho)}\ket{0_l},
\end{equation}
where $D_{(H,\varrho)}$ is a linear operator of $\scH_E[l]$ of the form 
\begin{equation}
\label{i/whgsts11}
D_{(H,\varrho)}=\mycom{{}_\sss}{{}_{x\in E[l]}}F_l{}^+\ket{x}\,\omega^{\sigma_{(H,\varrho)}(x)}\hfpt\bra{x}F_l.
\end{equation}
Here, $F_l$ is the quantum Fourier transform operator and $\omega=\exp(2\pi i/p^r)$. 
The vectors $\ket{x}$, $x\in E[l]$, constitute the qudit Hadamard basis of
$\scH_E[l]$ to which $\ket{0_l}$ belongs. The phase function
$\sigma_{(H,\varrho)}:E[l]\rightarrow\msP$  reads as
\begin{equation}
\label{i/whgsts9}
\sigma_{(H,\varrho)}(x)=\mycom{{}_\sss}{{}_{X\in H}}\mycom{{}_\sss}{{}_{w\in\msZ^X}}\varrho_X(w)
\tr\left(\mycom{{}_\ppp}{{}_{r\in X}}x_r{}^{w(r)}\right) 
\end{equation}
with $x\in E[l]$. 
The state $\ket{(H,\varrho)}$ is encoded by the hypergraph $(H,\varrho)$ through $\sigma_{(H,\varrho)}$. 

The definition of hypergraphs states we have submitted is justified ultimately 
by the following theorem which is the main result reached in II. 

\begin{theor} \label{theor:main}
The calibrated hypergraph state map defined by expressions \ceqref{i/whgsts10}--\ceqref{i/whgsts9}
satisfies properties \ceqref{i/whgsts16}--\ceqref{i/whgsts19}.
\end{theor}

\noindent
The hypergraph states obtained in this way have all the properties expected
from such states.

\begin{theor} \label{theor:basic}
The calibrated hypergraph states given by expressions \ceqref{i/whgsts10}--\ceqref{i/whgsts9}
are locally maximally entangleable stabilizer states \ccite{Kruszynska:2008lem}. 
\end{theor}

\noindent
The  stabilizer group operators $K_{(H,\varrho)}(a)$ of the hypergraph states $\ket{(H,\varrho)}$ 
and the associated hypergraph orthonormal bases $\ket{(H,\varrho),a}$ can also be obtained. 

Our procedure for generating calibrated hypergraph states is in its main lines
analogous to that yielding weighted hypergraph states appearing in earlier literature such as refs.
\!\!\ccite{Qu:2012eqs,Rossi:2012qhs}
or, in the qudit case, refs. \!\!\ccite{Steinhoff:2016:qhs,Xiong:2017qhp}. Couched in our $\varOmega$ monadic language
for the purpose of comparison, those authors associate with any weighted hypergraph $(H,\alpha)\in G_W\varOmega$
a weighted hypergraph state $\ket{(H,\alpha)}\in\scH_E[l]$ given by
an expression analogous to \ceqref{i/whgsts10}, \ceqref{i/whgsts11} but with $\sigma_{(H,\alpha)}$ of the form 
\begin{equation}
\label{i0/whgsts9}
\sigma_{(H,\alpha)}(x)=\mycom{{}_\sss}{{}_{X\in H}}\alpha_X\tr\left(\mycom{{}_\ppp}{{}_{r\in X}}x_r\right)\!,
\end{equation}
where $\alpha_X\in\msP$ is the weight factor of hyperedge $X$.
Albeit the weighted hypergraph state construction is sound, the hypergraph states it produces 
are not sufficiently general for all the properties \ceqref{i/whgsts16}--\ceqref{i/whgsts19} to be enjoyed.
The calibrated hypergraph states proposed in the present endeavour conversely
possess those properties by design. 

We quote next a further results which may support our proposal.

\begin{theor} \label{theor:calweicomp}
For each weighted hypergraph $(H,\alpha)\in G_W[l]$, there exists a calibrated hypergraph
$(H,\varrho)\in G_C[l]$ with the same underlying hypergraph $H\in G[l]$ with the property that
$\ket{(H,\alpha)}=\ket{(H,\varrho)}$.
\end{theor}

\noindent
So, the calibrated hypergraph approach  subsumes the weighted hypergraph one.
The relationship of the two approaches could be summarized in simple words by saying that while 
the phase functions of calibrated hyper graph states
exhibit all possible powers of the Galois qudit variables those of the weighted hypergraph states
do only the $0$-th and $1$-st ones. 
The critical question about whether genuinely new hypergraph state entanglement classes are generated
in this way will not be addressed in this work. 

For qubits, for which the relevant Galois ring $\msR$ is the simple
binary field $\bbF_2$, the calibrated hypergraph state construction does not yield
new states  beyond the weighted hypergraph ones already known.

\begin{theor} \label{theor:qudqubcomp}
Assume that $\msR=\bbF_2$ is the binary field. Then, for each calibrated hypergraph $(H,\varrho)\in G_C[l]$
there exists a weighted hypergraph $(K,\alpha)\in G_W[l]$ such that 
$\ket{(K,\alpha)}=\ket{(H,\varrho)}$ up to a sign.
\end{theor}

The calibrated hypergraph state construction and its properties are expounded in sect. 3 of II,
where the above results and more are obtained.

When the Galois ring $\msR$ is a field, the general expression of the calibrated hyper\-graph state map
can be recast into a fully polynomial form thanks to the existence of a distinguished subset of
polynomials taken from the polynomial ring $\msR[\sfx]$ computing the  powers of the field's elements to
the cyclicity monoid's exponents.
This and  other selected applications of the theory of calibrated hypergraph states we
have developed can be found in sect. 4 of II.


\subsection{\textcolor{blue}{\sffamily Scope and outlook of the work}}\label{subsec:scope}

With the present work, we have provided a more general definition of hypergraph states based on and justified by an underlying
graded monadic framework for hypergraphs and quantum multipartite states. 
The approach we have followed 
is motivated by the wish of 
{\it

\begin{enumerate}[label=\alph*)] 

\item {\rm finding original applications of the methods of category theory to quantum information theory,}

\item {\rm achieving a more complete and precise theoretical understanding of hypergraphs states,}

\item {\rm providing an alternative formulation of the problems of hypergraph state theory which might help to find their solution and }

\item {\rm devising new computational techniques for hypergraph states.}
  
\end{enumerate}
}
\noindent
Only time will say whether these goals will be fully achieved. More modestly, one might concentrate on the first two of them. 


The pending fundamental question is whether the hypergraph states generated in this way
represent genuinely new entanglement classes in the standard local classification schemes --
local unitary/local Clifford, separable and stochastic local operations with classical communications --
or else fall into already known entanglement classes. This matter is non trivial and cannot be solved
by us in the limited space of a single work. We therefore leave it for future work.

\vfill\eject

\noindent
\markright{\textcolor{blue}{\sffamily Conventions}} 


\noindent
\textcolor{blue}{\sffamily Conventions.}
In this paper, we adopt the following notational conventions.

\begin{enumerate}[leftmargin=*]

\item \label{it:conv1} We indicate by $|A|$ the cardinality of a finite set $A$.

\item \label{it:conv2} For any set $A$, we let $e_A:\emptyset\rightarrow A$ be 
the empty function with range $A$. 


\item \label{it:conv3} Complying with the most widely used convention computer science,
we denote by $\bbN$ the set of all non negative integers. Hence, $0\in\bbN$.
We denote  by $\bbN_+$ the set of all strictly positive integers. 

\item \label{it:conv4} For $l\in\bbN$, we let $[l]=\emptyset$ if $l=0$ and $[l]=\{0,\dots,l-1\}$ if $l>0$
be the standard finite von Neumann ordinals.

\item \label{it:conv5} If $l\in\bbN$ and $A\subset\bbN$, we let $A+l=\{i+l|i\in A\}$. Accordingly, if $M\subset P\bbN$
(the power set of $\bbN$) is a collection of subsets of $\bbN$, we let $M+l=\{A+l|A\in M\}$.

\item \label{it:conv6} An indexed finite set is a finite set $A$ together with a bijection $a:[|A|]\rightarrow A$.
An indexed finite set $(A,a)$ is so described by the $|A|$--tuple $(a(0),\ldots,a(|A|-1))$.
As a rule, 
we shall write the indexed set as $A$ and the associated $|A|$--tuple
as $(a_0,\ldots,a_{|A|-1})$ for simplicity.
When $A$ is a totally ordered set, e.g. a set of integers,
then it is tacitly assumed that the indexing employed
is the one such that $a_r<a_s$ for $r,s\in[|A|]$ with $r<s$, unless otherwise stated. 

\item \label{it:conv7} If $A$, $B$ are finite sets and $f:A\rightarrow B$ is a function and furthermore $A$ is indexed
as $(a_0,\ldots,a_{|A|-1})$, then $f$ is specified by the $|A|$--tuple $(f(a_0),\ldots,f(a_{|A|-1})_{A}$
of its values. 

\end{enumerate}

\vfill\eject

\vfill\eject

\renewcommand{\sectionmark}[1]{\markright{\thesection\ ~~#1}}

\section{\textcolor{blue}{\sffamily The ordinal category $\varOmega$,  $\varOmega$ categories
    and graded $\varOmega$ monads}}\label{sec:omegactg}

The $\varOmega$ monadic framework is the universal categorical paradigm underlying our treatment of
calibrated hypergraph and hypergraph states. In this section, we introduce it and study its main properties. 

The theory of the finite ordinal category $\varOmega$ is reviewed in subsect. \cref{subsec:finordcat}.
In subsect. \cref{subsec:omegacat}, we present $\varOmega$ categories, the kind of categories
relevant in our analysis, and investigate their properties. Graded $\varOmega$ monads 
are defined and studied in subsect. \cref{subsec:dcatjoint}. Some useful background material of category theory
can be found in app. \cref{app:cat}. \ccite{MacLane:1978cwm,Fong:2019act} are standard references on the subject. 
The exemplification of the theory worked out here will be provided in later sections.


\subsection{\textcolor{blue}{\sffamily The finite ordinal category $\varOmega$}}\label{subsec:finordcat}

The category $\ul{\rm FinOrd}$ of finite ordinals and functions plays a central role in our analysis. 
$\ul{\rm FinOrd}$ is a full subcategory and the skeleton of the category $\ul{\rm FinSet}$
of finite sets and functions thereof.
We shall conveniently denote $\ul{\rm FinOrd}$ and $\ul{\rm FinSet}$ as $\varOmega$ and $\varSigma$,
respectively. 

\begin{prop} \label{prop:omegacat}
There exists a strict monoidal category $\varOmega$ with the following layout.

\begin{enumerate}

\item The objects set $\Obj_\varOmega$ of $\varOmega$ is the collection of the finite von Neumann
ordinals $[l]$ with $l\in\bbN$.

\item For any $l,m\in\bbN$, the morphism set $\Hom_\varOmega([l],[m])$ of $\varOmega$ consists of all functions
$f:[l]\rightarrow[m]$.
  
\item For any $l,m,n\in\bbN$, the composite $g\circ f\in\Hom_\varOmega([l],[n])$ of
  $f\in\Hom_\varOmega([l],[m])$, $g\in\Hom_\varOmega([m],[n])$ is the usual composite of functions. 


\item For any $l\in\bbN$, the identity $\id_{[l]}\in\Hom_\varOmega([l],[l])$ of $[l]\in\Obj_\varOmega$
is the usual identity function of the set $[l]$.

\item For any $l,m\in\bbN$, the monoidal product $[l]\smallsmile[m]\in\Obj_\varOmega$ of 
$[l]$, $[m]\in\Obj_\varOmega$ is 
\begin{equation}
\label{gamma1}
[l]\smallsmile[m]=[l+m].
\end{equation}

\item For any $l,m,p,q\in\bbN$, the monoidal product $f\smallsmile g\in\Hom_\varOmega([l]\smallsmile[m],[p]\smallsmile[q])$ 
of $f\in\Hom_\varOmega([l],[p])$, $g\in\Hom_\varOmega([m],[q])$ is given by 
\begin{equation}
\label{gamma2}
f\smallsmile g(i)=\Bigg\{
\begin{array}{ll}
f(i)&\text{if $i\in[l]$},\\
g(i-l)+p&\text{if $i\in[m]+l$}.
\end{array}
\end{equation}

\item The monoidal unit is $[0]\in\Obj_\varOmega$.
  
\end{enumerate}
\end{prop}

\noindent 
We notice that above one has
\vspace{-.5mm}
\begin{align}
\label{}
&\Hom_\varOmega([l],[0])=\emptyset \qquad\qquad \text{if $l>0$},
\\
\label{}
&\Hom_\varOmega([0],[m])=\{e_{[m]}\}  
\end{align}
\vspace{-.5mm}
for $l,m\in\bbN$. 


\begin{proof}
The composition operation $\circ$ and the identity assigning map $\id$ of $\varOmega$ are
the same as those of the finite set category $\varSigma$ and thus satisfy 
the associativity and unitality relations \eqref{catdef1}, \eqref{catdef2}, rendering $\varOmega$
a category.
Using relations \ceqref{gamma1}, \ceqref{gamma2}, it is also readily verified that the monoidal 
multiplication $\smallsmile$ and unit $[0]$ of $\varOmega$ fulfil the strict monoidality requisites
\ceqref{moncat0} --\ceqref{moncat4}, in this way making the category $\varOmega$ strict monoidal. 
\end{proof}

\noindent
$\varOmega$ is called the finite von Neumann ordinal or more simply ordinal category.


A distinguished property of the category $\varOmega$ is that for all $l\in\bbN$ the object $[l]$
is the $l$--th monoidal power of the object $[1]$: 
\begin{equation}
\label{omegapro}
[l]=[1]\smallsmile\cdots\smallsmile[1] \qquad \text{($l$ terms).}
\end{equation}
$\varOmega$ is therefore an instance of Pro category
\ccite{Maclane:1965cta,Boardman:1968hes,May:1972ils,Markl:2002otp},
a distinguished kind of strict monoidal category reviewed in app. \cref{subsec:procat}.

\begin{prop}\label{prop:omegapro}
$\varOmega$ is a Pro category with generating object $[1]$.
\end{prop}

\begin{proof}
By virtue of prop. \cref{prop:omegacat}, $\varOmega$ is a strict monoidal category with monoidal product $\smallsmile$. 
By relation \ceqref{omegapro}, $\varOmega$ satisfies moreover the defining requisite
of a Pro category, eq. \ceqref{procat1}, having $[1]$ as its generating object. 
\end{proof}

\begin{rem} \label{rem:omegaswap}
{\rm The finite ordinal category $\varOmega$ turns out to be in fact a symmetric strict monoidal category
(cf. def. \cref{def:symmon}). The swap maps of the underlying symmetric structure are the
morphisms $s_{[l],[m]}\in\Hom_\varOmega([l]\smallsmile[m],[m]\smallsmile[l])$, $l,m\in\bbN$, given by 
\begin{equation}
\label{symomg}
s_{[l],[m]}(i)=\Bigg\{
\begin{array}{ll}
i+m&\text{if $i\in[l]$},\\
i-l&\text{if $i\in[m]+l$}.
\end{array}
\end{equation}
%
%
The swap map relations \ceqref{symmon1}--\ceqref{symmon4} are straightforwardly shown to be fulfilled. 
Being both a symmetric strict monoidal and a Pro category renders $\varOmega$ a Prop category
(cf. app. \cref{subsec:procat}). This property of $\varOmega$
however plays no role in the subsequent analysis and is mentioned
here only for the sake of completeness. 
}
\end{rem}


\subsection{\textcolor{blue}{\sffamily $\varOmega$ categories}}\label{subsec:omegacat}

The finite ordinal category $\varOmega$ and its variants generated by narrowing its morphism sets are 
relevant in various branches of mathematics.
The reason why $\varOmega$ is key in our work is the following: 
$\varOmega$  constitutes a universal combinatorial model
for a number of concrete Pro categories employed in the theory of calibrated hypergraph states
developed by us, the so called $\varOmega$ categories. 
In this subsection, we introduce $\varOmega$ categories and analyze in some detail their general properties. 

\begin{defi} \label{def:omgcat}
A Pro category $\clC$ is an $\varOmega$ category if $\clC$ is isomorphic to $\varOmega$, that is
such that there is a strict monoidal isofunctor $D:\varOmega\rightarrow\clC$ (cf. def. \cref{def:monfun}).
\end{defi}

\noindent
$\clC$ can therefore be conveniently identified with the image category $D\varOmega$,
whose object collection $\Obj_{D\varOmega}$ consists of the images $D[l]$ of the objects $[l]\in\Obj_\varOmega$
and whose morphism sets $\Hom_{D\varOmega}(D[l],D[m])$ are constituted by the images $Df$
of the morphisms $f\in\Hom_\varOmega([l],[m])$.
$D$ furnishes so a one--to-one parametrization of the objects and morphisms $\clC$ by
those of $\varOmega$ and for this reason is called the stalk isofunctor of $\clC$.


The next proposition furnishes a general way of constructing $\varOmega$ categories 
often employed in practice. In essence, it delineates a method of generating strict monoidal
isofunctors $D$ having the ordinal category $\varOmega$ as source category.

\begin{prop} \label{prop:domegacat} 
Let

\begin{enumerate}[label=\alph*)] 

\item \label{O1} a non empty class $O$, 

\item \label{O2} a collection of non empty sets $H(l,m)$, $l,m\in\bbN$, such that for $l,m,p,q\in\bbN$
with $(l,p)\neq(m,q)$ $H(l,p)\cap H(m,q)=\emptyset$,  

\item \label{O3} an injective map $D:\Obj_\varOmega\rightarrow O$ and 

\item \label{O4} an injective map $D:\Hom_\varOmega([l],[m])\rightarrow H(l,m)$ for every
$l,m\in\bbN$
  
\end{enumerate}

\noindent
be given. Then, there exists a strict monoidal category $D\varOmega$ with the following layout.

\begin{enumerate} 

\item\label{d1}
The set $\Obj_{D\varOmega}=\{D[l]\hfpt|\hfpt l\in\bbN\}$ is the object set of $D\varOmega$.

\item\label{d2} The sets $\Hom_{D\varOmega}(D[l],D[m])=\{Df\hfpt|\hfpt f\in\Hom_\varOmega([l],[m])\}$ for varying
$l,m\in\bbN$ are the morphism sets of $D\varOmega$.

\item\label{d3} For every $l,m,n\in\bbN$ and  $f\in\Hom_\varOmega([l],[m])$, $g\in\Hom_\varOmega([m],[n])$, the composite
 $Dg\circ Df\in\Hom_{D\varOmega}(D[l],D[n])$ of $Df$, $Dg$ is given by 
\begin{equation}
\label{domg1}
Dg\circ Df=D(g\circ f).
\end{equation}

\item\label{d4} For every $l\in\bbN$, the identity $\id_{D[l]}\in\Hom_{D\varOmega}(D[l],D[l])$ of $D[l]$ reads as
\begin{equation}
\label{domg2}
\id_{D[l]}=D\id_{[l]}.
\end{equation}
 
\item\label{d5} For any $l,m\in\bbN$, the monoidal product $D[l]\smallsmile D[m]\in\Obj_{D\varOmega}$ of $D[l]$, $D[m]$ is 
\begin{equation}
\label{domg3}
D[l]\smallsmile D[m]=D([l]\smallsmile[m]).
\end{equation}
  
\item\label{d6} For every $l,m,p,q\in\bbN$ and $f\in\Hom_\varOmega([l],[p])$, $g\in\Hom_\varOmega([m],[q])$,
the monoidal product $Df\smallsmile Dg\in\Hom_{D\varOmega}(D[l]\smallsmile D[m],D[p]\smallsmile D[q])$
of $Df$, $Dg$  reads as 
\begin{equation}
\label{domg4}
Df\smallsmile Dg=D(f\smallsmile g).
\end{equation}

\item\label{d7} The monoidal unit of $D\varOmega$ is $D[0]\in\Obj_{D\varOmega}$.  

\end{enumerate}

\noindent
$D\varOmega$ is further a Pro category with generating object $D[1]$. Moreover, the map $D$ gives rise to 
a strict monoidal isofunctor $D:\varOmega\rightarrow D\varOmega$. $D\varOmega$ is so an $\varOmega$ category
and $D$ is its stalk isofunctor. 
\end{prop}



\begin{proof}
To begin with, we notice that the injectivity of the mapping $D$ ensures that both $\Obj_{D\varOmega}$
and the $\Hom_{D\varOmega}(D[l],D[m])$ are genuine sets and not multisets. It further guarantees that
the composition law and identity assigning map as well as the monoidal multiplication on objects
and morphisms of $D\varOmega$ are unambiguously defined. 
By prop. \cref{prop:omegacat}, $\varOmega$ is a strict monoidal category.
By \ceqref{domg1}, \ceqref{domg2}, the composition operation $\circ$ and the identity assigning map $\id$
of $D\varOmega$ obey the associativity and unitality relations \eqref{catdef1}, \eqref{catdef2} as
a consequence of those of $\varOmega$ doing so, making $D\varOmega$ a category. 
Using \ceqref{domg3}--\ceqref{domg4}, it is also straightforward to verify that the 
monoidal product $\smallsmile$ and unit $D[0]$ of $D\varOmega$ obey the strict monoidality 
relations \ceqref{moncat0}--\ceqref{moncat4} as a result of those of $\varOmega$ doing so, 
rendering the category $D\varOmega$ strict monoidal. From \ceqref{omegapro} and \ceqref{domg3}, 
$D\varOmega$ satisfies moreover the defining condition of a Pro category, viz  eq. \ceqref{procat1},
having $D[1]$ as its generating object.
Finally, the way the object and morphism sets of $D\varOmega$ are defined
and relations \ceqref{domg1}--\ceqref{domg4} manifestly imply that 
the mapping $D$ gives rise to a strict monoidal functor $D:\varOmega\rightarrow D\varOmega$.
As $D$ is bijective on objects and morphism, $D$ is an isofunctor. It follows that $D\varOmega$ is an $\varOmega$ category
with stalk isofunctor $D$. 
\end{proof}

\begin{cor} \label{rem:dtrick}
Let $D:\varOmega\rightarrow\clT$ be a functor injective on objects and morphisms, where $\varOmega$ is viewed as a 
plain category with no monoidal structure and $\clT$ is some target category.
Then, there exists an $\varOmega$ category $D\varOmega$ with the design shown in items
\cref{d1}--\cref{d7} of prop. \cref{prop:domegacat} having the functor $D$  
with the target category restricted from $\clT$ to $D\varOmega$ as its stalk isofunctor.
Furthermore, $D\varOmega$ is a non monoidal subcategory of $\clT$.
\end{cor}

\begin{proof}
Prop. \cref{prop:domegacat}, taken in the special case  where $O=\Obj_\clT$ is the object class, the
$H(l,m)=\Hom_\clT(D[l],D[m])$, $l,m\in\bbN$, are the morphism sets and the functor $D$ is the layout map,
guarantees the existence of an $\varOmega$ category $D\varOmega$ described by items
\cref{d1}--\cref{d7} and having $D$ as its stalk isofunctor. The functoriality
of $D$ ensures further that $D\varOmega$, as an ordinary  non monoidal category, is a subcategory of $\clT$. 
\end{proof}


\begin{rem} \label{rem:dcatiso} {\rm
$\varOmega$ categories are all isomorphic. Indeed, given two $\varOmega$ categories $D\varOmega$, $E\varOmega$, 
the composite $\varPhi_{DE}:E\circ D^{-1}:D\varOmega\rightarrow E\varOmega$ of their stalk isofunctors is evidently
a strict monoidal isofunctor. Explicitly, $\varPhi_{DE}$ acts as}
{\allowdisplaybreaks
\begin{align}
\label{diso1}
\varPhi_{DE}D[l]&=E[l], 
\\
\label{diso2}
\varPhi_{DE}Df&=Ef,
\end{align}
}
\!\!{\rm where $l\in\bbN$ and $f\in\Hom_\varOmega([l],[m])$ for some $l,m\in\bbN$.
In spite of this, multiple incarnations of $\varOmega$ categories enter meaningfully
and non trivially the constructions elaborated in the present work.}
\end{rem}

\begin{rem} \label{rem:domegaswap}
{\rm As $\varOmega$ itself, every $\varOmega$ category $D\varOmega$ is symmetric strict monoidal. 
Indeed, via the stalk isofunctor $D$, it is possible to 
push--forward the swap map set $s_{[l],[m]}$, $l,m\in\bbN$, of $\varOmega$ defined in \ceqref{symomg}
into $D\varOmega$ yielding a swap map set $s_{D[l],D[m]}$ of $D\varOmega$
consisting of the morphisms $s_{D[l],D[m]}\in\Hom_{D\varOmega}(D[l]\smallsmile D[m],D[m]\smallsmile D[l])$ given by 
\begin{equation}
\label{symdomg}
s_{D[l],D[m]}=Ds_{[l],[m]}. \vphantom{\bigg]} 
\end{equation}
%
%
From \ceqref{symdomg} and the strict monoidal functoriality relations
\ceqref{monfun1}--\ceqref{monfun3} obeyed by $D$, the $s_{D[l],D[m]}$ satisfy the
swap map relations \ceqref{symmon1}--\ceqref{symmon4} as a result of the $s_{[l],[m]}$
doing so. The stalk isofunctor $D$ is in this way symmetric, 
as the corresponding condition \ceqref{symmon5} is satisfied 
by construction by virtue of \ceqref{symdomg}.
Being both a symmetric strict monoidal and a Pro category renders 
every $\varOmega$ category $D\varOmega$ a Prop category as $\varOmega$ itself. 
Again, 
this property of $\varOmega$ categories plays no role in the subsequent analysis
and is mentioned here only for the sake of completeness.
}
\end{rem}


\subsection{\textcolor{blue}{\sffamily Graded $\varOmega$ monads}}\label{subsec:dcatjoint}

All $\varOmega$ categories dealt with in the present work are graded $\varOmega$ monads. In this subsection,
we introduce this further notion and show how it can be framed in the broader theory of graded monads.

Recall that a concrete category is a category whose objects and morphisms are sets and functions between sets.

\begin{defi} \label{def:jointstrct}
a graded $\varOmega$ monad is a concrete $\varOmega$ category $D\varOmega$ 
(cf. subsect. \cref{subsec:omegacat}, def. \cref{def:omgcat}) equipped with:

\begin{enumerate}[label=\alph*)]

\item for every $l,m\in\bbN$ a map $\smallsmile:D[l]\times D[m]\rightarrow D([l]\smallsmile[m])$,
called monadic multiplication, associating the joint product $\alpha\smallsmile\beta\in D([l]\smallsmile[m])$
with every pair $\alpha\in D[l]$, $\beta\in D[m]$;

\item a distinguished element $\iota\in D[0]$, called monadic unit. 

\end{enumerate}

\noindent
The following properties must further hold. 

\begin{enumerate}

\item For $l,m,n\in\bbN$, $\lambda\in D[l]$, $\mu\in D[m]$, $\nu\in D[n]$
\begin{equation}
\label{joint1}
\lambda\smallsmile(\mu\smallsmile\nu)=(\lambda\smallsmile\mu)\smallsmile\nu.
\end{equation}

\item For every $l\in\bbN$, $\lambda\in D[l]$, \hphantom{xxxxxxxxx}
\begin{equation}
\label{joint2}
\lambda\smallsmile\iota=\iota\smallsmile\lambda=\lambda.
\end{equation}

\item For all $l,m,p,q\in\bbN$, $f\in\Hom_\varOmega([l],[p])$, $g\in\Hom_\varOmega([m],[q])$, $\lambda\in D[l]$, $\mu\in D[m]$
\begin{equation}
\label{joint3}
Df\smallsmile Dg(\lambda\smallsmile \mu)=Df(\lambda)\smallsmile Dg(\mu). 
\end{equation}

\end{enumerate}

\end{defi} 

\noindent
Notice that we denote the monoidal and monadic multiplication of the category $D\varOmega$ with the same symbol
$\smallsmile$. This is not accidental: indeed, they stand in a relationship analogous to that occurring between
tensor multiplication of Hilbert spaces and vectors belonging to those spaces, both denoted by $\otimes$.
Said parenthetically, this is ultimately the reason why the multiple monadic multiplicative 
structures dealt with in this paper eventually have a Hilbertian realization.


There exists also a notion of morphism of graded $\varOmega$ monads
provided by the next definition.

\begin{defi} \label{def:catjoint}
A morphism $\sfm:D\varOmega\rightarrow E\varOmega$ of graded $\varOmega$ monads
features for each $l\in\bbN$ a function $\sfm:D[l]\rightarrow E[l]$ with the properties which follow. 

\begin{enumerate}

\item For $l,m\in\bbN$ and $f\in\Hom_\varOmega([l],[m])$, one has 
\begin{equation}
\label{catjoint1}
\sfm\circ Df=Ef\circ\sfm.
\end{equation}

\item For $l,m\in\bbN$ and $\lambda\in D[l]$, $\mu\in D[m]$, 
\begin{equation}
\label{catjoint2}
\sfm(\lambda\smallsmile\mu)=\sfm(\lambda)\smallsmile\sfm(\mu). 
\end{equation}

\item Finally, it holds that \hphantom{xxxxxxxxxxxx}
\begin{equation}
\label{catjoint3}
\sfm(\iota)=\iota.
\end{equation}

\end{enumerate}
\end{defi}

\noindent
In \ceqref{catjoint1}--\ceqref{catjoint2}, the symbol $\circ$, $\smallsmile$ and $\iota$ denote composition,
monadic multiplication and unit of the monads $D\varOmega$ and $E\varOmega$ respectively in the left and right hand
side. 


Graded $\varOmega$ monads and their morphisms constitute a category. 

\begin{prop} \label{prop:jomgcat}
There is a category $\ul{\rm GM}_\varOmega$ specified by the assignments below:

\begin{enumerate}[label=\alph*)]

\item the class $\Obj_{\ul{\rm GM}_\varOmega}$ of all graded $\varOmega$ monads as object class;

\item the set $\Hom_{\ul{\rm GM}_\varOmega}(D\varOmega,E\varOmega)$ of all morphisms
$\sfm:D\varOmega\rightarrow E\varOmega$ of graded $\varOmega$ monads
as morphism set of the objects $D\varOmega,E\varOmega\in\Obj_{\ul{\rm GM}_\varOmega}$;

\item for all objects $D\varOmega,E\varOmega,F\varOmega\in\Obj_{\ul{\rm GM}_\varOmega}$
and morphisms $\sfm\in\Hom_{\ul{\rm GM}_\varOmega}(D\varOmega,E\varOmega)$,
$\sfn\in\Hom_{\ul{\rm GM}_\varOmega}(E\varOmega,F\varOmega)$ of $\ul{\rm GM}_\varOmega$, the function collection
$\sfn\,\circ\,\sfm:D[l]\rightarrow F[l]$, \linebreak $l\in\bbN$, got by the composition
of the functions $\sfm:D[l]\rightarrow E[l]$, $\sfn:E[l]\rightarrow F[l]$ at each $l$
as composition of $\sfm$, $\sfn$;

\item for every object $D\varOmega\in\Obj_{\ul{\rm GM}_\varOmega}$, 
the function collection $\id_{D\varOmega}:D[l]\rightarrow D[l]$, $l\in\bbN$,
of the identity functions $\id_{D[l]}$ at each $l$ as the identity $\id_{D\varOmega}$
of $D\varOmega$. 

\end{enumerate}
\end{prop}

%

\begin{proof}
To begin with, we must show 
that if the function collections $\sfm:D[l]\rightarrow E[l]$, $\sfn:E[l]\rightarrow F[l]$, 
$l\in\bbN$, satisfy all relations \ceqref{catjoint1}--\ceqref{catjoint3}, then their composition
$\sfn\circ\sfm:D[l]\rightarrow F[l]$ also does and that the identity collection
$\sfi\sfd_{D\varOmega}:D[l]\rightarrow D[l]$, $l\in\bbN$, also satisfy \ceqref{catjoint1}--\ceqref{catjoint2}.
This is a matter of a straightforward verification. Next, we have to show the associativity of morphism
composition and the neutrality of the identity morphisms. 
Since morphisms are collections of functions indexed by $\bbN$ and morphism composition
consists of composition of functions at fixed index, the said properties are an evident consequence of the
corresponding properties of associativity and unitality of function composition and identity functions.
\end{proof}

Graded monads and their morphisms \ccite{Mellies:2012mea,Katsumata:2014pms,Mellies:2017pcm,Fujii:2019gim}
are a kind of abstract categorical constructs illustrated in app. \cref{subsec:gradmon}
with ramifications in several branches of applied mathematics. 
They also provide an appropriate categorical framework for graded $\varOmega$ monads 
and their morphisms. The analysis of these categorical properties is included here only
to show that the notions introduced in this subsection can be framed in a broader well established
categorical design. It is not however a precondition for the understanding of the rest of this work.
The uninterested reader can so safely skip directly to the next section. 

Graded monads and graded monad morphisms constitute a category $\ul{\rm GM}$.
In general, the specification of mathematical entities such as these involves a large amount of data
(cf. defs. \cref{def:gradmon}, \cref{def:gradmonmor}). Such a task can be simplified in certain cases using
paramonoidal graded monads and monad morphisms 
(cf. defs. \cref{def:gradmonqvr}, \cref{def:gradmonmap}), a simpler kind of categorical construct
amenable to a more economical detailing. They also form a category, $\ul{\rm pMGM}$, that is mapped 
into the graded monad category $\ul{\rm GM}$ by a special functor $\msG\msM:\ul{\rm pMGM}\rightarrow\ul{\rm GM}$
(cf. theor. \cref{theor:gradmonqvr}).

The following proposition can be shown.

\begin{prop} \label{prop:jointds}
Let $\clS$ be a small subcategory of the category $\ul{\rm GM}_\varOmega$ of graded $\varOmega$ monads and
morphisms thereof. 
Then, there exists an encompassing strict monoidal category $\clU_{\hfpt\clS}$ with the following properties.
$\clU_{\hfpt\clS}$ is a non monoidal subcategory of the category $\ul{\rm Set}$
of sets and functions. For any graded $\varOmega$ monad $D\varOmega\in\Obj_{\hfpt\clS}$, there is a paramonoidal graded monad
$\msK_{\hfpt\clS}\rmD\varOmega$ consisting of 

\begin{enumerate}[label=\alph*)]

\item \label{it:jointds1} the category $\rmD\varOmega$ itself as grading category, 

\item \label{it:jointds2} the category $\clU_{\hfpt\clS}$ as target category, 

\item \label{it:jointds3} the restrictions $\smallsmile|_{D[l]\times D[m]}$ of the monadic multiplication function
$\smallsmile$ of $\rmD\varOmega$ to the sets $D[l]\times D[m]$ for varying $l,m\in\bbN$ as multiplication morphisms and  

\item \label{it:jointds4} the monadic unit $\iota$ of $\rmD\varOmega$ as unit morphism. 

\end{enumerate}

\noindent
For any two graded $\varOmega$ monads $D\varOmega,E\varOmega\in\Obj_{\hfpt\clS}$ and 
morphism $\sfm\in\Hom_{\hfpt\clS}(D\varOmega,E\varOmega)$, there is further a
paramonoidal graded monad morphism
$\msK_{\hfpt\clS}\sfm:\msK_{\hfpt\clS}\rmD\varOmega\rightarrow\msK_{\hfpt\clS}\rmE\varOmega$ with 

\begin{enumerate}[label=\alph*),start=5]

\item \label{it:jointds5} the functor $\varPhi_{DE}:D\varOmega\rightarrow E\varOmega$ acting on objects and morphisms 
according to \ceqref{diso1}, \ceqref{diso2} as grading category functor, 

\item \label{it:jointds6} the identity functor $\id_{\clU_{\hfpt\clS}}$ of $\clU_{\hfpt\clS}$ as target category functor and 

\item \label{it:jointds7} the restrictions $\sfm|_{D[l]}$ of the morphism function of $\sfm$
to the sets $D[l]$ for varying $l\in\bbN$ as entwining morphisms. 

\end{enumerate}

\noindent
Finally, the mapping $\msK_{\hfpt\clS}:\clS\rightarrow\ul{\rm pMGM}$ 
of the categories $\clS$, $\ul{\rm pMGM}$ defined by items \cref{it:jointds1}--\cref{it:jointds7} above constitutes a
functor injective on objects and morphisms. 

\end{prop} 

\begin{proof}
The details of the construction of $\clU_{\hfpt\clS}$ in a somewhat more abstract and general
perspective are given in the steps {\it \cref{it:prot1}, \cref{it:prot2}} and {\it \cref{it:prot3}}
at the end of app. \cref{subsec:gradmon}.
The quiver $\Delta$ mentioned in step {\it \cref{it:prot1}} has the
monads $D\varOmega\in\Obj_{\hfpt\clS}$ as vertices and the functors $\varPhi_{DE}$
 as edges (cf. def \cref{def:quiv}). 
The monoidal quiver $\clQ$ dealt with in step {\it \cref{it:prot2}} can be described as follows 
(cf. def \cref{def:monquiv}).
The vertex free monoid $\Ver_{\hfpt\clQ}$ of $\clQ$ has the objects $D[l]$, $l\in\bbN$, 
of each monad $D\varOmega\in\Obj_{\hfpt\clS}$ as generating vertices. 
The edge set $\Edg_{\hfpt\clQ}$ comprises the morphisms $Df\in\Hom_{D\varOmega}(D[l],D[m])$,
$l,m\in\bbN$, the multiplication morphisms $\smallsmile|_{D[l]\times D[m]}$, $l,m\in\bbN$, and unit morphism $\iota$
of each category $D\varOmega\in\Obj_{\hfpt\clS}$ as morphism, multiplication and unit edges
and the entwining morphisms $\sfm|_{D[l]}$, $l\in\bbN$, 
of each morphism $\sfm\in\Hom_{\hfpt\clS}(D\varOmega,E\varOmega)$ with objects 
$D\varOmega,E\varOmega\in\Obj_{\hfpt\clS}$ 
as entwining edges and no other edges, 
upon regarding the sets $D[l]$, $D[l]\times D[m]$, etc. as the vertices $D[l]$, $D[l]D[m]$ etc. 
A strict monoidal category $\msF\msM_{GM}\clQ$ is then associated with the monoidal quiver $\clQ$
as detailed in step in step {\it \cref{it:prot2}}. Each object 
$D_0[l_0]\ldots D_{p-1}[l_{p-1}]\in\Obj_{\msF\msM_{GM}\clQ}$ of $\msF\msM_{GM}\clQ$
can be identified with the set of ordered $p$--letter words $\lambda_0\ldots\lambda_{p-1}$ with
$\lambda_0\in D_0[l_0]$, $\ldots$, $\lambda_{p-1}\in D_{p-1}[l_{p-1}]$. Each morphism
$h\in\Hom_{\msF\msM_{GM}\clQ}(D_0[l_0]\ldots D_{p-1}[l_{p-1}],E_0[m_0]\ldots E_{q-1}[m_{q-1}])$
of $\msF\msM_{GM}\clQ$ can similarly be identified with a function mapping each $p$--letter words of the set
attached to its source into a $q$ letter words of the set attached to its target. 
By replacing the composition operation $\circ_{\msF\msM_{GM}\clQ}$ and identity assigning map
$\id_{\msF\msM_{GM}\clQ}$ of $\msF\msM_{GM}\clQ$ with their set theoretic counterparts
$\circ_{\ul{\rm Set}}$ and $\id_{\ul{\rm Set}}$, the category $\msF\msM_{GM}\clQ$ generates a
strict monoidal category $\clU_\clS$, which is a monoidal subcategory of the category
$\ul{\rm Set}$ of sets and functions. With the category $\clU_\clS$ finally constructed, the
rest of the proposition is straightforward to show.
For any $\varOmega$ monad $D\varOmega\in\Obj_{\hfpt\clS}$, the quadruple $\msK_{\hfpt\clS} D\varOmega$
defined by items {\it \cref{it:jointds1}--\cref{it:jointds4}} satisfies by virtue of the properties
\ceqref{joint1}--\ceqref{joint3} of the monadic multiplication of $D\varOmega$
the requirements \ceqref{gradmon12}--\ceqref{gradmon11} stated in def. \cref{def:gradmonqvr} 
and is therefore a paramonoidal graded monad. 
Likewise,
for any two $\varOmega$ monads $D\varOmega,E\varOmega\in\Obj_{\hfpt\clS}$ and morphism $\sfm\in\Hom_{\hfpt\clS}(D\varOmega,E\varOmega)$,
the triple $\msK_{\hfpt\clS}\sfm$ defined by items {\it \cref{it:jointds5}--\cref{it:jointds7}} meets thanks to 
properties \ceqref{catjoint1}--\ceqref{catjoint3} of the morphism $\sfm$ 
all the conditions \ceqref{gradmon18}--\ceqref{gradmon20} stated in def. \cref{def:gradmonmap}
and is consequently a paramonoidal graded monad morphism.
The correspondence $\msK_{\hfpt\clS}:\clS\rightarrow\ul{\rm pMGM}$ of the categories $\clS$, $\ul{\rm pMGM}$
defined by items \cref{it:jointds1}--\cref{it:jointds7} above 
 being a functor follows readily from \ceqref{gradmon22}, \ceqref{gradmon23}. The injectivity of $\msK_{\hfpt\clS}$
on objects and morphisms is evident from items \cref{it:jointds1}--\cref{it:jointds7} again. 
\end{proof}

\noindent 
The functor composite $\msH_{\hfpt\clS}=\msG\msM\circ\msK_{\hfpt\clS}:\clS\rightarrow\ul{\rm GM}$
is injective on objects and morphisms since both the functors $\msK_{\hfpt\clS}$ and $\msG\msM$ 
are (cf. theor. \cref{theor:gradmonqvr}). For this reason, we can identify the subcategory $\clS$
of $\ul{\rm GM}_\varOmega$ with its image $\msH_{\hfpt\clS}\clS$ in $\ul{\rm GM}$, a
subcategory of $\ul{\rm GM}$. This justifies the name we have given to  $\ul{\rm GM}_\varOmega$. 




\vfill\eject

\renewcommand{\sectionmark}[1]{\markright{\thesection\ ~~#1}}

\section{\textcolor{blue}{\sffamily The hypergraph $\varOmega$ monads}}\label{sec:graphctg}

The categorical description of calibrated hypergraph states relies on an underlying categorical description of
calibrated hypergraphs. The categorical model we employ for the treatment of hypergraphs, which we illustrate
in the present section, is that of graded $\varOmega$ monad studied in sect. \cref{sec:omegactg}.

Hypergraphs by themselves are not sufficient for the construction of hypergraph states. 
To that purpose, they need in fact to be supplemented with appropriate hyperedge data called calibrations.
However, we have found conceptually more transparent to treat separately in subsect. \cref{subsec:graphctg}
the bare hypergraph $\varOmega$ monad, which is of independent interest by itself, 
and then based on this to move to the study of the full calibrated hypergraph $\varOmega$ monad,
which is detailed in subsect. \cref{subsec:wggraphctg}. 
It would however possible to directly define and study the calibrated hypergraph $\varOmega$ monad. 
A third related hypergraph monad, the weighted hypergraph $\varOmega$ monad, is considered in 
subsect. \cref{subsec:hemult} mostly to illustrate the relationship of our approach to hypergraphs
to others existing in the literature. This subsection may be skipped in a first reading of the paper. 
Explicit examples are illustrated in detail.
Other examples will be provided in sect. 3 of II 
in relation to hypergraph states. 

The hypergraph categorical construction expounded in subsect. \cref{subsec:graphctg}
is related in some of its aspects to the approach of ref. \!\!\ccite{Spivak:2009hdn}. The material presented in subsects. 
\cref{subsec:wggraphctg}, \cref{subsec:hemult} is instead fully original to the best of our knowledge.

\vspace{1mm}

\subsection{\textcolor{blue}{\sffamily The hypergraph $\varOmega$ monad $G\varOmega$}}\label{subsec:graphctg}

The hypergraph monad $G\varOmega$ is a key graded $\varOmega$ monad playing a central role in our analysis
of hypergraphs and hypergraph states. 
In this subsection, we shall describe its construction and study its properties. 

The analysis presented below employs two endofunctors of the finite set category $\varSigma$ (cf. def. \cref{def:functdef}):
the finite power set endofunctor $P:\varSigma\rightarrow\varSigma$ and the non empty finite power
set endofunctor $P_+:\varSigma\rightarrow\varSigma$. 
The functor $P$ associates with any finite set $A$ its power set $PA$, the set of all subsets of $A$,
and with any function $\phi:A\rightarrow B$ of finite sets the induced power set function $P\phi:PA\rightarrow PB$
given by $P\phi(X)=\phi(X)\in PB$ for $X\in PA$, where in the right hand side $\phi(X)\subset B$ denotes the image of the subset
$X\subset A$ by $\phi$. 
The functor $P_+$ associates with any finite set $A$ its non empty power set $P_+A$, the set of all non empty subsets of $A$,
and with any function $\phi:A\rightarrow B$ of finite sets the induced non empty power
set function $P_+\phi:P_+A\rightarrow P_+B$ defined as $P_+\phi(X)=\phi(X)\in P_+B$
for $X\in P_+A$ analogously to before.  

Hypergraphs are finite collections of hyperedges, which in turn are non empty subsets of a finite set of vertices.
For a fixed vertex set $V$, a hyperedge $E$ is therefore an element of the non empty power set $P_+V$ of $V$
and consequently a hypergraph $I$ is 
an element of the power set $PP_+V$ of $P_+V$. Thus, the iterated power set $PP_+V$ of $V$ is the set of all
hypergraphs that can be built with the vertices of $V$.
For $l\in\bbN$, the ordinal $[l]$ is an abstract vertex set model describing combinatorially
all concrete vertex sets $V$ with $|V|=l$.
Correspondingly, the iterated power set $PP_+[l]$ of $[l]$ is an abstract hypergraph set model
providing a combinatorial description of the concrete hypergraph sets $PP_+V$ of $V$ when $|V|=l$. 


Any map $\phi:V\rightarrow W$ of vertex sets $V$, $W$ induces a map $P_+\phi:P_+V\rightarrow P_+W$
of the associated hyperedge sets and in turn a map $PP_+\phi:PP_+V\rightarrow PP_+W$
of the associated hypergraph sets.
For $l,m\in\bbN$, an ordinal map $f:[l]\rightarrow[m]$ constitutes an abstract vertex set map model
furnishing a combinatorial description of a concrete vertex set map $\phi$ as the above 
when $|V|=l$, $|W|=m$.
Correspondingly, the iterated power set map $PP_+f:PP_+[l]\rightarrow PP_+[m]$ of $f$ provides an 
abstract hypergraph set map model describing combinatorially the concrete hypergraph set map
$PP_+\phi$. 

If $\phi$, $\psi$ are composable vertex set maps, then the associated hypergraph set maps
$PP_+\phi$, $PP_+\psi$ also are and $PP_+\psi\circ PP_+\phi=PP_+(\psi\circ\phi)$. Further,
if $\id_V$ is the identity vertex map of a vertex set $V$, then $PP_+\id_V$ is the identity hypergraph set
map of the hypergraph set $PP_+V$ as $PP_+\id_V=\id_{PP_+V}$.
Therefore, the map $PP_+$ associating a hypergraph set map with each vertex set map
is compatible with composition and identities. These properties are mirrored by totally 
analogous properties of the ordinal maps construed as abstract vertex maps and their
abstract hypergraphs correlates. 



Recall the operation of disjoint union renders the category $\varSigma$ of finite sets and functions
a (non strict) monoidal one. Not only is the disjoint union $A\sqcup B$ of two finite sets $A$, $B$ defined
but also the disjoint union $\phi\sqcup\psi:A\sqcup B\rightarrow C\sqcup D$ of two finite set
functions is $\phi:A\rightarrow C$, $\psi:B\rightarrow D$ is. The latter is defined in obvious fashion
such that $\phi\sqcup\psi|_A=\phi$, $\phi\sqcup\psi|_B=\psi$.

Given two vertex sets $V$, $W$, their disjoint union $V\sqcup W$ is also a vertex
set. The disjoint union of the hypergraph sets $PP_+V$, $PP_+W$ however differs from 
the hypergraph set $PP_+(V\sqcup W)$. On vertices, therefore, 
the iterated power set map $PP_+$ is not compatible with conventionally defined disjoint union. 
We can remedy to this flaw by introducing a weld operation $\smallsmile$
on vertex sets and hypergraph sets thereof defined such that $V\smallsmile W=V\sqcup W$ and
$PP_+V\smallsmile PP_+W=PP_+(V\sqcup W)$ so that the identity $PP_+V\smallsmile PP_+W=PP_+(V\smallsmile W)$
holds.
For $l,m\in\bbN$, the analog of the weld of the ordinals $[l]$, $[m]$ regarded 
as abstract vertex sets is  their monoidal product $[l]\smallsmile[m]$, 
as is evident from eq. \ceqref{gamma1}. Based on the above remark, we may directly define the monoidal product of
the abstract hypergraph sets $PP_+[l]$, $PP_+[m]$ to be
the hypergraph set $PP_+[l]\smallsmile PP_+[m]=PP_+([l]\smallsmile[m])$. 

Likewise, given two maps $\phi:V\rightarrow R$, $\psi:W\rightarrow S$ of vertex sets,
their disjoint union $\phi\sqcup\psi:V\sqcup W\rightarrow R\sqcup T$ is also a vertex set map.
However, the disjoint union of the associated hypergraph set maps,
$PP_+\phi\sqcup PP_+\psi:PP_+V\sqcup PP_+W\rightarrow PP_+R\sqcup PP_+S$, differs from
the hypergraph set map $PP_+(\phi\sqcup\psi):PP_+(V\sqcup W)\rightarrow PP_+(R\sqcup S)$.
On vertex set maps too, so, the iterated power set map $PP_+$ is not in step  with disjoint union. 
Again, we can remedy to this drawback by introducing a weld operation $\smallsmile$ 
on vertex set maps and hypergraph set maps thereof defined by setting $\phi\smallsmile \psi=\phi\sqcup\psi$ and
$PP_+\phi\smallsmile PP_+\psi=PP_+(\phi\sqcup \psi)$ so that the identity
$PP_+\phi\smallsmile PP_+\psi=PP_+(\phi\smallsmile \psi)$ holds.
For $l,m,p,q\in\bbN$, the analog of the weld of the ordinal maps $f:[l]\rightarrow[p]$, $g:[m]\rightarrow[q]$
seen as abstract vertex set maps is their monoidal product $f\smallsmile g$, as follows from inspection of
eqs. \ceqref{gamma1}, \ceqref{gamma2}. By what said above, we may directly define the monoidal product of
the abstract hypergraph set maps $PP_+f$, $PP_+g$ to be
the hypergraph set map $PP_+f\smallsmile PP_+g=PP_+(f\smallsmile g)$. 

The above discussion suggests that abstract hypergraphs and their maps should
naturally organize as an $\varOmega$ category whose stalk isofunctor acts as
the composite functor $P\circ P_+$ (cf. def. \cref{def:omgcat}). This is indeed the line of thought that we find 
most useful and follow in the present work. 


\begin{prop} \label{prop:hffun}
The functor $G=P\circ P_+\big|_\varOmega:\varOmega\rightarrow\varSigma$ is injective on objects and morphisms.
Thus, there exists an $\varOmega$ category $G\varOmega$ having the functor $G$  
with the target category restricted from $\varSigma$ to $G\varOmega$ as its stalk isofunctor.
Furthermore, $G\varOmega$ is a non monoidal subcategory of $\varSigma$.
\end{prop}

\begin{proof}
The $\varSigma$ endofunctors $P$ and $P_+$ are both injective on objects and morphisms. So is thus
their composite $P\circ P_+$. Consequently, $G$ has the same property. The proposition now 
follows immediately from cor. \cref{rem:dtrick}.
\end{proof}

\noindent
We shall refer to the $\varOmega$ category $G\varOmega$ and its stalk
isofunctor $G$ as the hypergraph category and functor, respectively. 
Items \cref{d1}--\cref{d7} of prop. \cref{prop:domegacat}, with $D$ replaced by $G$, furnish an explicit
description of $G\varOmega$. We notice that since $G\varOmega$ is a subcategory of the finite set category $\varSigma$,
the objects and morphisms of $G\varOmega$ are genuine sets and set 
functions respectively and the composition law and the identity assigning map of
$G\varOmega$ are the usual set theoretic ones.
As special cases, we have that 
$G[0]=\{\emptyset\}$ and that $Gf(\emptyset)=\emptyset\in G[m]$ for $f\in\Hom_\varOmega([0],[m])$.


\begin{exa} \label{exa:moracthg} Hypergraph morphism action. 
{\rm Let $f\in\Hom_\varOmega([5],[4])$ be the function 
\begin{equation}
\label{moracth1}
f=(0,1,3,3,1)
\end{equation}
(cf. convention items \cref{it:conv6}, \cref{it:conv7}). Consider further the hypergraph
$H\in G[5]$ with 
\begin{equation}
\label{moracth2}
H=\{X^0,X^1,X^2\},~~\text{where}~~ X^0=\{0,4\},~X^1=\{0,1\}, ~X^2=\{1,2,3\}.
\end{equation}
Then, $Gf(H)\in G[4]$ is the hypergraph 
\begin{align}
\label{moracth3}
&Gf(H)=\{Y^0,Y^1\},~~\text{where}~~ Y^0=\{0,1\}, ~Y^1=\{1,3\}.
\end{align}
}
\end{exa} 





The disjoint union 
of two hypergraphs $I\in PP_+V$, $J\in PP_+W$
is a hypergraph $I\sqcup J\in PP_+(V\sqcup W)$.
In the abstract setting described by hypergraph category $G\varOmega$, a corresponding
operation must associate with any two abstract hypergraph 
$H\in PP_+[l]$, $K\in PP_+[m]$ an abstract hypergraph
$H\smallsmile K\in PP_+([l]\smallsmile[m])$ that in the appropriate sense
to be defined is the disjoint union of $H$, $K$. This is formalized in the following definition. 

\begin{defi} \label{def:hgadj}
For $l,m\in\bbN$, the hypergraph monadic multiplication \pagebreak  is the binary operation 
$\smallsmile:G[l]\times G[m]\rightarrow G([l]\smallsmile[m])$
reading as 
\begin{equation}
\label{hgsmile0}
H\smallsmile K=H\cup(K+l) 
\end{equation}
for $H\in G[l]$, $K\in G[m]$. The hypergraph monadic unit is $O=\emptyset\in G[0]$.
\end{defi}

\noindent In eq. \ceqref{hgsmile0}, the right hand side features the set theoretic union of the sets $H$, $(K+l)$
regarded as subsets of $P_+[l+m]$. We notice that $H\cap(K+l)=\emptyset$. 

\begin{exa} \label{exa:jointhg} Monadic product of hypergraphs.
{\rm Consider the hypergraphs $H\in G[5]$, $K\in G[4]$ specified by 
\begin{align}
\label{jointhgex1}
&H=\{X^0,X^1,X^2\},~~\text{where}~~ X^0=\{0,4\},~X^1=\{0,1\}, ~X^2=\{1,2,3\},
\\
\nonumber  
&K=\{Y^0,Y^1\},~~\text{where}~~ Y^0=\{0,1\},~Y^1=\{1,3\}. 
\end{align}
Then, from \ceqref{hgsmile0}, the joint $H\smallsmile K\in G[9]$ of $H,K$ is 
\begin{align}
\label{jointhgex2}
&H\smallsmile H=\{Z^0,Z^1,Z^2,Z^3,Z^4\},~~\text{where}~~Z^0=\{0,4\},~Z^1=\{0,1\}, 
\\
\nonumber
&\hphantom{H\smallsmile H=\{Z^0,Z^1,Z^2,Z^3,Z^4\},~~\text{wh}~~}Z^2=\{1,2,3\},~Z^3=\{5,6\}, ~Z^4=\{6,8\}.
\end{align}
}  
\end{exa}

\noindent
A hypergraph $\varOmega$ monadic multiplicative structure is now available (cf. def. \cref{def:jointstrct}). 

\begin{prop} \label{prop:hgjoint}
The $\varOmega$ category $G\varOmega$ with the  monadic multiplication $\smallsmile$ and unit $O$ is a graded $\varOmega$ monad.
\end{prop}

\begin{proof}
Using the defining relation \ceqref{hgsmile0} it is immediate to check that 
for $l,m,n\in\bbN$, $H\in G[l]$, $K\in G[m]$, $L\in G[n]$
\begin{equation}
\label{hgsmile1}
H\smallsmile(K\smallsmile L)=(H \smallsmile K)\smallsmile L=H\cup(K+l)\cup(L+l+n).
\end{equation}
The monadic multiplication $\smallsmile$ thus obeys relation \ceqref{joint1}.
It is furthermore readily verified that for every $l\in\bbN$, $H\in G[l]$, \hphantom{xxxxxxxxxx}
\begin{equation}
\label{hgsmile2}
H\smallsmile O=O\smallsmile H=H.
\end{equation}
The monadic unit $O$ so satisfies relation \ceqref{joint2}. There remains to be shown that 
for all $l,m,p,q\in\bbN$, $f\in\Hom_\varOmega([l],[p])$, $g\in\Hom_\varOmega([m],[q])$ and $H\in G[l]$, $K\in G[m]$
\begin{equation}
\label{hgsmile3}
Gf\smallsmile Gg(H\smallsmile K)=Gf(H)\smallsmile Gg(K) \pagebreak
\end{equation}
to have also relation \ceqref{joint3} satisfied. We have 
{\allowdisplaybreaks
\begin{align}
\label{hgsmile3p1}
Gf\smallsmile Gg(H\smallsmile K)&=G(f\smallsmile g)(H\smallsmile K)
\\
\nonumber 
&=PP_+(f\smallsmile g)(H\smallsmile K)=P_+(f\smallsmile g)(H\smallsmile K),
\end{align}
}
\!\!where $P_+(f\smallsmile g)(H\smallsmile K)$ denotes the image of $H\smallsmile K\subset P_+[l+m]$ by 
$P_+(f\smallsmile g)$.
Let $Y\in P_+(f\smallsmile g)(H\smallsmile K)$. Then, there is $X\in H\smallsmile K$ such that 
\begin{equation}
\label{hgsmile3p2}
Y=P_+(f\smallsmile g)(X)=f\smallsmile g(X),
\end{equation}
where $f\smallsmile g(X)$ denotes the image of $X\subset [l+m]$ by $f\smallsmile g$.
As $H\smallsmile K=H\cup(K+l)$ and $H\cap(K+l)=\emptyset$, either $X\in H$ or $X\in K+l$.
By virtue of \ceqref{gamma2},
\begin{equation}
\label{hgsmile3p3}
f\smallsmile g(X)=f(X)=P_+f(X) 
\end{equation}
if $X\in H$ and \hphantom{xxxxxxxxxxxx}
\begin{equation}
\label{hgsmile3p4}
f\smallsmile g(X)=g(X-l)+p=P_+g(X-l)+p 
\end{equation}
if $X\in K+l$. Seeing that in the former case $P_+f(X)\in P_+f(H)$ and in the latter
$P_+g(X-l)+p\in P_+g(K)+p$, where $P_+f(H)$, $P_+g(K)$ denote the images of $H\subset P_+[l]$, \linebreak 
$K\subset P_+[m]$ by $P_+f$, $P_+g$ respectively, 
and taking furthermore into account that $P_+f(H)\cup(P_+g(K)+p)=P_+f(H)\smallsmile P_+g(K)$,
we have $f\smallsmile g(X)\in P_+f(H)\smallsmile P_+g(K)$ implying that $Y\in P_+f(H)\smallsmile P_+g(K)$. 
Next, let $Y\in P_+f(H)\smallsmile P_+g(K)$. Being that 
$P_+f(H)\smallsmile P_+g(K)=P_+f(H)\cup(P_+g(K)+p)$ and $P_+f(H)\cap(P_+g(K)+p)=\emptyset$, we have that either
$Y\in P_+f(H)$ or $Y\in P_+g(K)+p$. If $Y\in P_+f(H)$, there exists $X\in H$
with the property that \hphantom{xxxxxxxx}
\begin{equation}
\label{hgsmile3p5}
Y=P_+f(X)=f(X), 
\end{equation}
while if $Y\in P_+g(K)+p$, there is $X\in K+l$ such that 
\begin{equation}
\label{hgsmile3p6}
Y=P_+g(X-l)+p=g(X-l)+p.
\end{equation}
Owing to \ceqref{gamma2}, as $H\cup(K+l)=H\smallsmile K$, 
there is $X\in H\smallsmile K$ such that 
\begin{equation}
\label{hgsmile3p7}
Y=f\smallsmile g(X)=P_+(f\smallsmile g)(X),
\end{equation}
whence $Y\in P_+(H\smallsmile K)$. We conclude that 
\begin{equation}
\label{hgsmile3p8}
P_+(f\smallsmile g)(H\smallsmile K)=P_+f(H)\smallsmile P_+g(K).
\end{equation}
By \ceqref{hgsmile3p8}, \ceqref{hgsmile3p1} takes the form 
{\allowdisplaybreaks
\begin{align}
\label{hgsmile3p9}
Gf\smallsmile Gg(H\smallsmile K)&=P_+f(H)\smallsmile P_+g(K).
\\
\nonumber 
&=PP_+f(H)\smallsmile PP_+g(K)=Gf(H)\smallsmile Gg(K),
\end{align}
}
\!\! showing \ceqref{hgsmile3}.
\end{proof}


\subsection{\textcolor{blue}{\sffamily The calibrated hypergraph $\varOmega$ monad $G_C\varOmega$}}\label{subsec:wggraphctg}

Bare hypergraphs are not sufficient for the full construction of hypergraph states. To that purpose, it 
is necessary to append to the relevant hypergraphs suitable hyperedge data or calibrations. 
The calibrated hypergraph monad $G_C\varOmega$ is a graded $\
varOmega$ monad describing hypergraphs equipped with such data.
It is the topic of this subsection, where its multiple aspects are analyzed in depth. 

Exponent functions and calibrations thereof are key elements in our construction of calibrated
hypergraphs states. Our analysis so begins with them. 

Let $\msA$ be a finite additive commutative monoid. 

\begin{defi} \label{def:expfnc}
Let $X\subset\bbN$ be a finite subset of $\bbN$. The exponent monoid of $X$ is the commutative monoid
$\msA^X$ of all $\msA$--valued functions on $X$ under pointwise addition.
\end{defi}

\noindent
When $X=\emptyset$, then $\msA^X=0$ is the trivial commutative monoid consisting
of the neutral element $0_\emptyset=0$ only. 
The elements $w\in\msA^X$ are called exponent functions of $X$ in the following.

It is important to have a concise notation for exponent functions. 
We index standardly the finite subset $X\subset\bbN$ as $X=(r_0,\ldots,r_{|X|-1})$ with $r_0<\cdots<r_{|X|-1}$
in conformity with convention item \cref{it:conv6}. Convention item \cref{it:conv7} then allows us to specify 
an exponent function $w\in\msA^X$ through an $|X|$--tuple of elements of $\msA$ of the form 
\begin{equation}
\label{expcnv1}
w=(w(r_0),\ldots,w(r_{|X|-1})).
\end{equation}

\begin{defi} \label{def:fstarw}
Let $X,Y\subset\bbN$ be finite subsets and $f:X\rightarrow Y$ be a function. The exponent 
function push--forward map $f_\star:\msA^X\rightarrow\msA^Y$ is given for each $w\in\msA^X$ by 
\begin{equation}
\label{exppf0}
f_\star(w)(s)=\mycom{{}_\sss}{{}_{r\in X,f(r)=s}}w(r)
\end{equation}
with $s\in Y$. 
\end{defi}

\noindent
If $s\notin f(X)$, then $f_\star(w)(s)=0$. In particular, 
when $X=\emptyset$ and so $f=e_Y$, we find that $e_{Y\star}(0_\emptyset)=0_Y$, the vanishing exponent function of $Y$,
as expected. When $f$ is invertible, $f_\star(w)=w\circ f^{-1}$.

\begin{exa} \label{exa:pfef} Exponent function push--forward. 
{\rm Let $X,Y\subset\bbN$ with  $X=\{2,4,5,9\}$, $Y=\{0,3,7,8,9\}$ be finite integer sets.
Consider the function $f:X\rightarrow Y$ given in accordance with convention items
\cref{it:conv6}, \cref{it:conv7} by 
\begin{equation}
\label{pfefex1}
f=(3,3,7,8).
\end{equation}
Let $\msA=\bbH_{3,1}$, 
the commutative monoid with underlying set $[4]$, additive unit $0$
and addition defined by $u+_{3,1}v=\min(u+v,3)$ for $u,v\in[4]$.
Let $w_0,w_1,w_2,w_3\in\msA^X$ be the exponent functions 
\begin{align}
\label{pfefex3} 
&w_0=(3,2,1,3), &w_1=(2,1,1,3),
\\
\nonumber
&w_2=(1,3,0,2), &w_3=(1,2,3,2). 
\end{align}
Their push--forward $f_\star(w_0),f_\star(w_1),f_\star(w_2),f_\star(w_3)\in\msA^Y$ of $w_0,w_1,w_2,w_3$ are
\begin{align}
\label{pfefex4}
&f_\star(w_0)=f_\star(w_1)=(0,3,1,3,0), &f_\star(w_2)=(0,3,0,2,0), 
\\
\nonumber
&f_\star(w_3)=(0,3,3,2,0)
\end{align}
by a simple application of \ceqref{exppf0}.
}
\end{exa}  

The exponent function push--forward operation we have introduced
enjoys a number of nice properties.

\begin{prop} \label{prop:star1}
Let $X,Y\subset\bbN$ be finite subsets and $f:X\rightarrow Y$ be a function. Then,
$f_\star$ is a monoid morphism of $\msA^X$ into $\msA^Y$.
\end{prop}

\begin{proof}
The monoid morphism property of $f_\star$ follow straightforwardly
from the defining expression \ceqref{exppf0}.
\end{proof}

\begin{prop} \label{prop:star2}
Let $X,Y,Z\subset\bbN$ be finite subsets and $f:X\rightarrow Y$, $g:Y\rightarrow Z$ be functions. Then,
the relation \hphantom{xxxxxxx}
\begin{equation}
\label{exppf1}
(g\circ f)_\star=g_\star\circ f_\star
\end{equation}
holds. Let $X\subset\bbN$ be a finite subset. Then, one has 
\begin{equation}
\label{exppf2}
\id_{X\star}=\id_{\msA^X}.
\end{equation}
Let $X,Y\subset\bbN$ be finite subsets and $f:X\rightarrow Y$ be an invertible function. Then, $f_\star$ is a also
invertible and moreover it holds that \hphantom{xxxxxxxxxxxx}
\begin{equation}
\label{exppf3}
f^{-1}{}_\star=f_\star{}^{-1}. 
\end{equation}
\end{prop}

\begin{proof} 
We show first \ceqref{exppf1}. Take $w\in\msA^X$. For $t\in Z$, 
{\allowdisplaybreaks
\begin{align}
\label{exppfp1}
g_\star\circ f_\star(w)(t)&=\mycom{{}_\sss}{{}_{s\in Y,g(s)=t}}f_\star(w)(s)
\\
\nonumber
&=\mycom{{}_\sss}{{}_{s\in Y,g(s)=t}}\mycom{{}_\sss}{{}_{r\in X,f(r)=s}}w(r)
\\
\nonumber
&=\mycom{{}_\sss}{{}_{s\in Y}}\mycom{{}_\sss}{{}_{r\in X}}\delta_{g(s),t}\delta_{f(r),s}w(r)
\\
\nonumber
&=\mycom{{}_\sss}{{}_{r\in X}}\delta_{g\hfpt\circ f(r),t}w(r)
\\
\nonumber
&=\mycom{{}_\sss}{{}_{r\in X,g\hfpt\circ f(r)=t}}w(r)=(g\circ f)_\star (w)(t),
\end{align}
}
\!\!where the Kronecker delta functions have been used to turn restricted summations
into unrestricted ones. \ceqref{exppf1} is so shown. We show next \ceqref{exppf2}. Take $w\in\msA^X$. For $s\in X$, 
{\allowdisplaybreaks
\begin{align}
\label{exppfp2}
\id_{X\star}(w)(s)&=\mycom{{}_\sss}{{}_{r\in X,\id_X(r)=s}}w(r)
\\
\nonumber
&=\mycom{{}_\sss}{{}_{r\in X,r=s}}w(r)=\id_{\msA^X}(w)(s).
\end{align}
}
\!\!\ceqref{exppf2} is thus also shown. \ceqref{exppf3} follows immediately from \ceqref{exppf1},
applied to the function pairs $(f,f^{-1})$, $(f^{-1},f)$ and \ceqref{exppf2}.
\end{proof}

\noindent 
Props. \cref{prop:star1}, \cref{prop:star1} together entail that the maps $X\mapsto\msA^X$, $f\mapsto f_\star$
define a functor from the category of finite integer sets and set functions to the category of
finite monoids and monoid morphisms. However, we shall not insist on this categorical interpretation
in the following. 

The definition of calibrations involves introducing another finite additive commutative monoid $\msM$.

\begin{defi} \label{def:edgewgt}
Let $X\subset\bbN$ be any finite subset. The calibration monoid of $X$ is the commutative monoid
$\msM^{\msA^X}$ of all $\msM$--valued functions on $\msA^X$ under pointwise addition.
\end{defi}

\noindent
If $X=\emptyset$, then $\msM^{\msA^X}\equiv\msM$ since $\msA^X=0$. The elements $\varpi\in\msM^{\msA^X}$ are called
calibrations of $X$ and can be thought of as weight functions of the exponent functions of $X$.

It is important to have a compact notation also for calibrations.
The exponent monoid $\msA^X$ has no canonical indexing, but any chosen indexing reads as
\begin{equation}
\label{expcnv2}
\msA^X=(w_0,\ldots,w_{|\msA^X|-1})
\end{equation}
with distinct $w_i\in\msA^X$, where $|\msA^X|=|\msA|^{|X|}$. 
Once an indexing of $\msA^X$ is selected, a calibration $\varpi\in\msM^{\msA^X}$ is specified by an
$|\msA^X|$--tuple of elements of $\msM$ of the form
\begin{equation}
\label{cali1}
\varpi=(\varpi(w_0),\ldots,\varpi(w_{|\msA^X|-1})). 
\end{equation}

\begin{defi} \label{def:edgewgtpf}
Let $X,Y\subset\bbN$ be finite subsets and $f:X\rightarrow Y$ be a function.
The calibration push--forward map $f_*:\msM^{\msA^X}\rightarrow \msM^{\msA^Y}$ is given for each $\varpi\in \msM^{\msA^X}$ by
\begin{equation}
\label{wgtpf0}
f_*(\varpi)(v)=\mycom{{}_\sss}{{}_{w\in \msA^X,f_\star(w)=v}}\varpi(w)
\end{equation}
with $v\in\msA^Y$. 
\end{defi}

\noindent
If $v\notin f_\star(\msA^X)$ above, then $f_*(\varpi)(v)=0$. In particular, 
when $X=\emptyset$ and so $f=e_Y$,
Notice that $f_*(\varpi)=\varpi\circ f^{-1}{}_\star$ when $f$ is invertible. 

\begin{exa} \label{exa:pfcf} Calibration push--forward.
{\rm Consider again the subsets $X,Y\subset\bbN$, function $f:X\rightarrow Y$ and monoid $\msA$
introduced earlier on in ex. \cref{exa:pfef}. Let $\msM=\bbH_{0,5}$, where $\bbH_{0,5}$ is the commutative
monoid underlying the order $5$ cyclic group $\bbZ_5$.
Let $\msA^X=(w_0,\ldots,w_{255})$, $\msA^Y=(v_0,\ldots,w_{1023})$ be indexings of $\msA^X$, $\msA^Y$
with $w_0,w_1,w_2,w_3$ given by expressions \ceqref{pfefex3} and $v_0,v_1,v_2$ given by the right hand sides
of \ceqref{pfefex4} for convenience. Let e.g. $\varpi_0,\varpi_1,\varpi_2,\varpi_3\in\msM^{\msA^X}$ be the calibrations 
given by
\begin{align}
\label{pfcfex3}
&\varpi_0=(0,2,4,3,0,\ldots,0), &\varpi_1=(2,1,0,4,0,\ldots,0),
\\
\nonumber
&\varpi_2=(4,2,4,3,0,\ldots,0), &\varpi_3=(4,1,3,4,0,\ldots,0).
\end{align}
By \ceqref{pfefex4} and \ceqref{wgtpf0}, their push--forwards 
$f_*(\varpi_0),f_*(\varpi_1),f_*(\varpi_2),f_*(\varpi_3)\in\msM^{\msA^Y}$ are 
\begin{align}
\label{pfcfex4}
&f_*(\varpi_0)=(2,4,3,0,\ldots,0), &f_*(\varpi_1)=(3,0,4,0,\ldots,0),
\\
\nonumber  
&f_*(\varpi_2)=(1,4,3,0,\ldots,0), &f_*(\varpi_3)=(0,3,4,0,\ldots,0). 
\end{align}
}
\end{exa}

The calibration push--forward operation also enjoys 
several nice properties, reflecting the corresponding properties of the underlying exponent functions. 


\begin{prop} \label{prop:wxsmgr}
Let $X,Y\subset\bbN$ be finite subsets and $f:X\rightarrow Y$ be a function. Then,
$f_*$ is a monoid morphism of $\msM^{\msA^X}$ into $\msM^{\msA^Y}$.
\end{prop}

\begin{proof}
The monoid morphism property of $f_*$ is evident from the defining expression \ceqref{wgtpf0}. 
\end{proof}

\begin{prop} \label{prop:wxsfunct}
Let $X,Y,Z\subset\bbN$ be finite subsets and $f:X\rightarrow Y$, $g:Y\rightarrow Z$ be functions.
Then, the relation \hphantom{xxxxxxxxx}
\begin{equation}
\label{wgtpf1}
(g\circ f)_*=g_*\circ f_*
\end{equation}
holds. Let $X\subset\bbN$ be a finite subset. Then, 
\begin{equation}
\label{wgtpf2}
\id_{X*}=\id_{\msM^{\msA^X}}. 
\end{equation}
Let $X,Y\subset\bbN$ be finite subsets and $f:X\rightarrow Y$ be an invertible function. Then, $f_*$ is a also
invertible and moreover one has \hphantom{xxxxxxxxxxxx}
\begin{equation}
\label{wgtpf3}
f^{-1}{}_*=f_*{}^{-1}. 
\end{equation}
\end{prop}


\begin{proof} We show first \ceqref{wgtpf1}. Fix $\varpi\in \msM^{\msA^X}$. For $u\in\msA^Z$,
{\allowdisplaybreaks
\begin{align}
\label{wgtpfp1}
g_*\circ f_*(\varpi)(u)&=\mycom{{}_\sss}{{}_{v\in\msA^Y,g_\star(v)=u}}f_*(\varpi)(v)
\\
\nonumber
&=\mycom{{}_\sss}{{}_{v\in\msA^Y,g_\star(v)=u}}\mycom{{}_\sss}{{}_{w\in\msA^X,f_\star(w)=v}}\varpi(w)
\\
\nonumber
&=\mycom{{}_\sss}{{}_{v\in\msA^Y}}\mycom{{}_\sss}{{}_{w\in\msA^X}}\delta_{g_\star(v),u}\delta_{f_\star(w),v}\varpi(w)
\\
\nonumber
&=\mycom{{}_\sss}{{}_{w\in\msA^X}}\delta_{g_\star\circ f_\star(w),u}\varpi(w)
\\
\nonumber
&=\mycom{{}_\sss}{{}_{w\in\msA^X,(g\hfpt\circ f)_\star(w)=u}}\varpi(w)=(g\circ f)_*\varpi(u),
\end{align}
}
\!\!where we used \ceqref{exppf1}, giving
\ceqref{wgtpf1}. We show next \ceqref{wgtpf2}. Fix $\varpi\in \msM^{\msA^X}$. Then, 
\begin{align}
\label{wgtpfp2}
\id_{X*}(\varpi)(v)&=\mycom{{}_\sss}{{}_{w\in\msA^X,\id_{X\star}(w)=v}}\varpi(w)
\\
\nonumber
&=\mycom{{}_\sss}{{}_{w\in\msA^X,w=v}}\varpi(w)=\id_{\msM^{\msA^X}}(\varpi)(v)
\end{align}
for $v\in\msA^X$, where \ceqref{exppf2} was employed. \ceqref{wgtpf2} is thus shown.
\ceqref{wgtpf3} follows readily from \ceqref{wgtpf1},
applied to the function pairs $(f,f^{-1})$, $(f^{-1},f)$ and \ceqref{wgtpf2}.
\end{proof}

\noindent
Props. \cref{prop:wxsmgr}, \cref{prop:wxsfunct} have a categorical reading analogous to that of 
props. \cref{prop:star1}, \cref{prop:star1}.

Calibrations can be used to `decorate' the hypergraph category $G\varOmega$
introduced and studied in subsect. \cref{subsec:graphctg}. To this end, we need a further notion. 

\begin{defi} \label{def:hygrpwgt}
Let $l\in\bbN$. A calibration $\varrho$ of a hypergraph $H\in G[l]$ is an assignment to every hyperedge
$X\in H$ of a calibration $\varrho_X\in \msM^{\msA^X}$. $C(H)$ is the set of all hypergraph calibrations of $H$.  
\end{defi}

\noindent 
If $H=\emptyset$, then $C(H)=\{e_\emptyset\}$, where $e_\emptyset$ is the empty function with 
range $\emptyset$. This follows from noting that 
\begin{equation}
\label{}
C(H)=\mycom{{}_\ppp}{{}_{X\in H}}\msM^{\msA^X}=\bigg\{\varrho\hfpt\bigg|\varrho:H\rightarrow\mycom{{}_\uuu}{{}_{X\in H}}\msM^{\msA^X}, 
\forall X\in H:\varrho_X\in \msM^{\msA^X}\bigg\} 
\end{equation}
in general. By virtue of the above expression, 
$C(H)$ consists therefore only of $e_\emptyset$ when $H=\emptyset$. 

The push--forward operation we defined for calibrations of any finite subset of the non negative integers
naturally induces an ordinal morphism push--forward action on hypergraph calibrations as follows.


\begin{defi} \label{def:whgpf}
Let $l,m\in\bbN$ and $f\in\Hom_\varOmega([l],[m])$. Let $H\in G[l]$ be a hypergraph. The hypergraph
calibration push--forward map $f_{H*}:C(H)\rightarrow C(Gf(H))$ is given by 
\begin{equation}
\label{whgpf0}
f_{H*}(\varrho)_Y=\mycom{{}_\sss}{{}_{X\in H,f(X)=Y}}f|_{X*}(\varrho_X)
\end{equation}
with $Y\in Gf(H)$ for $\varrho\in C(H)$.
\end{defi} 

\noindent
In the right hand side of \ceqref{whgpf0}, the calibration 
$f|_{X*}(\varrho_X)\in \msM^{\msA^{f(X)}}$ is the push--forward by $f|_X:X\rightarrow f(X)$ of 
$\varrho_X\in \msM^{\msA^X}$ defined according to \eqref{wgtpf0}. Further, the summation 
is to be regarded as monoid addition in $\msM^{\msA^Y}$. 
If $H=Gf(H)=\emptyset$, then $f_{\emptyset*}(e_\emptyset)=e_\emptyset$ for reasons explained above.

\begin{exa} \label{exa:chgmact} 
Calibrated hypergraph push--forward action. 
{\rm We assume that $\msA=\bbH_{1,1}$, the commutative monoid of the set $[2]$ with additive unit $0$ 
and addition obeying $1+_{1,1}1=1$, and that $\msM=\bbH_{0,3}$, the commutative
monoid underlying the order $3$ cyclic group $\bbZ_3$.
Consider the hypergraph $H\in G[5]$ shown in \ceqref{moracth2}. 
To describe the calibrations of $H$, we need 
suitable indexings of the exponent monoids $\msA^{X^0}$, $\msA^{X^1}$, $\msA^{X^2}$
of the hyperedges $X^0,X^1,X^2$ of $H$. We assume here these are 
$\msA^{X^0}=(w^0{}_0,\ldots,w^0{}_3)$, $\msA^{X^1}=(w^1{}_0,\ldots,w^1{}_3)$,
$\msA^{X^2}=(w^2{}_0,\ldots,w^2{}_7)$, where 
{\allowdisplaybreaks
\begin{align}
\label{chgmactex1}
&w^0{}_0=(0,0), &w^0{}_1=(0,1),\hspace{.4cm}&&w^0{}_2=(1,0), \hspace{.4cm}&&w^0{}_3=(1,1), \hspace{.4cm}
\\
\nonumber
&w^1{}_0=(0,0), &w^1{}_1=(0,1),\hspace{.4cm} &&w^1{}_2=(1,0),\hspace{.4cm} &&w^1{}_3=(1,1), \hspace{.4cm}
\\
\nonumber 
&w^2{}_0=(0,0,0), &w^2{}_1=(0,0,1), &&w^2{}_2=(0,1,0), &&w^2{}_3=(0,1,1), 
\\
\nonumber
&w^2{}_4=(1,0,0), &w^2{}_5=(1,0,1), &&w^2{}_6=(1,1,0), &&w^2{}_7=(1,1,1). 
\end{align}
}
\!\!A calibration $\varrho\in C(H)$ is e.g. 
\begin{align}
\label{chgmactex2}
&\varrho_{X^0}=(0,1,1,2), &\varrho_{X^1}=(2,0,1,0), &&\varrho_{X^2}=(1,1,2,0,2,0,0,1). 
\end{align}
Consider next the morphism $f\in\Hom_\varOmega([5],[4])$ given by eq. \ceqref{moracth1}.
The transformed hypergraph $Gf(H)\in G[4]$ is shown in \ceqref{moracth3}.
To express the push--forward action of $f$ on $\varrho$, the push--forward action of $f$ on the exponent functions $w^a{}_i$
is required. We need therefore indexings also of
the exponent monoids $\msA^{Y^0}$, $\msA^{Y^1}$ of the hyperedges $Y^0,Y^1$ of $Gf(H)$. 
We assume these are $\msA^{Y^0}=(v^0{}_0,\ldots,v^0{}_3)$, $\msA^{Y^1}=(v^1{}_0,\ldots,v^1{}_3)$ with 
{\allowdisplaybreaks
\begin{align}
\label{chgmactex3}
&v^0{}_0=(0,0), &v^0{}_1=(0,1), \hspace{.4cm}&&v^0{}_2=(1,0), \hspace{.4cm}&&v^0{}_3=(1,1), \hspace{.4cm}
\\
\nonumber
&v^1{}_0=(0,0), &v^1{}_1=(0,1),\hspace{.4cm} &&v^1{}_2=(1,0), \hspace{.4cm} &&v^1{}_3=(1,1). \hspace{.4cm}
\end{align}
}
\!\!The exponent functions $f|_{X^0\star}(w^0{}_0),\ldots,f|_{X^0\star}(w^0{}_3), 
f|_{X^1\star}(w^1{}_0),\ldots,f|_{X^1\star}(w^1{}_3)\in\msA^{Y^0}$,
$f|_{X^2\star}(w^2{}_i),\ldots,f|_{X^2\star}(w^2{}_7)\in\msA^{Y^1}$ are 
{\allowdisplaybreaks
\begin{align}
\label{chgmactex4}
&f|_{X^0\star}(w^0{}_0)=(0,0), &f|_{X^0\star}(w^0{}_1)=(0,1),&&f|_{X^0\star}(w^0{}_2)=(1,0),
\\
\nonumber
&f|_{X^0\star}(w^0{}_3)=(1,1),
\\
\nonumber
&f|_{X^1\star}(w^1{}_0)=(0,0), &f|_{X^1\star}(w^1{}_1)=(0,1),&&f|_{X^1\star}(w^1{}_2)=(1,0),
\\
\nonumber
&f|_{X^1\star}(w^1{}_3)=(1,1), 
\\
\nonumber
&f|_{X^2\star}(w^2{}_0)=(0,0), &f|_{X^2\star}(w^2{}_1)=(0,1),&&f|_{X^2\star}(w^2{}_2)=(0,1), 
\\
\nonumber
&f|_{X^2\star}(w^2{}_3)=(0,1), &f|_{X^2\star}(w^2{}_4)=(1,0),&&f|_{X^2\star}(w^2{}_5)=(1,1),
\\
\nonumber
&f|_{X^2\star}(w^2{}_6)=(1,1), &f|_{X^2\star}(w^2{}_7)=(1,1)&&~ 
\end{align}  
}
\!\!by \ceqref{exppf0}.
Employing these identities, it is now straightforward using \ceqref{wgtpf0} and \ceqref{whgpf0} to obtain the
push--forward $f_{H*}(\varrho)\in C(Gf(H))$ of the calibration $\varrho$. We find
{\allowdisplaybreaks
\begin{align}
\label{chgmactex5}
&f_{H*}(\varrho)_{Y^0}=(2,1,2,2), &f_{H*}(\varrho)_{Y^1}=(1,0,2,1). 
\end{align}
}
\vspace{-1cm}}
\end{exa}

The hypergraph calibration push--forward operation enjoys nice properties too. 

\begin{prop} \label{prop:whgpfprop}
Assume that $l,m,n\in\bbN$, $f\in\Hom_\varOmega([l],[m])$, $g\in\Hom_\varOmega([m],[n])$. Let further $H\in G[l]$. Then,
one has 
\begin{equation}
\label{whgpf1}
(g\circ f)_{H*}=g_{Gf(H)*}\circ f_{H*}.
\end{equation}
Let $l\in\bbN$. Let further $H\in G[l]$. Then, one has 
\begin{equation}
\label{whgpf2}
\id_{[l]H*}=\id_{C(H)}.
\end{equation}
\end{prop}

\begin{proof} We demonstrate \ceqref{whgpf1} first. By virtue of prop. \cref{prop:wxsmgr} and 
relation \ceqref{wgtpf1}, for $\varrho\in C(H)$ and $Z\in Gg(Gf(H))=G(g\circ f)(H)$
{\allowdisplaybreaks
\begin{align}
\label{whgpfp1}
g_{Gf(H)*}\circ f_{H*}(\varrho)_Z&=\mycom{{}_\sss}{{}_{Y\in Gf(H),g(Y)=Z}}g|_{Y*}(f_{H*}(\varrho)_Y)
\\
\nonumber
&=\mycom{{}_\sss}{{}_{Y\in Gf(H),g(Y)=Z}}\mycom{{}_\sss}{{}_{X\in H,f(X)=Y}}g|_{Y*}\circ f|_{X*}(\varrho_X)
\\
\nonumber
&=\mycom{{}_\sss}{{}_{Y\in Gf(H)}}\mycom{{}_\sss}{{}_{X\in H}}\delta_{g(Y),Z}\delta_{f(X),Y} (g|_{f(X)}\circ f|_X)_*(\varrho_X)
\\
\nonumber
&=\mycom{{}_\sss}{{}_{X\in H}}\delta_{g\hfpt\circ f(X),Z}(g\circ f)|_{X*}(\varrho_X)
\\
\nonumber
&=\mycom{{}_\sss}{{}_{X\in H,g\hfpt\circ f(X)=Z}}(g\circ f)|_{X*}(\varrho_X)=(g\circ f)_{H*}(\varrho)_Z,
\end{align}
}
\!\!yielding \ceqref{whgpf1}. Next, we show \ceqref{whgpf2}. For $\varrho\in C(H)$ and $Y\in G\id_{[l]}(H)=H$, 
\begin{align}
\label{whgpfp2}
\id_{[l]H*}(\varrho)_Y&=\mycom{{}_\sss}{{}_{X\in H,\id_{[l]}(X)=Y}}\id_{[l]}|_{X*}(\varrho_X)
\\
\nonumber
&=\mycom{{}_\sss}{{}_{X\in H,X=Y}}\id_{X*}(\varrho_X)=\id_{C(H)}(\varrho)_Y,
\end{align}
leading to \ceqref{whgpf2}. 
\end{proof}

Relying on the analysis of hypergraph calibrations developed above,
we can now introduce the $\varOmega$ category of calibrated hypergraphs
designed as the enhancement of the hypergraph category $G\varOmega$ studied in subsect. \cref{subsec:graphctg}
procured by appropriately incorporating the hyperedge calibration data. 

\begin{prop} \label{prop:gwfinj}
The prescription assigning

\begin{enumerate}[label=\alph*)] 

\item an object $G_C[l]\in\Obj_\varSigma$ given by 
\begin{equation}
\label{wgfundt1}
G_C[l]=\{(H,\varrho)|H\in G[l],\varrho\in C(H)\}
\end{equation}
for every $l\in\bbN$ and 

\item a morphism $G_Cf\in\Hom_\varSigma(G_C[l],G_C[m])$ given by 
\begin{equation}
\label{wgfundt2}
G_Cf(H,\varrho)=(Gf(H),f_{H*}(\varrho))
\end{equation}
with $(H,\varrho)\in G_C[l]$ for every $l,m\in\bbN$ and morphism $f\in\Hom_\varOmega([l],[m])$

\end{enumerate}

\noindent
defines a functor $G_C:\varOmega\rightarrow\varSigma$ injective on objects and morphisms.
Therefore, there exists an $\varOmega$ category $G_C\varOmega$ having the functor $G_C$  
with the target category restricted \linebreak from $\varSigma$ to $G_C\varOmega$ as its stalk isofunctor.
Furthermore, $G_C\varOmega$ is a non monoidal subcategory of $\varSigma$.
\end{prop}

\begin{proof}
Let $l,m,n\in\bbN$, $f\in\Hom_\varOmega([l],[m])$, $g\in\Hom_\varOmega([m],[n])$ and $(H,\varrho)\in G_C[l]$.
By \ceqref{wgfundt2}, recalling that $G$ is a functor and using \ceqref{whgpf1}, we find 
{\allowdisplaybreaks
\begin{align}
\label{wgfunctp1}
G_Cg\circ G_Cf(H,\varrho)&=(Gg\circ Gf(H),g_{Gf(H)*}\circ f_{H*}(\varrho))  
\\
\nonumber
&=(G(g\circ f)(H),(g\circ f)_{H*}(\varrho))=G_C(g\circ f)(H,\varrho),
\end{align}
}
\!so that $G_C$ preserves morphism composition. Let further $l\in\bbN$ and $(H,\varrho)\in G_C[l]$. Then,
by \ceqref{wgfundt2}, exploiting the functoriality of $G$ and using \ceqref{whgpf2}, we find 
\begin{align}
\label{wgfunctp2}
G_C\id_{[l]}(H,\varrho)&=(G\id_{[l]}(H),\id_{[l]H*}(\varrho))
\\
\nonumber
&=(\id_{G[l]}(H),\id_{C(H)}(\varrho))=\id_{G_C[l]}(H,\varrho),
\end{align}
showing that $G_C$ preserves identity assignments. Therefore, $G_C$ is a functor. 
Next, let $l\in\bbN$. There exists an obvious projection $p_{[l]}:G_C[l]\rightarrow G[l]$
defined by $p_{[l]}(H,\varrho)=H$ for $(H,\varrho)\in G_C[l]$. It has the property that $p_{[l]}(G_C[l])=G[l]$.
Let $l,m\in\bbN$ and suppose that $G_C[l]=G_C[m]$. Then, for every $(H,\varrho)\in G_C[l]=G_C[m]$ it holds that 
$p_{[l]}(H,\varrho)=p_{[m]}(H,\varrho)=H$. Hence, $p_{[l]}=p_{[m]}$. Consequently,
\begin{equation}
\label{}
G[l]=p_{[l]}(G_C[l])=p_{[m]}(G_C[m])=G[m].
\end{equation}
As the functor $G$ is injective on objects by prop. \cref{prop:hffun}, $[l]=[m]$. Therefore, the functor
$G_C$ is also injective on objects.  Next, let $l,m\in\bbN$, $f,g\in\Hom_\varOmega([l],[m])$ with $G_Cf=G_Cg$. Then, 
{\allowdisplaybreaks
\begin{align}
\label{}
G_Cf(H,\varrho)&=(Gf(H),f_{H*}(\varrho))
\\
\nonumber
&=G_Cg(H,\varrho)=(Gg(H),g_{H*}(\varrho))
\end{align}
}
\!\!for $(H,\varrho)\in G_C[l]$.
Hence, for all $H\in p_{[l]}G_C[l]=G[l]$, $Gf(H)=Gg(H)$. Given that the functor $G$ is injective on morphisms
by prop. \cref{prop:hffun}, $f=g$. It follows that $G_C$ is injective on morphisms too.
The existence of an $\varOmega$ category $G_C\varOmega$ with $G_C$ as its stalk isofunctor
and its being a subcategory of $\varSigma$ now follow readily from cor. \cref{rem:dtrick}.
\end{proof}

\noindent
We shall refer to the $\varOmega$ category $G_C\varOmega$ and its stalk
isofunctor $G_C$ as the calibrated hypergraph category and functor, respectively. 
Items \cref{d1}--\cref{d7} of prop. \cref{prop:domegacat}, with $D$ replaced by $G_C$, make available an explicit
description of $G_C\varOmega$. We note that given that  $G_C\varOmega$ is a subcategory of the finite set category
$\varSigma$, the objects and morphisms of $G_C\varOmega$ are again genuine sets and set 
functions respectively and the composition law and the identity assigning map of
$G_C\varOmega$ are once more the usual set theoretic ones. In particular, we have 
that $G_C[0]=\{(\emptyset,e_\emptyset)\}$ and that $G_Ce_{[m]}(\emptyset,e_\emptyset)=(\emptyset,e_\emptyset)\in G_C[m]$
for $e_{[m]}\in\Hom_\varOmega([0],[m])$.

Let $l\in\bbN$ and let $X\subset\bbN+l$ be a finite subset. The $l$--shift function of $X$ 
is the function $t_{Xl}:X\rightarrow X-l$ defined by \hphantom{xxxxxxxx}
\begin{equation}
\label{tlw0}
t_{Xl}(r)=r-l
\end{equation}
with $r\in X$. If $X=\emptyset$, then $X-l=\emptyset$ and $t_{\emptyset l}=e_\emptyset$,
the empty function with range $\emptyset$. 
It is straightforward to check that $t_{Xl}$ is a bijection. We notice that 
\begin{equation}
\label{tlw11}
t_{X-m\hfpt l}\circ t_{Xm}=t_{Xl+m}
\end{equation}
for any $l,m\in\bbN$ and finite subset $X\subset\bbN+l+m$. Further, 
\begin{equation}
\label{tlw2}
t_{X0}=\id_X
\end{equation}
for any finite subset $X\subset\bbN$. 

\begin{defi} \label{def:hgwsm}
For all $l,m\in\bbN$ and $H\in G[l]$, $K\in G[m]$, the hypergraph calibration monadic multiplication
$\smallsmile:C(H)\times C(K)\rightarrow C(H\smallsmile K)$ (cf. def. \cref{def:hgadj}) is defined as 
\begin{equation}
\label{rsmiles1}
(\varrho\smallsmile\varsigma)_X=\left\{
\begin{array}{ll}
\varrho_X&\text{if $X\in H$},\\
\varsigma_{X-l}\circ t_{Xl\star}&\text{if $X\in K+l$}
\end{array}
\right.
\end{equation}
with $X\in H\smallsmile K$ for $\varrho\in C(H)$, $\varsigma\in C(K)$.
Further, $\varepsilon\in C(O)$ is given by $\varepsilon=e_\emptyset$.
\end{defi}

\begin{exa} \label{exa:hgcjntex} Calibration monadic product. 
{\rm In ex. \cref{exa:jointhg}, we considered the hypergraphs $H\in G[5]$, $K\in G[4]$ given 
in \ceqref{jointhgex1} and obtained their monoidal product $H\smallsmile K\in G[9]$
exhibited in \ceqref{jointhgex2}.
The hyperedges $X^0,X^1,X^2$ of $H$ were examined in detail earlier in ex. \cref{exa:chgmact}. 
We shall provide their exponent monoids $\msA^{X^0}$, $\msA^{X^1}$, $\msA^{X^2}$
with the convenient indexings  
$\msA^{X^0}=(w^0{}_0,\ldots,w^0{}_3)$, $\msA^{X^1}=(w^1{}_0,\ldots,w^1{}_3)$,
$\msA^{X^2}=(w^2{}_0,\ldots,w^2{}_7)$
displayed in \ceqref{chgmactex1}.
The hyperedges $Y^0,Y^1$ of $K$ were examined in ex. \cref{exa:chgmact} too. We shall likewise 
furnish their exponent monoids $\msA^{Y^0}$, $\msA^{Y^1}$ with 
the indexings $\msA^{Y^0}=(v^0{}_0,\ldots,v^0{}_3)$, $\msA^{Y^1}=(v^1{}_0,\ldots,v^1{}_3)$
shown in \ceqref{chgmactex3}. The hyperedges $Z^0,Z^1,Z^2,Z^3,Z^4$ of $H\smallsmile K$ are to be considered
presently. Convenient indexings
$\msA^{Z^0}=(u^0{}_0,\ldots,u^0{}_3)$, $\msA^{Z^1}=(u^1{}_0,\ldots,u^1{}_3)$,
$\msA^{Z^2}=(u^2{}_0,\ldots,u^2{}_7)$ of the 
exponent monoids $\msA^{Z^0}$, $\msA^{Z^1}$, $\msA^{Z^2}$ are yielded by the right sides
of the \ceqref{chgmactex1} with $X^0,X^1,X^2$ replaced by $Z^0,Z^1,Z^2$, as we have
$Z^0=X^0$, $Z^1=X^1$, $Z^2=X^2$. Suitable indexings of exponent monoids $\msA^{Z^3}$, $\msA^{Z^4}$
are $\msA^{Z^3}=(u^3{}_0,\ldots,u^3{}_3)$, $\msA^{Z^4}=(u^4{}_0,\ldots,u^4{}_3)$,
where $u^3{}_i=t_{Z^35\star}{}^{-1}(v^0{}_i)$, $u^4{}_i=t_{Z^45\star}{}^{-1}(v^1{}_i)$
recalling that $Z^3=Y^0+5$, $Z^4=Y^1+5$. Explicitly, they read as the right hand sides of the
\ceqref{chgmactex3} with $Y^0,Y^1$ replaced by $Z^3,Z^4$.
Suppose again that $\msM=\bbH_{0,3}$, the commutative monoid of the order $3$
cyclic group $\bbZ_3$. Assuming the above indexings, 
the following calibrations $\varrho\in C(H)$, $\varsigma\in C(K)$ can be considered
{\allowdisplaybreaks
\begin{align}
\label{}
&\varrho_{X^0}=(0,1,1,2), &\varrho_{X^1}=(2,0,1,0), &&\varrho_{X^2}=(1,1,2,0,2,0,0,1),
\\
\nonumber
&\varsigma_{Y^0}=(2,1,2,2), &\varsigma_{Y^1}=(1,0,2,1). 
\end{align}
}
\!\!Their monadic product $\varrho\smallsmile\varsigma\in C(H\smallsmile K)$ is 
{\allowdisplaybreaks
\begin{align}
\label{}
&(\varrho\smallsmile\varsigma)_{Z^0}=(0,1,1,2), &(\varrho\smallsmile\varsigma)_{Z^1}=(2,0,1,0),
\\
\nonumber
&(\varrho\smallsmile\varsigma)_{Z^2}=(1,1,2,0,2,0,0,1), &(\varrho\smallsmile\varsigma)_{Z^3}=(2,1,2,2), 
\\
\nonumber
&(\varrho\smallsmile\varsigma)_{Z^4}=(1,0,2,1). &~
\end{align}
}\vspace{-1cm}
}
\end{exa}

The calibration monadic multiplication we have introduced above turns out to be associative and unital. 

\begin{prop} \label{prop:wgsmile}
For $l,m,n\in\bbN$, $H\in G[l]$, $K\in G[m]$, $L\in G[n]$
and $\varrho\in C(H)$, $\varsigma\in C(K)$, $\varphi\in C(L)$, one has 
\begin{equation}
\label{rsmiles2}
\varrho\smallsmile(\varsigma\smallsmile \varphi)=(\varrho \smallsmile \varsigma)\smallsmile \varphi.
\end{equation}
Further, for every $l\in\bbN$, $H\in G[l]$ and $\varrho\in C(H)$, 
\begin{equation}
\label{rsmiles3}
\varrho\smallsmile\varepsilon=\varepsilon\smallsmile \varrho=\varrho.
\end{equation}
\end{prop}

\begin{proof}
We show first \ceqref{rsmiles2}. We have $H\smallsmile K\smallsmile L$ $=H\cup(K+l)\cup(L+l+n)$, 
where $H\cap(K+l)=H\cap(L+l+n)=(K+l)\cap(L+l+n)=\emptyset$, by \ceqref{hgsmile1}. 
From here, using \ceqref{exppf1}, \ceqref{tlw11} and \ceqref{rsmiles1}, it is readily verified that 
for $X\in H\smallsmile K\smallsmile L$ 
\begin{equation}
\label{}
\mhfpt(\varrho\smallsmile(\varsigma\smallsmile \varphi))_X=((\varrho \smallsmile \varsigma)\smallsmile \varphi)_X
=\left\{\begin{array}{ll}
\varrho_X&\text{if $X\in H$},\\
\varsigma_{X-l}\circ t_{Xl\star}&\text{if $X\in K+l$},\\
\varphi_{X-l-m}\circ t_{Xl+m\star}&\text{if $X\in L+l+m$}.
\end{array}
\right.
\end{equation}
This proves \ceqref{rsmiles2}. Next, we show \ceqref{rsmiles3}. Recalling \ceqref{hgsmile2} and
using \ceqref{exppf2} and \ceqref{tlw2}, it is now simple to check that for $X\in H$
\begin{equation}
\label{}
(\varrho\smallsmile\varepsilon)_X=(\varepsilon\smallsmile\varrho)_X=\varrho_X,
\end{equation}
showing \ceqref{rsmiles3}.
\end{proof}

\begin{prop} \label{prop:fsgrss}
Let $l,m,p,q\in\bbN$, $f\in\Hom_\varOmega([l],[p])$, $g\in\Hom_\varOmega([m],[q])$ and let $H\in G[l]$, $K\in G[m]$.
Then, for $\varrho\in C(H)$, $\varsigma\in C(K)$
\begin{equation}
\label{fsgrss}
(f\smallsmile g)_{H\smallsmile K*}(\varrho\smallsmile\varsigma)
=f_{H*}(\varrho)\smallsmile g_{K*}(\varsigma).
\end{equation}
\end{prop}

\begin{proof} To begin with, we verify the consistency of relation \ceqref{fsgrss}. On the left hand side, 
we have that $f\smallsmile g\in\Hom_\varOmega([l]\smallsmile[m],[p]\smallsmile[q])$, $H\smallsmile K\in G([l]\smallsmile[m])$
and $\varrho\smallsmile\varsigma\in C(H\smallsmile K)$ so that  
$(f\smallsmile g)_{H\smallsmile K*}(\varrho\smallsmile\varsigma)\in C(G(f\smallsmile g)(H\smallsmile K))$
in accordance with defs. \cref{def:whgpf} and \cref{def:hgwsm}.
Similarly, we have that $f_{H*}(\varrho)\in C(Gf(H))$, $g_{K*}(\varsigma)\in C(Gg(K))$
so that $f_{H*}(\varrho)\smallsmile g_{K*}(\varsigma)\in C(Gf(H)\smallsmile Gg(K))$ on the right hand side,
again in keeping with defs. \cref{def:whgpf} and \cref{def:hgwsm}. Relation \ceqref{fsgrss} so makes sense, 
as $G(f\smallsmile g)(H\smallsmile K)=Gf\smallsmile Gg(H\smallsmile K)
=Gf(H)\smallsmile Gg(K)$ by virtue of prop. \cref{prop:hgjoint}. 


As $H\smallsmile K=H\cup(K+l)$ with $H\cap(K+l)=\emptyset$, if $X\in H\smallsmile K$, then
either $X\in H$ or $X\in K+l$, but not both. Likewise, 
given that $Gf(H)\smallsmile Gg(K)=Gf(H)\cup(Gg(K)+p)$
with $Gf(H)\cap(Gg(K)+p)=\emptyset$, if $Y\in Gf(H)\smallsmile Gg(K)$, then 
either $Y\in Gf(H)$ or $Y\in Gg(K)+p$, but not both. 

Let $Y\in Gf(H)$. Then, there is $X\in H$ such that $Y=f(X)$ but for no $X\in K+l$ one has
$Y=g(X-l)+p$, else $Y\in Gg(K)+p$. Owing to \ceqref{gamma2}, it follows that 
$X\in H\smallsmile K$ and $f\smallsmile g(X)=Y$ if and only if $X\in H$ and $f(X)=Y$. Consequently,
by virtue of \ceqref{whgpf0} and \ceqref{rsmiles1}
{\allowdisplaybreaks
\begin{align}
\label{}
(f\smallsmile g)_{H\smallsmile K*}(\varrho\smallsmile\varsigma)_Y
&=\mycom{{}_\sss}{{}_{X\in H\smallsmile K,f\smallsmile g(X)=Y}}(f\smallsmile g)|_{X*}((\varrho\smallsmile\varsigma)_X)
\\
\nonumber
&=\mycom{{}_\sss}{{}_{X\in H,f(X)=Y}}f|_{X*}(\varrho_X)=f_{H*}(\varrho)_Y
\end{align}
}
\!when $Y\in Gf(H)$. 

Let  $Y\in Gg(K)+p$. Then, there exists $X\in K+l$ such that $Y=g(X-l)+p$ but for no $X\in H$ one has
$Y=f(X)$, else $Y\in Gf(H)$. By \ceqref{gamma2} again, it follows that 
$X\in H\smallsmile K$ and $f\smallsmile g(X)=Y$ if and only if $X\in K+l$ and $g(X-l)+p=Y$. So,
by virtue of \ceqref{whgpf0} again, when $Y\in Gg(K)+p$ 
{\allowdisplaybreaks
\begin{align}
\label{}
(f\smallsmile g)_{H\smallsmile K*}(\varrho\smallsmile\varsigma)_Y
&=\mycom{{}_\sss}{{}_{X\in H\smallsmile K,f\smallsmile g(X)=Y}}(f\smallsmile g)|_{X*}((\varrho\smallsmile\varsigma)_X)
\\
\nonumber
&=\mycom{{}_\sss}{{}_{X\in K+l,g(X-l)+p=Y}}(t_{g(X-l)+p\hfpt p}{}^{-1}\circ g|_{X-l}\circ t_{Xl})_*(\varsigma_{X-l}\circ t_{Xl\star})
\\
\nonumber
&=\mycom{{}_\sss}{{}_{X+l\in K,g(X-l)+p=Y}}t_{Y p}{}^{-1}{}_*\circ g|_{X-l*}\circ t_{X\hfpt l*}\circ
t_{X\hfpt l}{}^{-1}{}_*(\varsigma_{X-l})
\\
\nonumber
&=\mycom{{}_\sss}{{}_{X\in K,g(X)=Y-p}}g|_{X*}(\varsigma_X)\circ t_{Y p*}      
=g_{K*}(\varsigma)_{Y-p}\circ t_{Y p\star}.
\end{align}
}
\!\!In the second line above, we expressed $(\varrho\smallsmile\varsigma)_X$ through \ceqref{rsmiles1} 
and used the identity $(f\smallsmile g)|_X=t_{g(X-l)+p\hfpt p}{}^{-1}\hfpt\circ\hfpt g|_{X-l}\hfpt\circ\hfpt t_{Xl}$  
ensuing from \ceqref{gamma2}. In the third line, we used the relation 
$\varsigma_{X-l}\circ t_{Xl\star}=t_{Xl}{}^{-1}{}_*(\varsigma_{X-l})$, which can be shown easily from 
\ceqref{wgtpf0}, and applied \ceqref{wgtpf1}. In the fourth line, we employed the relation 
$g|_{X*}(\varsigma_X)\circ t_{Y p*}=t_{Y p}{}^{-1}{}_*\circ g|_{X*}(\varsigma_X)$ of the same form as the one
used in the previous line. Finally, we used \ceqref{whgpf0} again. 


By the above calculations, owing to \ceqref{rsmiles1}, we conclude that
\begin{equation}
\label{}
(f\smallsmile g)_{H\smallsmile K*}(\varrho\smallsmile\varsigma)_Y=f_{H*}(\varrho)\smallsmile g_{K*}(\varsigma)_Y
\end{equation}
for $Y\in Gf(H)\smallsmile Gg(K)$, showing \ceqref{fsgrss}. 
\end{proof}

By these results, the monadic multiplicative structure of the hypergraph $\varOmega$ monad $G\varOmega$
(cf. subsect. \cref{subsec:graphctg}, def. \cref{def:hgadj}) extends to the calibrated hypergraph
category $G_C\varOmega$.


\begin{defi} \label{def:whgadj}
For any two $l,m\in\bbN$, the calibrated hypergraph monadic multiplication 
$\smallsmile:G_C[l]\times G_C[m]\rightarrow G_C([l]\smallsmile[m])$, reads for $(H,\varrho)\in G_C[l]$,
$(K,\varsigma)\in G_C[m]$ as
\begin{equation}
\label{whgsmile0}  
(H,\varrho)\smallsmile(K,\varsigma)=(H\smallsmile K,\varrho\smallsmile\varsigma).
\end{equation}
$(O,\varepsilon)\in G_C[0]$ is the calibrated hypergraph
monadic unit. 
\end{defi}

\noindent

\noindent 
That the above defines indeed a genuine graded $\varOmega$ monadic structure
on the category $G_C\varOmega$
is established by the following proposition (cf. def. \cref{def:jointstrct}).

\begin{prop} \label{prop:whgjoint}
The calibrated hypergraph category $G_C\varOmega$ equipped with the monadic multiplication $\smallsmile$
and unit $(O,\varepsilon)$ is a graded $\varOmega$ monad. 
\end{prop}

\begin{proof} 
From the defining relation \ceqref{whgsmile0}, using the identities \ceqref{hgsmile1}, 
and \ceqref{rsmiles2}, it is immediately verified that 
for all $l,m,n\in\bbN$, $(H,\varrho)\in G_C[l]$, $(K,\varsigma)\in G_C[m]$, $(L,\varphi)\in G_C[n]$ we have 
\begin{equation}
\label{whgsmile1}
(H,\varrho)\smallsmile((K,\varsigma)\smallsmile (L,\varphi))=((H,\varrho) \smallsmile (K,\varsigma))\smallsmile (L,\varphi).
\end{equation}
The monadic multiplication $\smallsmile$ thus satisfies relation \ceqref{joint1}.
Further, employing \ceqref{hgsmile2} and \ceqref{rsmiles3}, it is readily found that 
for every $l\in\bbN$, $(H,\varrho)\in G[l]$ it holds that 
\begin{equation}
\label{whgsmile2}
(H,\varrho)\smallsmile(O,\varepsilon)=(O,\varepsilon)\smallsmile(H,\varrho)=(H,\varrho).
\end{equation}
$\smallsmile$ so satisfies also relation \ceqref{joint2}. There remains to be shown that 
for all $l,m,p,q\in\bbN$, $f\in\Hom_\varOmega([l],[p])$, $g\in\Hom_\varOmega([m],[q])$ and $(H,\varrho)\in G_C[l]$,
$(K,\varsigma)\in G_C[m]$
\begin{equation}
\label{whgsmile3}
G_Cf\smallsmile G_Cg((H,\varrho)\smallsmile (K,\varsigma))=G_Cf(H,\varrho)\smallsmile Gg(K,\varsigma)
\end{equation}
and so prove that relation \ceqref{joint3} is satisfied too. 
The verification relies on a straightforward combination of the identities
\ceqref{hgsmile3}, \ceqref{wgfundt2}, \ceqref{fsgrss} and \ceqref{whgsmile0}, 
{\allowdisplaybreaks
\begin{align}
\label{}
G_Cf\smallsmile G_Cg((H,\varrho)\smallsmile (K,\varsigma))&=G_C(f\smallsmile g)(H\smallsmile K,\varrho\smallsmile\varsigma)
\\
\nonumber
&=(G(f\smallsmile g)(H\smallsmile K), (f\smallsmile g)_{H\smallsmile K*}(\varrho\smallsmile\varsigma))
\\
\nonumber
&=(Gf\smallsmile Gg(H\smallsmile K),f_{H*}(\varrho)\smallsmile g_{K*}(\varsigma))
\\
\nonumber
&=(Gf(H)\smallsmile Gg(K),f_{H*}(\varrho)\smallsmile g_{K*}(\varsigma))
\\
\nonumber
&=(Gf(H),f_{H*}(\varrho))\smallsmile(Gg(K),g_{K*}(\varsigma))
\\
\nonumber
&\hspace{5cm}=G_Cf(H,\varrho)\smallsmile Gg(K,\varsigma),
\end{align}
}
\!\!where the functorial relations $G(f\smallsmile g)=Gf\smallsmile Gg$ and 
$G_C(f\smallsmile g)=G_Cf\smallsmile G_Cg$ were used.
\end{proof}

Since the calibrated hypergraph $\varOmega$ monad $G_C\varOmega$ has been fashioned as an enhancement 
of the hypergraph $\varOmega$ monad $G\varOmega$ \pagebreak (cf. prop. \cref{prop:hgjoint}), the natural question arises
about whether their relationship can be described through an $\varOmega$ monad morphism.
There exists indeed a special morphism of the $\varOmega$ monads
$G_C\varOmega$, $G\varOmega$ projecting the first onto the second
which formalizes the above intuition. 

\begin{prop} \label{prop:prjc}
The assignment of the projection $\sfp_C:G_C[l]\rightarrow G[l]$ given by 
\begin{equation}
\label{spcfmor}
\sfp_C(H,\varrho)=H
\end{equation}
with $(H,\varrho)\in G_C[l]$ to each $l\in\bbN$
defines a  morphism $\sfp_C\in\Hom_{\ul{\rm GM}_\varOmega}(G_C\varOmega,G\varOmega)$
of the objects $G_C\varOmega$, $G\varOmega\in\Obj_{\ul{\rm GM}_\varOmega}$ in $\ul{\rm GM}_\varOmega$.
\end{prop}

\begin{proof}
We have to check that $\sfp_C$ meets the requirements \ceqref{catjoint1}--\ceqref{catjoint3}.
From \ceqref{spcfmor} and \ceqref{wgfundt2}, we find
\begin{equation}
\label{}
\sfp_C\circ G_Cf(H,\varrho)=Gf\circ \sfp_C(H,\varrho)=Gf(H),
\end{equation}
where $f\in\Hom_\varOmega([l],[m])$, $(H,\varrho)\in G_C[l]$ with $l,m\in\bbN$. So, $\sfp_C$ satisfies \ceqref{catjoint1}.
Next, from \ceqref{spcfmor} and \ceqref{whgsmile0},  we get further 
\begin{equation}
\label{}
\sfp_C((H,\varrho)\smallsmile(K,\varsigma))=\sfp_C(H,\varrho)\smallsmile\sfp_C(K,\varsigma)=H\smallsmile K,
\end{equation}
where $(H,\varrho)\in G_C[l]$, $(K,\varsigma)\in G_C[m]$ with $l,m\in\bbN$.
In this manner, $\sfp_C$ satisfies also \ceqref{catjoint2}. Finally, since \hphantom{xxxxxxxxxxx}
\begin{equation}
\label{}
\sfp_C(O,\varepsilon)=O
\end{equation}
by \ceqref{spcfmor}, \ceqref{catjoint3} holds too. 
\end{proof}

\noindent

\noindent
The elaboration of the formal calibrated hypergraph $\varOmega$ monadic framework is with this complete.

\vspace{1mm}


\subsection{\textcolor{blue}{\sffamily The weighted hypergraph $\varOmega$ monad $G_W\varOmega$}}\label{subsec:hemult}

Hypergraph calibrations (cf. def. \cref{def:hygrpwgt}), quantifying the prevalence of exponent functions,
represent a refinement of hypergraph weights, describing quantitatively hyperedge occurrence. 
Just as hypergraph calibrations are constitutive elements of the calibrated hypergraph $\varOmega$ monad $G_C\varOmega$
studied in subsect. \cref{subsec:wggraphctg}, hypergraph weights are basic components of the weighted hypergraph
$\varOmega$ monad $G_W\varOmega$ whose construction is illustrated in the present subsection. 
Elucidating the relationship occurring between $G_C\varOmega$ and $G_W\varOmega$
will help putting $G_C\varOmega$ in an appropriate perspective on one hand \linebreak and afford clarifying the connection
of calibrated and weighted hypergraph states in sect. 3 of II on the other.  
Given the formal similarities to 
the construction of $G_C\varOmega$
carried out in subsect. \cref{subsec:wggraphctg}, we shall limit ourselves to stating the basic definitions and results
leaving the details of the proofs to the reader.

The basic datum of the 
weighted hypergraph category $G_W\varOmega$ is 
a finite commutative monoid $\msM$, the same as the one used for calibrated hypergraph category $G_C\varOmega$.

\begin{defi}
Let $l\in\bbN$ and let $H\in G[l]$ be a hypergraph. The weight monoid of $H$ is the commutative monoid
$W(H)=\msM^H$ of all $\msM$--valued functions on $H$ under pointwise addition.
\end{defi}

\noindent
If $H=\emptyset$, then $W(H)=0$ is the trivial commutative monoid consisting
of the neutral element $0_\emptyset=0$ only. The elements $\alpha\in W(H)$ are called the 
weight functions of $H$.


The push--forward operation we defined for hypergraph calibrations (cf. def. \cref{def:whgpf})
has an analog for hypergraph weight functions. 

\begin{defi} \label{def:hgmpf}
Let $l,m\in\bbN$ and $f\in\Hom_\varOmega([l],[m])$. Let $H\in G[l]$ be a hypergraph. The hypergraph
weight push--forward map $f_{H*}:W(H)\rightarrow W(Gf(H))$ is given by 
\begin{equation}
\label{mhgpf0}
f_{H*}(\alpha)_Y=\mycom{{}_\sss}{{}_{X\in H,f(X)=Y}}\alpha_X \vphantom{\Big]^f_g}
\end{equation}
with $Y\in Gf(H)$ for $\alpha\in W(H)$.
\end{defi}

\noindent
If $H=\emptyset$ and so $\alpha=0_\emptyset$, then $Gf(H)=\emptyset$ and so $f_{H*}(\alpha)=0_\emptyset$.

\begin{exa} \label{exa:whgpfac} Weighted hypergraph push--forward action. 
{\rm Assume that $\msM=\bbH_{0,3}$, the commutative monoid underlying the order $3$ cyclic group $\bbZ_3$.
Consider the hypergraph $H\in G[5]$ shown in \ceqref{moracth2}, $H=\{X^0,X^1,X^2\}$. A weight function 
$\alpha\in W(H)$ is e.g. 
\begin{align}
\label{}
&\alpha_{X^0}=2, &\alpha_{X^1}=1, &&\alpha_{X^2}=2. 
\label{}
\end{align}
Consider next the morphism $f\in\Hom_\varOmega([5],[4])$ given by eq. \ceqref{moracth1}.
The transformed hypergraph $Gf(H)\in G[4]$ is shown in \ceqref{moracth3}, $Gf(H)=\{Y^0,Y^1\}$. \pagebreak 
By \ceqref{mhgpf0}, the push--forward $f_{H*}(\alpha)$ of the weight function $\alpha$ reads as 
\begin{align}
&f_{H*}(\alpha)_{Y^0}=0, \hspace{-3cm}&f_{H*}(\alpha)_{Y^1}=2. 
\label{}
\end{align}
}
\end{exa}

The weight push--forward operation enjoys the standard push--forward properties which
the calibration push--forward operation does (cf. prop. \cref{prop:whgpfprop}). 

\begin{prop} \label{prop:mhgpfprop}
Assume that $l,m,n\in\bbN$, $f\in\Hom_\varOmega([l],[m])$, $g\in\Hom_\varOmega([m],[n])$. Let further $H\in G[l]$. Then,
one has 
\begin{equation}
\label{mhgpf1}
(g\circ f)_{H*}=g_{Gf(H)*}\circ f_{H*}. 
\end{equation}
Let $l\in\bbN$. Let further $H\in G[l]$. Then, one has 
\begin{equation}
\label{mhgpf2}
\id_{[l]H*}=\id_{W(H)}.
\end{equation}
\end{prop}

\begin{proof} The proof is analogous to 
  that of prop. \cref{prop:whgpfprop} and is left to the reader.
\end{proof}

Proceeding analogously to the construction of the calibrated hypergraph category
$G_C\varOmega$ in subsect. \cref{subsec:wggraphctg}, we introduce a new $\varOmega$ category
enhancing the hypergraph category $G\varOmega$ studied in subsect. \cref{subsec:graphctg} through the addition
of hyperedge weight data. 

\begin{prop} \label{prop:mhffun} 
The prescription assigning  

\begin{enumerate}[label=\alph*)] 

\item an object $G_W[l]\in\Obj_\varSigma$ given by 
\begin{equation}
\label{mgfundt1}
G_W[l]=\{(H,\alpha)|H\in G[l],\alpha\in W(H)\}
\end{equation}
for every $l\in\bbN$ and 

\item a morphism $G_Wf\in\Hom_\varSigma(G_W[l],G_W[m])$ given by 
\begin{equation}
\label{mgfundt2}
G_Wf(H,\alpha)=(Gf(H),f_{H*}(\alpha))
\end{equation}
with $(H,\alpha)\in G_W[l]$ for every $l,m\in\bbN$ and morphism $f\in\Hom_\varOmega([l],[m])$.

\end{enumerate}

\noindent
defines a functor $G_W:\varOmega\rightarrow\varSigma$ injective on objects and morphisms.
There is in this way an $\varOmega$ category $G_W\varOmega$ having the functor $G_W$  
with the target category restricted from $\varSigma$ to $G_W\varOmega$ as its stalk isofunctor.
$G_W\varOmega$ is further a non monoidal subcategory of $\varSigma$.  
\end{prop}

\begin{proof} The proof proceeds along lines analogous to those of prop. \cref{prop:gwfinj}
and is left to the reader. 
\end{proof}

\noindent
We shall refer to the $\varOmega$ category $G_W\varOmega$ and its stalk isofunctor $G_W$ as the weighted hypergraph
category and functor, respectively. Items \cref{d1}--\cref{d7} of prop. \cref{prop:domegacat}, with $D$ replaced by $G_W$, 
describe the layout  of $G_W\varOmega$. We notice that since the $G_W\varOmega$ is a subcategory of the finite set category
$\varSigma$, the objects and morphisms of $G_W\varOmega$ are again genuine sets and set 
functions respectively and the composition law and the identity assigning map of
$G_W\varOmega$ are once more the usual set theoretic ones. In particular, we have 
that  $G_W[0]=\{(\emptyset,0_\emptyset)\}$ and that 
$G_We_{[m]}(\emptyset,0_\emptyset)=(\emptyset,0_\emptyset)\in G_W[m]$
for $e_{[m]}\in\Hom_\varOmega([0],[m])$,.

A notion of monadic multiplication for hypergraph weight functions is available. 

\begin{defi} \label{def:mhgwsm}
For all $l,m\in\bbN$ and $H\in G[l]$, $K\in G[m]$, the hypergraph weight monadic multiplication
$\smallsmile:W(H)\times W(K)\rightarrow W(H\smallsmile K)$ (cf. def. \cref{def:hgadj}) is defined by 
\begin{equation}
\label{mrsmiles1}
(\alpha\smallsmile\beta)_X=\left\{
\begin{array}{ll}
\alpha_X&\text{if $X\in H$},\\
\beta_{X-l} &\text{if $X\in K+l$}
\end{array}
\right.
\end{equation}
with $X\in H\smallsmile K$ for $\alpha\in W(H)$, $\beta\in W(K)$.
Further, $\upsilon\in W(O)$ is given by $\upsilon=0_\emptyset$.
\end{defi}

\begin{exa} \label{exa:hgwjnt} Hypergraph weight function monadic product. 
{\rm In ex. \cref{exa:jointhg}, we considered the hypergraphs $H\in G[5]$, $K\in G[4]$
and determined their monoidal product $H\smallsmile K\in G[9]$.
The hyperedges $X^0,X^1,X^2$, $Y^0,Y^1$ of $H$, $K$ are displayed in \ceqref{jointhgex1};
the hyperedges $Z^0,Z^1,Z^2,Z^3,Z^4$ of $H\smallsmile K$ are exhibited
in \ceqref{jointhgex2}. Suppose again that $\msM=\bbH_{0,3}$, the commutative monoid of the 
cyclic group $\bbZ_3$ already met in ex \cref{exa:whgpfac}.
Weight functions $\alpha\in W(H)$, $\beta\in W(K)$ are e.g. 
{\allowdisplaybreaks
\begin{align}
\label{}
&\alpha_{X^0}=2, &\alpha_{X^1}=0, &&\alpha_{X^2}=1,
\\
\nonumber
&\beta_{Y^0}=0, &\beta_{Y^1}=1.
\end{align}
}
\!\!A simple application of formula \ceqref{mrsmiles1}  shows that their joint $\alpha\smallsmile\beta\in W(H\smallsmile K)$ is
{\allowdisplaybreaks
\begin{align}
\label{}
&(\alpha\smallsmile\beta)_{Z^0}=2, &(\alpha\smallsmile\beta)_{Z^1}=0, &&(\alpha\smallsmile\beta)_{Z^2}=1,
\\
\nonumber
&(\alpha\smallsmile\beta)_{Z^3}=0, &(\alpha\smallsmile\beta)_{Z^4}=1.
\end{align}
}
\vspace{-.8cm}
}
\end{exa}


The weight monadic multiplication we have introduced above turns out to be
associative and unital as its calibration counterpart (cf. prop. \cref{prop:wgsmile}).

\begin{prop} \label{prop:mwgsmile}
For $l,m,n\in\bbN$, $H\in G[l]$, $K\in G[m]$, $L\in G[n]$
and $\alpha\in W(H)$, $\beta\in W(K)$, $\gamma\in \msM^L$, one has 
\begin{equation}
\label{mrsmiles2}
\alpha\smallsmile(\beta\smallsmile \gamma)=(\alpha \smallsmile \beta)\smallsmile \gamma.
\end{equation}
Further, for every $l\in\bbN$, $H\in G[l]$ and $\alpha\in W(H)$, 
\begin{equation}
\label{mrsmiles3}
\alpha\smallsmile\upsilon=\upsilon\smallsmile \alpha=\alpha.
\end{equation}
\end{prop}

\begin{proof}
The proof, which follows lines analogous to those of prop. \cref{prop:wgsmile} and is
straightforward, is left to the reader
\end{proof}

\begin{prop}
Let $l,m,p,q\in\bbN$, $f\in\Hom_\varOmega([l],[p])$, $g\in\Hom_\varOmega([m],[q])$ and let $H\in G[l]$, $K\in G[m]$.
Then, for $\alpha\in W(H)$, $\beta\in W(K)$
\begin{equation}
\label{mfsgrss}
(f\smallsmile g)_{H\smallsmile K*}(\alpha\smallsmile\beta)=f_{H*}(\alpha)\smallsmile g_{K*}(\beta).
\end{equation}
\end{prop}

\begin{proof}
This result is analogous to prop. \cref{prop:fsgrss} and so is its proof. 
\end{proof}

\begin{defi} \label{def:mhgadj}
For any pair $l,m\in\bbN$, the weighted hypergraph monadic multiplication 
$\smallsmile:G_W[l]\times G_W[m]\rightarrow G_W([l]\smallsmile[m])$, reads
for $(H,\alpha)\in G_W[l]$, $(K,\beta)\in G_W[m]$ as
\begin{equation}
\label{mhgsmile0}  
(H,\alpha)\smallsmile(K,\beta)=(H\smallsmile K,\alpha\smallsmile\beta).
\end{equation}
Moreover, $(O,\upsilon)\in G_W[0]$ is the weighted hypergraph monadic unit. 
\end{defi}

\noindent
That the above defines indeed an authentic graded $\varOmega$ monadic multiplicative structure
on the category $G_W\varOmega$ is established by the following proposition (cf. def. \cref{def:jointstrct}).


\begin{prop} \label{prop:mhgjoint}
The weighted hypergraph category $G_W\varOmega$ endowed with the monadic multiplication $\smallsmile$
and unit $(O,\upsilon)$ is a graded $\varOmega$ monad. 
\end{prop}

\begin{proof}
This result is completely analogous to prop. \cref{prop:whgjoint} and is demonstrated in a similar fashion. 
\end{proof}

The calibrated hypergraph category \pagebreak $G_C\varOmega$ and the weighted hypergraph category
$G_W\varOmega$, therefore, share the property of being a graded $\varOmega$ monad
(cf. prop. \cref{prop:whgjoint}). There exists a special morphism of the monads
$G_C\varOmega$, $G\varOmega$ (cf. def. \cref{def:catjoint}). Its construction, however, is not straightforward and requires 
the introduction of the notion of hypergraph weight associated with a hypergraph calibration.

\begin{defi} \label{def:hemult}
Let $l\in\bbN$ and let $(H,\varrho)\in G_C[l]$ be a calibrated hypergraph.
The weight function of $(H,\varrho)$ is the map 
$\mu_{(H,\varrho)}\in W(H)$  
defined by 
\begin{equation}
\label{hemult0}
\mu_{(H,\varrho)X}=\mycom{{}_\sss}{{}_{w\in\msA^X}}\varrho_X(w)
\end{equation}
with $X\in H$. 
\end{defi}

\noindent For the empty calibrated hypergraph $(O,\varepsilon)\in G_W[0]$, we have 
that $\mu_{(O,\varepsilon)}=0_\emptyset$, as suggested by intuition. 

\begin{exa} Weight function of a calibrated hypergraph. 
{\rm The setting assumed here is taken from ex. \cref{exa:chgmact}. 
The basic commutative monoids $\msA$, $\msM$ involved are $\msA=\bbH_{1,1}$, $\msM=\bbH_{0,3}$.
The hypergraph $H\in G[5]$ considered is the one shown shown in \ceqref{moracth2} with hyperedges $X^0,X^1,X^2$.
We use the indexings of the exponent monoids $\msA^{X^0}$, $\msA^{X^1}$, $\msA^{X^2}$
of $X^0,X^1,X^2$ also introduced in ex. \cref{exa:chgmact}. 
A calibration $\varrho\in C(H)$ is e.g. 
\begin{align}
\label{}
&\varrho_{X^0}=(0,1,1,2), &\varrho_{X^1}=(2,0,1,0), &&\varrho_{X^2}=(1,1,2,0,2,0,0,1). 
\end{align}
With this the hypergraph $H$ is promoted to a calibrated hypergraph $(H,\varrho)\in G_W[5]$. 
Using formula \ceqref{hemult0}, the associated weight function $\mu_{(H,\varrho)}\in W(H)$ is 
\begin{align}
\label{}
&\mu_{(H,\varrho)X^0}=1, &\mu_{(H,\varrho)X^0}=0, &&\mu_{(H,\varrho)X^0}=1.
\end{align}
}
\end{exa}

The calibrated hypergraph weight function $\mu$ enjoys natural properties. First, $\mu$ 
behaves covariantly under the morphism action of the category $G_C\varOmega$. 

\begin{prop}
Let $l,m\in\bbN$ and $f\in\Hom_\varOmega([l],[m])$. Let further $(H,\varrho)\in G_C[l]$ be a calibrated hypergraph. Then, 
\begin{equation}
\label{hemult1}
f_{H*}(\mu_{(H,\varrho)})=\mu_{G_Cf(H,\varrho)}.
\end{equation}
\end{prop}

\noindent Above, $G_Cf(H,\varrho)$ is given by \ceqref{wgfundt2}. 

\begin{proof} From \ceqref{hemult0}, employing \ceqref{wgtpf0}, \ceqref{whgpf0}, \ceqref{wgfundt2}, we compute
\vspace{-1mm}
{\allowdisplaybreaks
\begin{align}
\label{}
\mu_{G_Cf(H,\varrho)Y}&=\mu_{(Gf(H),f_{H*}(\varrho))Y}
\\
\nonumber
&=\mycom{{}_\sss}{{}_{v\in\msA^Y}}f_{H*}(\varrho)_Y(v)
\\
\nonumber
&=\mycom{{}_\sss}{{}_{v\in\msA^Y}}\mycom{{}_\sss}{{}_{X\in H,f(X)=Y}}f|_{X*}(\varrho_X)(v)
\\
\nonumber
&=\mycom{{}_\sss}{{}_{v\in\msA^Y}}\mycom{{}_\sss}{{}_{X\in H,f(X)=Y}}\mycom{{}_\sss}{{}_{w\in\msA^X,f|_{X\star}(w)=v}}\varrho_X(w)
\\
\nonumber
&=\mycom{{}_\sss}{{}_{v\in\msA^Y}}\mycom{{}_\sss}{{}_{X\in H,f(X)=Y}}\mycom{{}_\sss}{{}_{w\in\msA^X}}
\delta_{f|_{X\star}(w),v}\varrho_X(w)
\\
\nonumber
&=\mycom{{}_\sss}{{}_{X\in H,f(X)=Y}}\mycom{{}_\sss}{{}_{w\in\msA^X}}\varrho_X(w)
\\
\nonumber
&=\mycom{{}_\sss}{{}_{X\in H,f(X)=Y}}\mu_{(H,\varrho)X}=f_{H*}(\mu_{(H,\varrho)})_Y
\end{align}
}
\!\!with $Y\in Gf(H)$, confirming \ceqref{hemult1}. 
\end{proof}

\noindent Second, the weight function $\mu$ behaves compatibly with respect to the monadic multiplications of the
$\varOmega$ monads $G_C\varOmega$, $G_W\varOmega$.


\begin{prop}
Let $l,m\in\bbN$ and let $(H,\varrho)\in G_C[l]$, $(K,\varsigma)\in G_C[m]$ be calibrated hypergraphs.
Then, it holds that \hphantom{xxxxxxxxxxxxx}
\begin{equation}
\label{hemult2}
\mu_{(H,\varrho)\smallsmile(K,\varsigma)}=\mu_{(H,\varrho)}\smallsmile\mu_{(K,\varsigma)}.
\end{equation}
Furthermore, \hphantom{xxxxxxxxxxxxxxxxxxxx}
\begin{equation}
\label{hemult3}
\mu_{(O,\varepsilon)}=\upsilon. 
\end{equation}
\end{prop}

\noindent We recall here that $(H,\varrho)\smallsmile(K,\varsigma)\in G_C([l]\smallsmile[m])$ is
given by \ceqref{whgsmile0} and that $(O,\varepsilon)\in G_C[0]$ is the empty calibrated hypergraph. 

\begin{proof}
From \ceqref{hemult0}, using \ceqref{whgsmile0}  
\begin{equation}
\label{hemult2p1}
\mu_{(H,\varrho)\smallsmile(K,\varsigma)X}=\mu_{(H\smallsmile K,\varrho\smallsmile\varsigma)X}
=\mycom{{}_\sss}{{}_{w\in\msA^X}}(\varrho\smallsmile\varsigma)_X(w)
\end{equation}
with $X\in H\smallsmile K$, \pagebreak where $(\varrho\smallsmile\varsigma)_X$ is given by \ceqref{rsmiles1}.
As $H\smallsmile K=H\cup(K+l)$ with $H\cap(K+l)=\emptyset$, either $X\in H$ or $X\in K+l$. By \ceqref{rsmiles1},
if $X\in H$ we have
\begin{equation}
\label{hemult2p2}
\mycom{{}_\sss}{{}_{w\in\msA^X}}(\varrho\smallsmile\varsigma)_X(w)
=\mycom{{}_\sss}{{}_{w\in\msA^X}}\varrho_X(w)=\mu_{(H,\varrho)X}. 
\end{equation}
By \ceqref{rsmiles1} again, if instead $X\in K+l$, we have 
{\allowdisplaybreaks
\begin{align}
\label{hemult2p3}
\mycom{{}_\sss}{{}_{w\in\msA^X}}(\varrho\smallsmile\varsigma)_X(w)
&=\mycom{{}_\sss}{{}_{w\in\msA^X}}\varsigma_{X-l}(t_{Xl\star}(w))
\\
\nonumber
&=\mycom{{}_\sss}{{}_{w\in\msA^{X-l}}}\varsigma_{X-l}(w)=\mu_{(K,\varsigma)X-l}, 
\end{align}
}
\!\!where we used the bijectivity of the map $t_{Xl\star}:A^X\rightarrow A^{X-l}$. From \ceqref{hemult2p1}--\ceqref{hemult2p3}, 
\begin{equation}
\label{hemult2p4}
\mu_{(H,\varrho)\smallsmile(K,\varsigma)X}=(\mu_{(H,\varrho)}\smallsmile\mu_{(K,\varsigma)})_X
\end{equation}
with $X\in H\smallsmile K$, yielding \ceqref{hemult2}. 

\ceqref{hemult3} follows immediately from the remark below def. \cref{def:hemult} upon recalling that
$(O,\varepsilon)=(\emptyset,e_\emptyset)$ and that $\upsilon=0_\emptyset$.
\end{proof}



The morphisms relating the calibrated to the weighted hypergraph $\varOmega$ monads $G_C\varOmega$
and $G_W\varOmega$ mentioned earlier can now be displayed. 

\begin{prop} \label{prop:hgmm}
The assignment of the function $\sfh:G_C[l]\rightarrow G_W[l]$ given by 
\begin{equation}
\label{sffmor}
\sfh(H,\varrho)=(H,\mu_{(H,\varrho)}) 
\end{equation}
with $(H,\varrho)\in G_C[l]$ to each $l\in\bbN$ 
defines a special  morphism $\sfh\in\Hom_{\ul{\rm GM}_\varOmega}(G_C\varOmega,G_W\varOmega)$
of the objects $G_C\varOmega$, $G_W\varOmega\in\Obj_{\ul{\rm GM}_\varOmega}$ in $\ul{\rm GM}_\varOmega$.
\end{prop}

\begin{proof}
We have to check that $\sfh$ meets the requirements \ceqref{catjoint1}--\ceqref{catjoint3}.
From \ceqref{sffmor}, using \ceqref{wgfundt2}, \ceqref{mgfundt2} and \ceqref{hemult1}, we find
{\allowdisplaybreaks
\begin{align}
\label{}
G_Wf\circ\sfh(H,\varrho)&=G_Wf(H,\mu_{(H,\varrho)})
\\
\nonumber
&=(Gf(H),f_{H*}(\mu_{(H,\varrho)}))
\\
\nonumber
&=(Gf(H),\mu_{G_Cf(H,\varrho)})
\\
\nonumber
&=(Gf(H),\mu_{(Gf(H),f_{H*}(\varrho))})
\\
\nonumber
&=\sfh(Gf(H),f_{H*}(\varrho))=\sfh\circ G_Cf(H,\varrho)
\end{align}
}
\!\!where $f\in\Hom_\varOmega([l],[m])$, $(H,\varrho)\in G_C[l]$ with $l,m\in\bbN$. So, $\sfh$ satisfies \ceqref{catjoint1}.
Next, from \ceqref{sffmor}, using \ceqref{whgsmile0}, \ceqref{mhgsmile0} and \ceqref{hemult2}, we get 
{\allowdisplaybreaks
\begin{align}
\label{}
\sfh((H,\varrho)\smallsmile(K,\varsigma))&=\sfh(H\smallsmile K,\varrho\smallsmile\varsigma)
\\
\nonumber
&=(H\smallsmile K,\mu_{(H\smallsmile K,\varrho\smallsmile\varsigma)})
\\
\nonumber
&=(H\smallsmile K,\mu_{(H,\varrho)}\smallsmile\mu_{(K,\varsigma)})
\\
\nonumber
&=(H,\mu_{(H,\varrho)})\smallsmile(K,\mu_{(K,\varsigma)})=\sfh(H,\varrho)\smallsmile\sfh(K,\varsigma).
\end{align}
}
\!\!where $(H,\varrho)\in G_C[l]$, $(K,\varsigma)\in G_C[m]$ with $l,m\in\bbN$. Therefore, $\sfh$ obeys \ceqref{catjoint2}.
Finally, from \ceqref{sffmor}, using \eqref{hemult3}, we have 
\begin{equation}
\label{}
\sfh(O,\varepsilon)=(O,\mu_{(O,\varepsilon)})=(O,\upsilon).
\end{equation}
So, $\sfh$ satisfies also \ceqref{catjoint3}.
Owing to def. \cref{def:catjoint}, $\sfh\in\Hom_{\ul{\rm GM}_\varOmega}(G_C\varOmega,G_W\varOmega)$. 
\end{proof}

There exists also a morphism of the weighted and bare hypergraph $\varOmega$ monads
$G_W\varOmega$ and $G\varOmega$ projecting the first onto the second
analogous to the similar morphism projecting the calibrated hypergraph $\varOmega$ monad $G_C\varOmega$ onto $G\varOmega$ 
(cf. prop. \cref{prop:prjc}). 

\begin{prop} \label{prop:prjw}
The assignment of the projection $\sfp_W:G_W[l]\rightarrow G[l]$ given by 
\begin{equation}
\label{spwfmor}
\sfp_W(H,\alpha)=H
\end{equation}
with $(H,\alpha)\in G_W[l]$ to each $l\in\bbN$
defines a morphism $\sfp_W\in\Hom_{\ul{\rm GM}_\varOmega}(G_W\varOmega,G\varOmega)$
of the objects $G_W\varOmega$, $G\varOmega\in\Obj_{\ul{\rm GM}_\varOmega}$ in $\ul{\rm GM}_\varOmega$.
\end{prop}

\begin{proof}
The proof is completely analogous to that of prop. \cref{prop:prjc}. 
\end{proof}

\vfill\eject

\renewcommand{\sectionmark}[1]{\markright{\thesection\ ~~#1}}

\section{\textcolor{blue}{\sffamily The multi dit mode $\varOmega$ monads}}\label{sec:grsttsctg}


Dits are generalizations of bits with an arbitrary finite number of modes. 
Their classical and quantum incarnations go under the names of cdits and qudits, respectively.

Multi cdit systems are characterized by their configurations; multi qudit systems are so 
by their states. Multi cdit configurations and qudit states can be organized in graded
$\varOmega$ monads, which we have studied in  subsect. \cref{subsec:omegacat}. 

Since multi cdit configurations and qubit states are related by basis encoding, the $\varOmega$ monads describing them
also are. We have found conceptually appropriate to treat separately the
the multi cdit configuration $\varOmega$ monad in subsect. \cref{subsec:grpct}, though this by itself
is not directly relevant for the theory worked out in sect. 3 of II, 
and then construct the important multi qudit state $\varOmega$ monad in subsect. \cref{subsec:graphhilb} via base encoding.
One could however directly define the latter without any reference to the former.
In subsect. \cref{subsec:mqsrmk}, we report some useful results concerning the 
operatorial structure of the multi qudit state $\varOmega$ monad for later use.
Some explicit examples are furnished. Others will be provided in sect. 3 of II 
in relation to the construction of hypergraph states.



\subsection{\textcolor{blue}{\sffamily The multi cdit configuration $\varOmega$ monad $E\varOmega$}}\label{subsec:grpct}

The multi cdit configuration monad $E\varOmega$ is a graded $\varOmega$ monad describing the configurations of a multi cdit
classical register, as suggested by its name. In this subsection, we shall delineate in detail
its construction and study its main properties. 



The states of a cdit constitute a finite ring $\msR$. Only the additive monoid underlying
$\msR$ enters the definition of the configuration category $E\varOmega$. Below, we therefore assume
only that $\msR$ is an additive finite commutative monoid with $\msR\neq 0$, where $0$ denotes the trivial monoid.

The construction of the $\varOmega$ category $E\varOmega$ is carried out along lines analogous to those
of the hypergraph categories studied in sect. \cref{sec:graphctg}. 

\begin{prop} \label{prop:qfnct}
The rule specifying 

\begin{enumerate}[label=\alph*)] 

\item an object $E[l]\in\Obj_\varSigma$ given by \hphantom{xxxxxxxxxxxxxxxx}
\begin{equation}
\label{qfundt1}
E[l]=\msR^l
\end{equation}
for every $l\in\bbN$ and 

\item a morphism $Ef\in\Hom_\varSigma(E[l],E[m])$ given by 
\begin{equation}
\label{qfundt2}
Ef(x)_s=\mycom{{}_\sss}{{}_{r\in[l],f(r)=s}}x_r
\end{equation}
with $s\in[m]$ and $x\in E[l]$ for every $l,m\in\bbN$ and morphism $f\in\Hom_\varOmega([l],[m])$

\end{enumerate}

\noindent 
defines a functor $E:\varOmega\rightarrow\varSigma$ injective on objects and morphisms.
Hence, there exists an $\varOmega$ category $E\varOmega$ having the functor $E$  
with the target category restricted  from $\varSigma$ to $E\varOmega$ as its stalk isofunctor.
Furthermore, $E\varOmega$ is a non monoidal subcategory of $\varSigma$.
\end{prop}

\noindent
Above, the $\msR^l$ are regarded as sets, although they have an obvious monoid structure.
It is further understood that $E[0]=0$. For $m\in\bbN$ and $f\in\Hom_\varOmega([0],[m])$,
$Ef=0$, the vanishing function with range $\msR^m$.

\begin{proof}
Assume that $l,m,n\in\bbN$, $f\in\Hom_\varOmega([l],[m])$, $g\in\Hom_\varOmega([m],[n])$ and $x\in E[l]$. Then,
from \ceqref{qfundt2}, we obtain \hphantom{xxxxxx}
\begin{align}
\label{qfunctp1}
Eg\circ Ef(x)_t&=Eg(Ef(x))_t
\\
\nonumber
&=\mycom{{}_\sss}{{}_{s\in[m],g(s)=t}}Ef(x)_s
\\
\nonumber
&=\mycom{{}_\sss}{{}_{s\in[m],g(s)=t}}\mycom{{}_\sss}{{}_{r\in[l],f(r)=s}}x_r
\\
\nonumber
&=\mycom{{}_\sss}{{}_{r\in[l]}}\mycom{{}_\sss}{{}_{s\in[m]}}\delta_{f(r),s}\delta_{g(s),t}x_r
\\
\nonumber
&=\mycom{{}_\sss}{{}_{r\in[l]}}\mycom{{}_\sss}{{}_{s\in[m]}}\delta_{f(r),s}\delta_{g\circ f(r),t}x_r
\\
\nonumber
&=\mycom{{}_\sss}{{}_{r\in[l],g\circ f(r)=t}}x_r=E(f\circ g)(x)_t,
\end{align}
where $t\in[n]$, demonstrating that $E$ preserves morphism composition. Suppose further that
$l\in\bbN$ and $x\in E[l]$. Then, we have 
\begin{align}
\label{qfunctp4}
E\id_{[l]}(x)_s&=\mycom{{}_\sss}{{}_{r\in[l],\id_{[l]}(r)=s}}x_r
\\
\nonumber
&=x_s=\id_{E[l]}(x)_s
\end{align}
with $s\in[l]$, showing that $E$ preserves identity assignments. It follows that $E$ is a functor.
Next, for every $l,m\in\bbN$, $E[l]=E[m]\Rightarrow\msR^l=\msR^m\Rightarrow l=m \Rightarrow [l]=[m]$.
Thus, $E$ is injective on objects. Suppose further that $l,m\in\bbN$, $f,g\in\Hom_\varOmega([l],[m])$
and  $Ef=Eg$ 
but that, by absurd, $f\neq g$. Then, there is $r\in[l]$ with $f(r)\neq g(r)$. Set $s=f(r)$,
$s'=g(r)$. Thus, $s,s'\in[m]$ with $s\neq s'$. Let $x\in E[l]$ with $x_r\neq 0$ and $x_{r'}=0$ for $r'\in[l]$, 
$r'\neq r$. A straightforward calculation employing \ceqref{qfundt2} gives $Ef(x)_s=x_r$, $Ef(x)_{s'}=0$, 
$Eg(x)_s=0$, $Eg(x)_{s'}=x_r$, which is impossible because $x_r\neq 0$ and $Ef=Eg$ by assumption. Therefore,
$f=g$  and so $E$ is injective on morphisms too.
The existence of an $\varOmega$ category $E\varOmega$ with $E$ as its stalk isofunctor
and its being a subcategory of $\varSigma$ are an immediate consequence of cor. \cref{rem:dtrick}.
\end{proof}

\noindent
Items \cref{d1}--\cref{d7} of prop. \cref{prop:domegacat}, with $D$ replaced by $E$, furnish a detailed 
description of $E\varOmega$. Since $E\varOmega$ is a subcategory of the finite set category
$\varSigma$, the objects and morphisms of $E\varOmega$ are again genuine sets and set 
functions respectively and the composition law and the identity assigning map of
$E\varOmega$ are the usual set theoretic ones. 

\begin{exa} \label{exa:qcsmact} Qudit configuration morphism action.
{\rm We recall that we use the conventions of items \cref{it:conv6}, \cref{it:conv7}. 
Consider the morphism $f\in\Hom_\varOmega([5],[7])$ 
\begin{equation}
\label{}
f=(0,6,1,2,2). 
\end{equation}
Let $\msR=\bbH_{0,8}$, the monoid of the order $8$ cyclic group $\bbZ_8$. Consider
the following elements $x,y,z\in E[5]$, 
\begin{equation}
\label{}
x=(7,1,1,0,2), \qquad y=(1,3,4,4,0),\qquad z=(1,0,7,2,2).
\end{equation}
Then, from \ceqref{qfundt2},
{\allowdisplaybreaks
\begin{align}
\label{}
&Ef(x)=(7,1,2,0,0,0,1), &Ef(y)=(1,4,4,0,0,0,3),
\\
\nonumber
&Ef(z)=(1,7,4,0,0,0,0). &~
\end{align}
}\!\! The three zero segment $0,0,0$ shared by these $7$ cdit configurations is a consequence of
$3,4,5\in[7]$ not belonging to the range of $f$. 
}
\end{exa}




\begin{defi} \label{def:mqadj}
For every two $l,m\in\bbN$, there is defined a multi cdit configuration monadic multiplication 
$\smallsmile:E[l]\times E[m]\rightarrow E([l]\smallsmile[m])$ by \pagebreak 
\begin{equation}
\label{mqsmile0}
x\smallsmile y_r
=\Bigg\{
\begin{array}{ll}
x_r&\text{if $r\in[l]$},\\
y_{r-l}&\text{if $r\in[m]+l$}.
\end{array}
\end{equation}
with $x\in E[l]$, $y\in E[m]$. We also set $o=0\in E[0]$.
\end{defi}

\noindent
Above, one has $x\smallsmile y_r=y_r$, $x\smallsmile y_r=x_r$ and $x\smallsmile y_r=0$
when $l=0$ and $m>0$, $l>0$ and $m=0$ and $l=m=0$, respectively.

\begin{exa} Multi cdit configuration monadic product. 
{\rm Let $\msR=\bbH_{0,8}$ as in ex. \cref{exa:qcsmact}. Consider the cdit configurations
$x\in E[5]$, $y\in E[7]$, 
\begin{equation}
\label{}
x=(0,1,1,2,0), \qquad y=(7,1,2,0,0,0,1).
\end{equation}
Then, by \ceqref{mqsmile0}, the joint $x\smallsmile y\in E[12]$ is 
\begin{equation}
\label{}
x\smallsmile y=(0,1,1,2,0,7,1,2,0,0,0,1). 
\end{equation}
}
\end{exa}

\noindent
A multi cdit configuration $\varOmega$ monadic structure is yielded in this way (cf. def. \cref{def:jointstrct}). 

\begin{prop} \label{prop:mqjoint}
The multi cdit configuration category $E\varOmega$ equipped with the monadic multiplication $\smallsmile$
and unit $o$ is a graded $\varOmega$ monad. 
\end{prop}

\begin{proof} From the defining relation \ceqref{mqsmile0}, it is immediately checked that 
\begin{equation}
\label{mqsmile1}
x\smallsmile(y\smallsmile z)=(x \smallsmile y)\smallsmile z
\end{equation}
for $l,m,n\in\bbN$, $x\in E[l]$, $y\in E[m]$, $z\in E[n]$ and that 
\begin{equation}
\label{mqsmile2}
x\smallsmile o=o\smallsmile x=x
\end{equation}
for $l\in\bbN$, $x\in E[l]$, verifying properties \ceqref{joint1}, \ceqref{joint2}.
There remains to show that for $l,m,p,q\in\bbN$, $f\in\Hom_\varOmega([l],[p])$, $g\in\Hom_\varOmega([m],[q])$ and
$x\in E[l]$, $y\in E[m]$ 
\begin{equation}
\label{mqsmile3}
Ef\smallsmile Eg(x\smallsmile y)=Ef(x)\smallsmile Eg(y), 
\end{equation}
demonstrating in this way property \ceqref{joint3}. Indeed, we have 
\begin{equation}
\label{}
Ef\smallsmile Eg(x\smallsmile y)_s=E(f\smallsmile g)(x\smallsmile y)_s
=\mycom{{}_\sss}{{}_{r\in[l+m],f\smallsmile g(r)=s}}x\smallsmile y_r
\end{equation}
with $s\in[p+q]$. Clearly, either $s\in[p]$ or $s\in[q]+p$, but not both. If $s\in[p]$, then
$r\in[l+m], f\smallsmile g(r)=s\Leftrightarrow r\in[l],f(r)=s$; if $s\in[q]+p$ instead,
then $r\in[l+m],f\smallsmile g(r)=s$ $\Leftrightarrow r\in[m]+l,g(r-l)+p=s$. 
Consequently,
{\allowdisplaybreaks
\begin{align}
Ef\smallsmile Eg(x\smallsmile y)_s
&=\left\{
\begin{array}{ll}
\mycom{{}_{\smallsss_{\vphantom{\ul{a}}}^{\vphantom{f}}}}{{}_{r\in[l],f(r)=s}}x\smallsmile y_r&\text{if $s\in[p]$},\\
\mycom{{}_{\smallsss_{\vphantom{f}}^{\vphantom{\ol{f}}}}}{{}_{r\in[m]+l,g(r-l)+p=s}}x\smallsmile y_r&\text{if $s\in[q]+p$}
\end{array}\right.
\\
\nonumber
&=\left\{
\begin{array}{ll}
\mycom{{}_{\smallsss_{\vphantom{\ul{a}}}^{\vphantom{f}}}}{{}_{r\in[l],f(r)=s}}x_r&\text{if $s\in[p]$},\\
\mycom{{}_{\smallsss_{\vphantom{f}}^{\vphantom{\ol{f}}}}}{{}_{r\in[m]+l,g(r-l)+p=s}}y_{r-l}&\text{if $s\in[q]+p$}
\end{array}\right.
\\
\nonumber
&=\left\{
\begin{array}{ll}
\mycom{{}_{\smallsss_{\vphantom{\ul{a}}}^{\vphantom{f}}}}{{}_{r\in[l],f(r)=s}}x_r&\text{if $s\in[p]$},\\
\mycom{{}_{\smallsss_{\vphantom{f}}^{\vphantom{\ol{f}}}}}{{}_{r\in[m],g(r)=s-p}}y_r&\text{if $s\in[q]+p$}
\end{array}\right.
\\
\nonumber
&=\Bigg\{
\begin{array}{ll}
Ef(x)_s&\text{if $s\in[p]$},\\
Eg(y)_{s-p}&\text{if $s\in[q]+p$}
\end{array}
=Ef(x)\smallsmile Eg(y)_s,
\end{align}
}
as required. 
\end{proof}


\subsection{\textcolor{blue}{\sffamily The multi qudit state $\varOmega$ monad $\scH_E\varOmega$}}
\label{subsec:graphhilb}

The multi qudit state $\varOmega$ monad $\scH_E\varOmega$ is a graded $\varOmega$ monad 
describing the states of a multi qudit quantum register. It is related, as already recalled,
to the multi cdit configuration monad $E\varOmega$ studied in subsect. \cref{subsec:grpct}
by basis encoding and is the last key  monad of our theory. $\scH_E\varOmega$ is the object of the present
subsection.




Consider the category $\ul{\rm fdHilb}$ of finite dimensional Hilbert spaces and linear maps thereof.
For notational convenience, we shall set $\bfsfH=\ul{\rm fdHilb}$.
In keeping with standard conventions,
the composition symbol $\circ$ of $\bfsfH$ is omitted
and the identity assigning map $\id$ of $\bfsfH$ is denoted by $1$. 

Let $\scH_1$ be a finite dimensional Hilbert space equipped with a distinguished orthonormal basis
$\ket{z}_1$, $z\in\msR$, where $\msR$ is the same finite commutative monoid used in the construction of the category $E\varOmega$.
Hence, $\dim\scH_1=|\msR|$. 

\begin{prop} \label{defi:scrhqdef}
The prescription assigning

\begin{enumerate}[label=\alph*)] 

\item \label{ite:hq1} an object $\scH_E[l]\in\Obj_{\bfsfH}$ given by 
\begin{equation}
\label{hqfundt1}
\scH_E[l]=\scH_1{}^{\otimes\hfpt l}
\end{equation}
for every $l\in\bbN$ and 

\item \label{ite:hq2} a morphism $\scH_Ef\in\Hom_{\bfsfH}(\scH_E[l],\scH_E[m])$ given by 
\begin{equation}
\label{hqfundt2}
\scH_Ef\ket{x}=\ket{Ef(x)}
\end{equation}
with $x\in E[l]$ for every $l,m\in\bbN$ and morphism $f\in\Hom_\varOmega([l],[m])$

\end{enumerate}

\noindent
specifies a functor $\scH_E:\varOmega\rightarrow\bfsfH$ injective on objects and morphisms.
Therefore, there exists an $\varOmega$ category $\scH_E\varOmega$ having the functor $\scH_E$  
with the target category $\bfsfH$ restricted to $\scH_E\varOmega$ as its stalk isofunctor.
Further, $\scH_E\varOmega$ is a non monoidal subcategory of $\bfsfH$.  
\end{prop}

\noindent
In {\it \cref{ite:hq1}}, $\scH_E[0]=\bbC$. In {\it \cref{ite:hq2}},  
$\ket{x}=\ket{x_1}_1\otimes\ldots\otimes\ket{x_l}_1$ with $x\in\msR^l$ denotes the orthonormal basis
of $\scH_1{}^{\otimes\hfpt l}$ associated with the basis $\ket{z}_1$ of $\scH_1$,
in keeping with the definition \ceqref{qfundt1} of $E[l]$. 
We shall not use the more precise notation $\ket{x}_l$ for $\ket{x}$,  
as the Hilbert space $\scH_E[l]$ to which $\ket{x}$ belongs can be inferred from the set $E[l]$
to which $x$ belongs. 

\begin{proof} 
Suppose that $l,m,n\in\bbN$ and $f\in\Hom_\varOmega([l],[m])$, $g\in\Hom_\varOmega([m],[n])$. 
By prop. \cref{prop:qfnct}, $E$ is a functor and so $E(g\circ g)=Eg\circ Ef$. Therefore,
for $x\in E[l]$, 
{\allowdisplaybreaks
\begin{align}
\label{}
\scH_E(g\circ f)\ket{x}&=\ket{E(g\circ f)(x)}
\\
\nonumber
&=\ket{Eg(Ef(x))}
\\
\nonumber
&=\scH_Eg\ket{Ef(x)}=\scH_Eg\scH_Ef\ket{x},
\end{align}
}
\!\!so that $\scH_E$ preserves morphism composition. Assume further that  $l\in\bbN$ and $x\in E[l]$. 
As $E$ is a functor as already recalled, $E\id_{l]}=\id_{E[l]}$. So, for $x\in E[l]$,
{\allowdisplaybreaks
\begin{align}
\label{}
\scH_E\id_{[l]}\mhfpt\ket{x}&=\ket{E\id_{[l]}(x)}
\\
\nonumber
&=\ket{\id_{E[l]}(x)}=1_{\scH_E[l]}\ket{x},
\end{align}
}
\!\!showing that $E$ preserves the identity assignment. Therefore, $\scH_E:$ is a functor. 
Next, suppose that $l,m\in\bbN$ and let $\scH_E[l]=\scH_E[m]$. Then, $\dim\scH_E[l]=\dim\scH_E[m]$ so that
$|\msR|^l=|\msR|^m$ whence $[l]=[m]$, where we used that $\dim\scH_E[p]=|\msR|^p$ for $p\in\bbN$.
This shows that $\scH_E$ is injective on objects. Assume further that $l,m\in\bbN$ and $f,g\in\Hom_\varOmega([l],[m])$
are such that $\scH_Ef=\scH_Eg$ holds. Then, for $x\in E[l]$, $\ket{Ef(x)}=\scH_Ef\ket{x}=\scH_Eg\ket{x}=\ket{Eg(x)}$
so that $Ef(x)=Eg(x)$ whence $Ef=Eg$. Since the functor $E$ is injective on morphisms, $f=g$. 
This proves that $\scH_E$ is injective on morphisms. 
Cor. \cref{rem:dtrick} now guarantees that there is an $\varOmega$ category $\scH_E\varOmega$ having 
$\scH_E$ as its stalk isofunctor and that $\scH_E\varOmega$ is a subcategory of $\bfsfH$. 
\end{proof}

\noindent
Items \cref{d1}--\cref{d7} of prop. \cref{prop:domegacat}, with $D$ replaced by $\scH_E$, provide the explicit
layout of $\scH_E\varOmega$. Now however, unlike previous similar instances, 
the objects and morphisms of $\scH_E\varOmega$ are Hilbert spaces and linear operators 
respectively and the composition law and the identity assigning map of
$\scH_E\varOmega$ are the customary Hilbert theoretic ones since the functor $\scH_E$ maps
to the finite dimensional Hilbert space category $\bfsfH$.





We notice here that by construction the monoidal product of objects and morphisms in the category $\scH_E\varOmega$ is just
the ordinary tensor product of Hilbert spaces and linear maps thereof. 

\begin{prop} \label{prop:htensor}
For $l,m\in\bbN$, it holds that 
\begin{equation}
\label{hmqsmile0}
\scH_E[l]\smallsmile\scH_E[m]=\scH_E[l]\otimes\scH_E[m].
\end{equation}  
Further, for any $l,m,p,q\in\bbN$ and $f\in\Hom_\varOmega([l],[p])$, $g\in\Hom_\varOmega([m],[q])$, one has 
\begin{equation}
\label{hmqsmile3}
\scH_Ef\smallsmile\scH_Eg=\scH_Ef\otimes\scH_Eg.
\end{equation}  
Finally, one has \hphantom{xxxxxxxxxxxxxxxxxxxxxx}
\begin{equation}
\label{hmqsmile4}
1_{\scH_E}=\bbC.
\end{equation}  
\end{prop}

\begin{proof}
By virtue of \ceqref{gamma1} and \ceqref{hqfundt1}  and the monoidal functoriality of $\scH_E$, we have 
{\allowdisplaybreaks
\begin{align}
\label{}
\scH_E[l]\smallsmile\scH_E[m]&=\scH_E([l]\smallsmile[m])
\\
\nonumber
&=\scH_E([l+m])
\\
\nonumber
&=\scH_1{}^{\otimes\hfpt l+m}
\\
\nonumber
&=\scH_1{}^{\otimes\hfpt l}\otimes\scH_1{}^{\otimes\hfpt m}=\scH_E[l]\otimes\scH_E[m],
\end{align}
}
\!\!showing \ceqref{hmqsmile0}. To prove \ceqref{hmqsmile3}. 
we note that 
for $u\in E[h]$, $v\in E[k]$ we have 
\begin{equation}
\label{}
\ket{u\smallsmile v}=\ket{u}\otimes\ket{v},
\end{equation}
as follows from the remarks just below def. \cref{defi:scrhqdef}.
Now, if $z\in E([l]\smallsmile[m])=E[l+m]$, then $z=x\smallsmile y$, where
$x\in E[l]$, $y\in E[m]$ with $x_r=z_r$ for $r\in[l]$
and $y_r=z_{r+l}$ for $r\in[m]+l$ by \ceqref{mqsmile0}. Then, by \ceqref{mqsmile3} and \ceqref{hqfundt2},  
{\allowdisplaybreaks
\begin{align}
\label{}
\scH_Ef\smallsmile\scH_Eg\ket{z}&=\scH_E(f\smallsmile g)\ket{x\smallsmile y}
\\
\nonumber
&=\ket{E(f\smallsmile g)(x\smallsmile y)}
\\
\nonumber
&=\ket{Ef\smallsmile Eg(x\smallsmile y)}
\\
\nonumber
&=\ket{Ef(x)\smallsmile Eg(y)}
\\
\nonumber
&=\ket{Ef(x)}\otimes\ket{Eg(y)}
\\
\nonumber
&=\scH_Ef\ket{x}\otimes\scH_Eg\ket{y}
\\
\nonumber
&=\scH_Ef\otimes\scH_Eg\hfpt\ket{x}\otimes\ket{y}
\\
\nonumber
&=\scH_Ef\otimes\scH_Eg\hfpt\ket{x\smallsmile y}=\scH_Ef\otimes\scH_Eg\hfpt\ket{z}. 
\end{align}
}
\!\!As $z\in E([l]\smallsmile[m])$ is arbitrary, relation \ceqref{hmqsmile3} holds. 
Identity \ceqref{hmqsmile4} follows immediately by noting that $1_{\scH_E}=\scH_E[0]=\bbC$. 
\end{proof}  

The category $\scH_E\varOmega$ also has a monadic multiplication: it is nothing but
the usual vector tensor multiplication.

\begin{defi} \label{def:mqsadj}
For any two $l,m\in\bbN$, the multi qudit state monadic multiplication is the function
$\smallsmile:\scH_E[l]\times \scH_E[m]\rightarrow\scH_E([l]\smallsmile[m])$ given by 
\begin{equation}
\label{hjoint}
\ket{\xi}\smallsmile\ket{\eta}=\ket{\xi}\otimes\ket{\eta}
\end{equation}
for $\ket{\xi}\in\scH_E[l]$, $\ket{\eta}\in\scH_E[m]$. The multi qudit state monadic unit is $\ket{0}\in\scH_E[0]$. 
\end{defi}

\noindent
We note that the above definition makes sense by virtue of \ceqref{hmqsmile0}. We note also that
\begin{equation}
\label{hjoint0}
\ket{x}\smallsmile\ket{y}=\ket{x\smallsmile y}
\end{equation}
for $x\in E[l]$, $y\in E[m]$. Indeed, from \ceqref{mqsmile0}, it is apparent that 
$\ket{x\smallsmile y}=\ket{x}\otimes\ket{y}$.
A multi qudit state $\varOmega$ monadic structure is so made available (cf. def. \cref{def:jointstrct}).

\begin{prop} \label{prop:hjoint}
The multi qudit state category $\scH_E\varOmega$ equipped with the monadic multiplication $\smallsmile$
and unit $\ket{0}$ is a graded $\varOmega$ monad. 
\end{prop}

\begin{proof}
From \ceqref{hjoint} and well--known facts of Hilbertian tensor algebra, one has 
\begin{equation}
\label{hjoint1}
\ket{\xi}\smallsmile(\ket{\eta}\smallsmile\ket{\zeta})=(\ket{\xi}\smallsmile\ket{\eta})\smallsmile\ket{\zeta}
\end{equation}
for $l,m,n\in\bbN$, $\ket{\xi}\in\scH_E[l]$, $\ket{\eta}\in\scH_E[m]$, $\ket{\zeta}\in\scH_E[n]$ and 
\begin{equation}
\label{hjoint2}
\ket{\xi}\smallsmile\ket{0}=\ket{0}\smallsmile\ket{\xi}=\ket{\xi}
\end{equation}
for $l\in\bbN$, $\ket{\xi}\in\scH_E[l]$, showing that properties \ceqref{joint1}, \ceqref{joint2}
hold true. Exploiting \ceqref{hmqsmile3}, we find further that
\begin{align}
\label{}
\scH_Ef\smallsmile\scH_Eg\hfpt\ket{\xi}\smallsmile\ket{\eta}&=\scH_Ef\otimes\scH_Eg\hfpt\ket{\xi}\otimes\ket{\eta}
\\
\nonumber
&=\scH_Ef\ket{\xi}\otimes\scH_Eg\ket{\eta}=\scH_Ef\ket{\xi}\smallsmile\scH_Eg\ket{\eta}
\end{align}  
for $l,m,p,q\in\bbN$, $f\in\Hom_\varOmega([l],[p])$, $g\in\Hom_\varOmega([m],[q])$, $u\in \scH_E[l]$, $v\in \scH_E[m]$,
so demonstrating that also property \ceqref{joint3} holds. 
\end{proof}


\subsection{\textcolor{blue}{\sffamily Some operatorial properties of $\scH_E\varOmega$}}\label{subsec:mqsrmk}

In this short final subsection, we examine further properties of the operator content
of the multi qudit state category $\scH_E\varOmega$ repeatedly cited in the hypergraph state construction
carried out in II.

From \ceqref{hqfundt2}, for $l,m\in\bbN$ the operator $\scH_Ef\in\Hom_{\bfsfH}(\scH_E[l],\scH_E[m])$
attached to a morphism $f\in\Hom_\varOmega([l],[m])$ can be written as 
\begin{equation}
\label{mqsrmk1}
\scH_Ef=\mycom{{}_\sss}{{}_{x\in E[l]}}\ket{Ef(x)}\bra{x}.
\end{equation}
This expressions is useful in many situations.

Let us denote by $\msI(\scK,\scL)$ the set of all isometric linear operators of the Hilbert spaces $\scK$, $\scL$
and by $\msU(\scK)$ the group of unitary operators of a Hilbert space $\scK$.
  
\begin{prop} \label{prop:unihef}
Let $l,m\in\bbN$ and let $f\in\Hom_\varOmega([l],[m])$. If $l<m$ and $f$ is injective, then 
$\scH_Ef\in\msI(\scH_E[l],\scH_E[m])$. If $l=m$ and $f$ is bijective, then $\scH_Ef\in\msU(\scH_E[l])$. 
\end{prop}

\begin{proof}
Suppose that $l<m$ and that $f$ is injective. Then, there exists a morphism $g\in\Hom_\varOmega([m],[l])$ such that
$g\hfpt\circ\hfpt f=\id_{[l]}$. By the functoriality of $E$, we have then that $Eg\hfpt\circ\hfpt Ef=\id_{E[l]}$.
Consequently, $Ef$ is injective too. On account of by \ceqref{mqsrmk1}, 
the operator $\scH_Ef$ is isometric,  since it maps the orthonormal vectors $\ket{x}$, $x\in E[l]$,
into distinct orthonormal vectors $\ket{Ef(x)}$. If $l=m$ and $f$ is bijective, then the isometric
operator $\scH_Ef$ is evidently unitary. 
\end{proof}


\vfill\eject

\vfill\eject

\appendix

\renewcommand{\sectionmark}[1]{\markright{\thesection\ ~~#1}}

\section{\textcolor{blue}{\sffamily Monoidal categories and graded monads}}\label{app:cat}

In the following appendices, we review the basic definitions and results of the theory of strict
monoidal categories and graded monads. We assume that the reader has some familiarity with
the notions of category, functor, and natural transformation, although we briefly
recall these concepts as well. Our focus is on the aspects most relevant to the material
covered in the main text. As such, the exposition is neither systematic nor complete, and
its scope is intentionally limited. Further, for brevity, 
proofs are either omitted or only outlined and the 
use of categorical diagrams, which are commonly employed in category theory to provide
intuition for its rather abstract constructions, is avoided. Standard references for the topics discussed here are
\ccite{MacLane:1978cwm,Fong:2019act}.

There is a whole wealth of examples of the notions dealt with below, which arise across various
areas of mathematics and could serve to illustrate them.  Again, for lack of space, we shall not detail any of them 
and refer the reader to the body of the paper, where many of them are instantiated, 
and refs. \!\!\ccite{MacLane:1978cwm,Fong:2019act}.


\subsection{\textcolor{blue}{\sffamily Categories, functors and natural transformations}}\label{subsec:catrev}

Category theory is a universal language of all areas of mathematics.
The notion of category is therefore central. 

\begin{defi} \label{def:catdef}
A category $\clC$ consists of the following data:

\begin{enumerate}[label=\alph*)]

\item a collection of objects $\Obj_{\hfpt\clC}$;

\item a set of morphisms $\Hom_{\hfpt\clC}(X,Y)$ for any two objects $X,Y\in\Obj_{\hfpt\clC}$;

\item a composition law $\circ_{\hfpt\clC}:\Hom_{\hfpt\clC}(X,Y)\times\Hom_{\hfpt\clC}(Y,Z)\rightarrow\Hom_{\hfpt\clC}(X,Z)$
for any three objects $X,Y,Z\in\Obj_{\hfpt\clC}$ 
associating with each pair of morphisms $f\in\Hom_{\hfpt\clC}(X,Y)$, $g\in\Hom_{\hfpt\clC}(Y,Z)$ their composition
$g\circ_{\hfpt\clC} f\in\Hom_{\hfpt\clC}(X,Z)$;

\item an identity assigning map attaching to any object $X\in\Obj_{\hfpt\clC}$
an identity morphism $\id_{\hfpt\clC X}\in\Hom_{\hfpt\clC}(X,X)$. 

\end{enumerate}

\noindent
The following properties must further be satisfied. 

\begin{enumerate}

\item Composition is associative: for every objects $X,Y,Z,W\in\Obj_{\hfpt\clC}$ and morphisms
$f\in\Hom_{\hfpt\clC}(X,Y)$, $g\in\Hom_{\hfpt\clC}(Y,Z)$, $h\in\Hom_{\hfpt\clC}(Z,W)$
\begin{equation}
\label{catdef1}
h\circ_{\hfpt\clC}(g\circ_{\hfpt\clC} f)=(h\circ_{\hfpt\clC} g)\circ_{\hfpt\clC} f.
\end{equation}

\item Composition is unital: for every object $X\in\Obj_{\hfpt\clC}$ and morphism
$f\in\Hom_{\hfpt\clC}(X,Y)$,
\begin{equation}
\label{catdef2}
\id_{\hfpt\clC Y}\circ_{\hfpt\clC} f=f\circ_{\hfpt\clC}\id_{\hfpt\clC X}=f.
\end{equation}

\end{enumerate}
\end{defi}

\noindent
If $X,Y\in\Obj_{\hfpt\clQ}$ are objects and $f\in\Hom_{\hfpt\clC}(X,Y)$ is a morphism,
$X$ and $Y$ are called respectively the source and target objects of $f$. The set of
all morphisms of $\clC$, i.e the union $\Hom_{\hfpt\clC}=\bigcup_{X,Y\in\Obj_{\hfpt\clC}}\Hom_{\hfpt\clC}(X,Y)$ of
all morphisms sets, is a relevant feature of $\clC$. 
Two morphisms $f,g\in\Hom_{\hfpt\clC}$ are composable if $f\in\Hom_{\hfpt\clC}(X,Y)$,
$g\in\Hom_{\hfpt\clC}(Y,Z)$ for some objects $X,Y,Z\in\Obj_{\hfpt\clQ}$ and their
composite $g\circ_{\hfpt\clC} f$ is defined.
The subscript $\clC$ attached to the composition and identity symbols $\circ_{\hfpt\clC}$ and
$\id_{\hfpt\clC}$ can be omitted when no confusion arises; we shall show them when required by 
conceptual clarity. 


Functors are maps of categories suitably
compatible with morphism composition and identities. 

\begin{defi} \label{def:functdef}
Let $\clC$ and $\clD$ be categories. A functor $F:\clC\rightarrow\clD$ consists of the following data:

\begin{enumerate}[label=\alph*)]

\item a mapping associating with each object $X\in\Obj_{\hfpt\clC}$ an object $FX\in\Obj_{\hfpt\clD}$;

\item for any two objects $X,Y\in\Obj_{\hfpt\clC}$
a mapping associating with each morphism $f\in\Hom_{\hfpt\clC}(X,Y)$ a morphism $Ff\in\Hom_{\hfpt\clD}(FX,FY)$.

\end{enumerate}

\noindent 
Further the following two conditions must hold:

\begin{enumerate} 

\item for any pair of composable morphisms $f,g\in\Hom_{\hfpt\clC}$, 
\begin{equation}
\label{functdef1}
F(g\circ_{\hfpt\clC} f)=Fg\circ_{\hfpt\clD\mhfpt} Ff;
\end{equation}

\item for every object for $X\in\Obj_{\hfpt\clC}$,
\begin{equation}
\label{functdef2}
F\id_{\hfpt\clC X}=\id_{\hfpt\clD F(X)}.
\end{equation}

\end{enumerate}
\end{defi}

\noindent
If $\clD=\clC$, $F$ is called an endofunctor. If $F$ is bijective on object and morphisms, $F$ is called an isofunctor.
Categories related by isofunctors are isomorphic. From a purely categorical perspective, they are indistinguishable.

Small categories are categories whose object collection is a set. 
Small categories and functors are the objects and morphisms of the category $\ul{\rm Cat}$
of small categories. The composition and identity symbols of $\ul{\rm Cat}$ are denoted simply 
by $\circ$ and $\id$ with no subscripts. The composite $G\circ F:\clC\rightarrow\clE$ of two functors
$F:\clC\rightarrow\clD$, $G:\clD\rightarrow\clE$ is defined by setting $G\circ FX=GFX$
and $G\circ Ff=GFf$ on objects and morphisms respectively. The identity functor
$\id_{\hfpt\clC}:\clC\rightarrow\clC$ acts as $\id_{\hfpt\clC}X=X$ and $\id_{\hfpt\clC}f=f$.


Natural transformations describe relationships of functors between the same pair of categories
compatible the categorical structure of both. 

\vspace{.75mm}

\begin{defi}  \label{def:natdef}
Let $\clC$ and $\clD$ be categories and let $F,G:\clC\rightarrow\clD$ be functors.
A natural transformation $\alpha:F\Rightarrow G$ is a collection of morphisms
$\alpha_X\in\Hom_{\hfpt\clD}(FX,GX)$, $X\in\Obj_{\hfpt\clD}$ such that for all objects $X,Y\in\Obj_{\hfpt\clD}$ 
and morphisms $f\in\Hom_{\hfpt\clC}(X,Y)$
\begin{equation}
\label{natdef}
\alpha_Y\circ_{\hfpt\clD\mhfpt} Ff=Gf\circ_{\hfpt\clD\mhfpt}\alpha_X.
\end{equation}
\end{defi}

\vspace{1mm}\noindent
The $\alpha_X$ are the components of the transformation. 
If the $\alpha_X$ are isomorphisms for all $X$, $\alpha$ is a natural isomorphism.

Given two categories $\clC$, $\clD$, the functors $F:\clC\rightarrow\clD$ and the natural transformations 
$\alpha:F\Rightarrow G$ are respectively the objects and morphisms of the functor category
$\Funct(\clC,\clD)$. The composition and identity symbol of $\Funct(\clC,\clD)$ are also denoted 
by $\circ$ and $\id$. The composite $\beta\circ\alpha:F\Rightarrow H$ of two natural transformations 
$\alpha:F\Rightarrow G$, $\beta:G\Rightarrow H$ is defined to be the one with components 
$(\beta\circ\alpha)_X=\beta{}_X\circ_{\hfpt\clD}\alpha_X$. The identity natural transformation 
$\id_F:F\Rightarrow F$ has components $\id_{FX}=\id_{\hfpt\clD FX}$.  
When $\clD=\clC$, the functor category $\Funct(\clC,\clD)$ becomes the endofunctor
category $\End(\clC)$.

There exist a wide range of constructions yielding a new category out of one or more given
ones. We present here only Cartesian multiplication of categories for its relevance in
the theory of monoidal categories reviewed later. 
The Cartesian product of two categories can be defined in straightforward fashion using systematically
Cartesian multiplication on objects and morphisms.

\begin{defi} \label{def:cartcat}
The Cartesian product $\clC_1\times\clC_2$ of two categories $\clC_1$, $\clC_2$
is the category specified by the following data. 

\begin{enumerate}[label=\alph*)] 
  
\item The object collection of $\clC_1\times\clC_2$ is $\Obj_{\hfpt\clC_1\times\clC_2}=\Obj_{\hfpt\clC_1}\times\Obj_{\hfpt\clC_2}$.
  
\item The morphisms set in $\clC_1\times\clC_2$ of any two objects $(X_1,X_2),(Y_1,Y_2)\in\Obj_{\hfpt\clC_1\times\clC_2}$
is $\Hom_{\hfpt\clC_1\times\clC_2}((X_1,X_2),(Y_1,Y_2))=\Hom_{\hfpt\clC_1}(X_1,Y_1)\times\Hom_{\hfpt\clC_2}(X_2,Y_2)$.

  
\item 
For any composable morphism pair $(f_1,f_2),(g_1,g_2)\in\Hom_{\hfpt\clC_1\times\clC_2}$
  \begin{equation}
\label{cartcat1}
(g_1,g_2)\circ_{\clC_1\times\clC_2}(f_1,f_2)=(g_1\circ_{\clC_1}f_1,g_2\circ_{\clC_2}f_2).
\end{equation}

\item For every object $(X_1,X_2)\in\Obj_{\hfpt\clC_1\times\clC_2}$
\begin{equation}
\label{cartcat2}
\id_{\hfpt\clC_1\times\clC_2(X_1,X_2)}=(\id_{\hfpt\clC_1X_1},\id_{\hfpt\clC_2X_2}).
\end{equation}

\end{enumerate}  
  
\end{defi}

\noindent 
It is straightforward to check that the conditions listed in def. \cref{def:catdef}
are satisfied by $\clC_1\times\clC_2$. 

A functor 
whose source category is a Cartesian product category is called
a bifunctor. A natural transformation 
of bifunctors is called binatural.


\subsection{\textcolor{blue}{\sffamily Strict monoidal categories}}\label{subsec:monocat}

The notions of strict monoidal category and functor is central in the theory
worked out in this paper. 

\begin{defi} \label{def:moncat}
A strict monoidal category is a category $\clC$ equipped with

\begin{enumerate}[label=\alph*)]
    
\item a bifunctor $\otimes_{\hfpt\clC} :\clC\times\clC\rightarrow\clC$ called monoidal product and

\item a distinguished object $1_\clC \in\Obj_{\hfpt\clC}$ called monoidal unit

\end{enumerate}

\noindent
such that the following properties hold.

\begin{enumerate}

\item Monoidal bifunctorial correspondence on objects and morphisms:
for every objects $X,Y,U,V\in\Obj_{\hfpt\clC}$ and morphisms $f\in\Hom_{\hfpt\clC}(X,U)$,
$g\in\Hom_{\hfpt\clC}(Y,V)$, one has 
$f\otimes_{\hfpt\clC}g\in\Hom_{\hfpt\clC}(X\otimes_{\hfpt\clC}Y,U\otimes_{\hfpt\clC}V)$.

\item Monoidal bifunctoriality of composition: for every  morphisms $f,g,h,k\in\Hom_{\hfpt\clC}$ such that
the pairs $f$, $g$ and $h$, $k$ are composable, the pair $g\otimes_{\hfpt\clC}k$, $f\otimes_{\hfpt\clC}h$ is and 
\begin{equation}
\label{moncat0}
(g\circ_{\hfpt\clC}f)\otimes_{\hfpt\clC}(k\circ_{\hfpt\clC}h)
=(g\otimes_{\hfpt\clC}k)\circ_{\hfpt\clC}(f\otimes_{\hfpt\clC}h).
\end{equation}

\item Monoidal bifunctoriality of identities: for every objects $X,Y\in\Obj_{\hfpt\clC}$,
\begin{equation}
\label{moncat-1}
\id_{\hfpt\clC X}\otimes_{\hfpt\clC}\id_{\hfpt\clC Y}=\id_{\hfpt\clC X\otimes_{\hfpt\clC} Y}.
\end{equation}

\item Strict monoidal associativity on objects: for all objects $X,Y,Z\in\Obj_{\hfpt\clC}$ 
\begin{equation}
\label{moncat1}
(X\otimes_{\hfpt\clC} Y)\otimes_{\hfpt\clC} Z=X\otimes_{\hfpt\clC} (Y\otimes_{\hfpt\clC} Z).
\end{equation}
    
\item Strict monoidal associativity on morphisms: for all morphisms $f,g,h\in\Hom_{\hfpt\clC}$
\begin{equation}
\label{moncat2}
(f\otimes_{\hfpt\clC} g)\otimes_{\hfpt\clC} h=f\otimes_{\hfpt\clC} (g\otimes_{\hfpt\clC} h).
\end{equation}

\item Strict monoidal unitality on objects: for all objects $X\in\Obj_{\hfpt\clC}$ 
\begin{equation}
\label{moncat3}
X\otimes_{\hfpt\clC} 1_\clC=1_\clC\otimes_{\hfpt\clC}X=X.
\end{equation}

 \item Strict monoidal unitality on morphisms: for all morphisms $f\in\Hom_{\hfpt\clC}$
\begin{equation}
\label{moncat4}
f\otimes_{\hfpt\clC} \id_{1_\clC}=\id_{1_\clC}\otimes_{\hfpt\clC}f=f.
\end{equation}

\end{enumerate}
\end{defi}

\noindent
It can be checked that these properties are compatible with the conditions listed in def. \cref{def:catdef}.

Strict monoidal functors are functors of strict monoidal categories respecting their monoidal structures.
This property is formalized in the following definition.

\begin{defi} \label{def:monfun}
A strict monoidal functor $F:\clC\rightarrow\clD$ of the strict monoidal categories $\clC$, $\clD$
is a functor obeying the following conditions.

\begin{enumerate}

\item Compatibility with monoidal multiplication of objects: for every $X,T\in\Obj_{\hfpt\clC}$
\begin{equation}
\label{monfun1}
F(X\otimes_{\hfpt\clC} Y)=FX\otimes_{\hfpt\clD\mhfpt}FY.
\end{equation}

\item Compatibility with monoidal multiplication of morphisms: for every $f,g\in\Hom_{\hfpt\clC}$
\begin{equation}
\label{monfun2}
F(f\otimes_{\hfpt\clC} g)=Ff\otimes_{\hfpt\clD\mhfpt}Fg.
\end{equation}

\item Preservation of monoidal units: \hphantom{xxxxxxxxxxx}
\begin{equation}
\label{monfun3}
F1_\clC =1_{\clD}. 
\end{equation}

\end{enumerate}

\end{defi}

\noindent
It is straightforward to verify that these properties are compatible with the conditions listed in def.
\cref{def:functdef}.

Small strict monoidal categories and the strict monoidal functors between them are the objects and morphisms
of the category $\ul{\rm SMCat}$ of small strict monoidal categories.
The composition of functors and the identity functors are defined as in the
small category category $\ul{\rm Cat}$, rendering $\ul{\rm SMCat}$ a subcategory of $\ul{\rm Cat}$. 

There exists also a notion of strict monoidal natural transformation.

\begin{defi}
A natural transformation $\alpha:F\Rightarrow G$ of strict monoidal functors $F,G:\clC\rightarrow\clC$
is strict monoidal when
\begin{equation}
\label{monnat1}
\alpha_{X\otimes_{\hfpt\clC}Y}=\alpha_X\otimes_{\hfpt\clD}\alpha_Y
\end{equation}
for all objects $X,Y\in\Obj_{\hfpt\clC}$. 
\end{defi}

For a given pair $\clC$, $\clD$ of strict monoidal categories, the strict monoidal functors
$F:\clC\rightarrow\clD$ and strict monoidal natural transformations 
$\alpha:F\Rightarrow G$ are respectively the objects and morphisms of the strict monoidal
functor category $\SMFunct(\clC,\clD)$.
The composition of natural transformations and the identity natural transformations are defined as in the
category $\Funct(\clC,\clD)$ and $\SMFunct(\clC,\clD)$ is so a subcategory of $\Funct(\clC,\clD)$.
For $\clC=\clD$, one has the the strict monoidal endofunctor category $\SMEnd(\clC)$.


\subsection{\textcolor{blue}{\sffamily Symmetric strict monoidal categories}}\label{subsec:sumdef}

Strict monoidal categories often come equipped with a symmetric structure which further enriches them.
In such a case, we have symmetric strict monoidal categories.

\begin{defi} \label{def:symmon}
A symmetric strict monoidal category is a strict monoidal category $\clC$ equipped with a collection of morphisms
$s_{\hfpt\clC X,Y}\in\Hom_{\hfpt\clC}(X\otimes_{\hfpt\clC} Y,Y\otimes_{\hfpt\clC} X)$ for varying
pairs of objects $X,Y\in\Obj_{\hfpt\clC}$, called swap maps, obeying the following conditions.

\begin{enumerate}

\item Naturality over monoidal multiplication: for all objects $X,Y,U,V\in\Obj_{\hfpt\clC}$ and morphisms
$f\in\Hom_{\hfpt\clC}(X,U)$,  $g\in\Hom_{\hfpt\clC}(Y,V)$, 
\begin{equation}
\label{symmon1}
s_{\hfpt\clC U,V}\circ_{\hfpt\clC} (f\otimes_{\hfpt\clC} g)=(g\otimes_{\hfpt\clC} f)\circ_{\hfpt\clC} s_{\hfpt\clC X,Y}.
\end{equation}
  
\item Compatibility with monoidal associativity: for all objects $X,Y,Z\in\Obj_{\hfpt\clC}$,
\begin{equation}
\label{symmon2}
(\id_{\hfpt\clC Y}\otimes_{\hfpt\clC}\,s_{\hfpt\clC X,Z})\circ_{\hfpt\clC}(s_{\hfpt\clC X,Y}\otimes_{\hfpt\clC}\id_{\hfpt\clC Z})
=s_{\hfpt\clC X,Y\otimes_{\hfpt\clC} Z}.
\end{equation}

\item Normalization: for all objects $X\in\Obj_{\hfpt\clC}$,
\begin{equation}
\label{symmon3}
s_{\hfpt\clC X,1_\clC}=s_{\hfpt\clC 1_\clC,X}=\id_{\hfpt\clC X}\!.
\end{equation}

\item Validity of the inverse law stating that for every objects $X,Y\in\Obj_{\hfpt\clC}$,
\begin{equation}
\label{symmon4}
s_{\hfpt\clC Y,X}\circ_{\hfpt\clC} s_{\hfpt\clC X,Y}=\id_{\hfpt\clC X\otimes_{\hfpt\clC} Y}\!.
\end{equation}

\end{enumerate}

\end{defi}

\noindent
It can be verified that these requirements are compatible with those listed in def.
\cref{def:moncat}. Symmetric strict monoidal categories are also known as permutative categories
in the mathematical literature. 

Symmetric strict monoidal functors are strict monoidal functors of symmetric strict monoidal categories
compatible with the symmetric structures of these latter.

\begin{defi} \label{def:symmonfnct}
A symmetric strict monoidal functor is a strict monoidal functor $F:\clC\rightarrow\clD$
of strict monoidal categories such that for all objects $X,Y\in\Obj_{\hfpt\clC}$,
\begin{equation}
\label{symmon5}
s_{\hfpt\clD FX,FY}=Fs_{\hfpt\clC X,Y}. 
\end{equation}
  
\end{defi}


\subsection{\textcolor{blue}{\sffamily Pro and Prop categories}}\label{subsec:procat}

Product or Pro (also written as PRO) categories are the kind of categories most relevant
in the present paper \ccite{Maclane:1965cta,Boardman:1968hes,May:1972ils,Markl:2002otp}.

\begin{defi} \label{def:procat}The 
A Pro is a strict monoidal category $\clP$ for which there is a distinguished object $G\in\Obj_\clP$ 
such that every object $X\in\Obj_\clP$ is of the form
\begin{equation}
\label{procat1}
G^{\otimes n}=G\otimes_{\hfpt\clP}\cdots\otimes_{\hfpt\clP}G\qquad\text{($n$ factors)}
\end{equation}
for some $n\in\bbN$. 
\end{defi}

\noindent The special object $G$ with the above property is then called the generating object of the Pro $\clP$.
A product and permutation or Prop (also written as PROP) category is a Pro $\clP$ which is symmetric as
a strict monoidal category. 

A related notion is that of model (also known as representation) of a Pro. Let $\clC$ be a strict monoidal category. 

\begin{defi} \label{def:procatmod}
A model in $\clC$ of a Pro $\clP$ is a strict monoidal functor $M:\clP\rightarrow\clC$.
\end{defi}

\noindent The object $MG\in\Obj_{\hfpt\clC}$ associated with a generator $G$ is called the underlying object of
the model $M$. For a Prop $\clP$, the functor $F$ is required moreover to be symmetric strict monoidal.


\subsection{\textcolor{blue}{\sffamily Free strict monoidal categories}}\label{subsec:frecat}

In this appendix, we present first the notion of free category, though we shall not need it directly, 
and then use it to introduce the notion of free strict monoidal category, 
which is key in the construction of certain graded monads described in full detail
in app. \cref{subsec:gradmon}.

Quivers and their morphisms \ccite{Gabriel:1972uds,Barr:2012ccs}
are the elementary constituents of free categories and functors.
Our treatment therefore necessarily begins with them.

A quiver is an abstract model of a collection of objects and morphisms
with specified source and target objects but
with no composition operation and identity assigning map defined on them. 

\begin{defi} \label{def:quiv}
A quiver $\clQ$ consists of the following data:

\begin{enumerate}[label=\alph*)]

\item a collection of vertices $\Ver_{\hfpt\clQ}$;

\item a set of edges $\Edg_{\hfpt\clQ}(X,Y)$ for any two vertices $X,Y\in\Ver_{\hfpt\clQ}$.

\end{enumerate}
  
\end{defi}

\noindent
The union $\Edg_{\hfpt\clQ}=\bigcup_{X,Y\in\Ver_{\hfpt\clQ}}\Edg_{\hfpt\clQ}(X,Y)$ of all edge sets of $\clQ$
is an item often useful to express concisely statements about $\clQ$. 

A morphism of quivers is analogously an abstract model for a map of two
collections of objects and morphisms 
with no composition operation and identity assigning map defined on them. 

\begin{defi} \label{def:quivmor}
Let $\clQ$, $\clR$ be quivers. A quiver morphism $L:\clQ\rightarrow\clR$ consists of the following data:

\begin{enumerate}[label=\alph*)]

\item a mapping associating with each vertex $X\in\Ver_{\hfpt\clQ}$ a vertex $LX\in\Ver_{\hfpt\clR}$;

\item for any vertices $X,Y\in\Ver_{\hfpt\clQ}$, 
a mapping assigning to each edge $e\in\Edg_{\hfpt\clQ}(X,Y)$ an edge $Le\in\Edg_{\hfpt\clR}(LX,LY)$.

\end{enumerate}

\end{defi}

Small quivers are quivers whose vertex collection is a set. 
Small quivers and quiver morphisms are the constituents of the small quiver category $\ul{\rm Quiv}$.
The composition and identity symbols of $\ul{\rm Quiv}$ are denoted simply 
by $\circ$ and $\id$ with no subscripts. The composite $M\circ L:\clQ\rightarrow\clS$ of two quiver morphisms
$L:\clQ\rightarrow\clR$, $M:\clR\rightarrow\clS$ is defined by setting $M\circ LX=MLX$
and $M\circ Le=MLe$ on vertices and edges respectively. The identity quiver morphism
$\id_{\hfpt\clQ}:\clQ\rightarrow\clQ$ acts as $\id_{\hfpt\clQ}X=X$ and $\id_{\hfpt\clQ}e=e$. 

Small categories and their functors are characterized by underlying quivers and quiver morphisms
specified by a special functor. 

\begin{defi} \label{def:forgfunct}
The forgetful quiver functor is the functor $\msU:\ul{\rm Cat}\rightarrow\ul{\rm Quiv}$ of the small category
category $\ul{\rm Cat}$ to the small quiver category $\ul{\rm Quiv}$ defined thusly.
$\msU$ associates with a small category $\clC$ a small quiver $\msU\clC$ so specified. 

\begin{enumerate}[label=\alph*)]

\item $\Ver_{\hfpt\msU\clC}=\Obj_{\hfpt\clC}$ is the vertex set of $\msU\clC$.

\item $\Edg_{\hfpt\msU\clC}(X,Y)=\Hom_{\hfpt\clC}(X,Y)$ is the edge set  of $\msU\clC$
  of two vertices $X,Y\in\Ver_{\hfpt\msU\clC}$.

\end{enumerate}

\noindent 
Furthermore, $\msU$ associates with a functor $F:\clC\rightarrow\clD$ of small categories
the morphism $\msU F:\msU\clC\rightarrow\msU\clD$ of small quivers specified as follows.

\begin{enumerate}[label=\alph*),start=3]

\item One has $\msU FX=FX\in\Ver_{\hfpt\msU\clD}$ for each vertex $X\in\Ver_{\hfpt\msU\clC}$.

\item One has $\msU Fe=Fe\in\Edg_{\hfpt\msU\clD}$ for every 
edge $e\in\Edg_{\hfpt\msU\clC}$. 

\end{enumerate}
\end{defi}

\noindent
In this way, the action of the functor $\msU$ completely erases 
all information about which morphisms are composite of others and which are identities.

Conversely, it is possible to associate a small free category with any small quiver
and a functor of small free categories with any morphism of small quivers
through an apposite functor. To define this latter, it is necessary to 
introduce a few preliminary notions.

Let $\clQ$ be a quiver. 
The free categorical data of $\clQ$ consists of a set of elementary free symbols and a set of
regular free expressions built with those symbols according to certain rules. 
The elementary free symbol set comprises for each vertex $X\in\Ver_{\hfpt\clQ}$
an object symbol also denoted by $X$, for each edge $e\in\Edg_{\hfpt\clQ}$
a morphism symbol also denoted by $e$, a composition symbol $\circ_{\hfpt\clQ}$
and an identity assigning symbol $\id_{\hfpt\clQ}$.
The regular free expressions are the bracketed expressions built with the symbols
$X$, $e$, $\circ_{\hfpt\clQ}$ and $\id_{\hfpt\clQ}$ following
the same rules which would apply if they 
were objects, morphisms, composition and identity assigning operations of a category
(cf. def. \cref{def:catdef}).
Expressions made with the $X$ and $e$ which can be turned into one another by using
the associativity and unitality relations \eqref{catdef1}, \eqref{catdef2} with $\clC$ replaced by $\clQ$ 
are identified. (This can be made completely precise by employing suitable equivalence relations on specified subsets of
expressions.)

The free functorial datum of a morphism $L:\clQ\rightarrow\clR$ of two quivers $\clQ$, $\clR$ is a map
also denoted by $L$ of the free categorical data of $\clQ$ to those of $\clR$ acting
in the following manner. $L$ maps the free elementary object and morphism
symbols $X$ and $e$ of $\clQ$ into the corresponding free elementary symbols
$LX$  and $Le$ of $\clR$ and turns each regular free expression $f$ of $\clQ$ into
one $Lf$ of $\clR$ by acting as if it was a functor of two categories (cf. def. \cref{def:functdef}). 

We are now ready to introduce the free category construction. 

\begin{defi} \label{def:frecat}
The free category functor is the functor $\msF:\ul{\rm Quiv}\rightarrow\ul{\rm Cat}$ of the small quiver category
$\ul{\rm Quiv}$ to the small category category $\ul{\rm Cat}$ defined as follows. 
$\msF$ associates with a small quiver $\clQ$ a small category $\msF\clQ$ specified as follows. 

\begin{enumerate} [label=\alph*)]

\item \label{it:frecat1} $\Obj_{\hfpt\msF\clQ}=\Ver_{\hfpt\clQ}$ is the object collection of $\msF\clQ$.

\item \label{it:frecat5} For any two objects $X,Y\in\Obj_{\hfpt\msF\clQ}$, the morphism set
$\Hom_{\hfpt\msF\clQ}(X,Y)$ consists of the regular free expressions of $\clQ$
with source and target object symbols $X$ and $Y$. 

\item \label{it:frecat6} For every composable morphism pair $f,g\in\Hom_{\hfpt\msF\clQ}$,
the composite morphism of $f$ and $g$ is given by $g\circ_{\hfpt\msF\clQ}f=g\circ_{\hfpt\clQ}f\in\Hom_{\hfpt\msF\clQ}$, 
where the right hand side is a regular free expression of $\clQ$.  

\item \label{it:frecat7} The identity morphism of an object $X\in\Obj_{\hfpt\msF\clQ}$ 
is $\id_{\hfpt\msF\clQ X}=\id_{\hfpt\clQ X}\in\Hom_{\hfpt\msF\clQ}(X,X)$,
where the right hand side is a free identity symbol of $\clQ$.  
  
\end{enumerate}

\noindent 
Furthermore, $\msF$ associates with a morphism $L:\clQ\rightarrow\clR$ of small quivers
the functor $\msF L:\msF\clQ\rightarrow\msF\clR$ of small categories defined in the following manner.

\begin{enumerate} [label=\alph*),start=5]

\item \label{it:frecat8} One has $\msF LX=LX\in\Obj_{\hfpt\msF\clR}$ for each object $X\in\Obj_{\hfpt\msF\clQ}$.

\item \label{it:frecat9} One has $\msF Lf=Lf\in\Hom_{\hfpt\msF\clR}$ for all morphisms $f\in\Hom_{\hfpt\msF\clQ}$,
where the right hand side $Lf$ denotes the free regular expression
yielded by the action of $L$ on the free regular expression $f$. 
  
\end{enumerate}

\end{defi}

\noindent
So, for a quiver $\clQ$, the distinct morphisms of $\msF\clQ$ are just the unbracketed composable strings
$e_{m-1}\circ_{\hfpt\clQ}\ldots\circ_{\hfpt\clQ}e_0$ of edges $e_i\in\Edg_{\hfpt\clQ}$ and the identities
$\id_{\hfpt\clQ X}$. 
The composition operation $\circ_{\hfpt\msF\clQ}$ is free concatenation of strings and the identity assigning
map $\id_{\hfpt\msF\clQ}$ as compositionally neutral. 
The associativity and unitality relations \eqref{catdef1}, \eqref{catdef2} are therefore satisfied by construction.

The free category construction enjoys a universality property stated formally in the following theorem.

\begin{theor} \label{theor:freuniv}
Let $\clQ$ be a small quiver. Then, there exists a small quiver morphism $I:\clQ\rightarrow\msU\msF\clQ$ such that
for any small category $\clC$ and any small quiver morphism  $L:\clQ\rightarrow\msU\clC$, there is a unique functor
$F:\msF\clQ\rightarrow\clC$ obeying $\msU F\circ I=L$. 
\end{theor}

Presentations constitute a useful way of displaying a category as the result of the implementation in a free category
of a set of congruences compatible with morphism composition and identity assignment.
To define this notion, we recall some results from set theory. A binary relation $R$ on a set $A$
is a subset $R\subseteq A\times A$.
The relation $R$ is an equivalence relation if it reflexive, symmetric
and transitive. In that case, the quotient set $A/R$ is defined as
the set of the maximal subsets of $A$ of elements related by $R$.
For every binary relation $R$ on a set $A$, there is a smallest equivalence relation $R_e$ on $A$
such that $R\subseteq R_e$. Below, we shall always tacitly identify any such $R$ with $R_e$
and treat it as an equivalence relation. When the set $A$ is equipped with
an algebraic structure with certain operations, the binary relations $R$ on $A$ considered are as a rule taken
to be congruences, that is to have the property that the action of the operations on elements of $A$ equivalent under
$R$ give as result elements of $A$ equivalent under $R$.

\begin{defi} \label{def:pres}
Let $\clQ$ be a small quiver and let $R$ be a mapping assigning to each object pair $X,Y\in\Obj_{\hfpt\msF\clQ}$
of the free category $\msF\clQ$ a congruence $R_{X,Y}$ on the morphism set $\Hom_{\hfpt\msF\clQ}(X,Y)$.
The category generated by $\clQ$ with relations $R$ is the quotient category $\msF\clQ/R$, that is the
category with object collection $\Obj_{\hfpt\msF\clQ/R}=\Obj_{\hfpt\msF\clQ}$ and morphism sets
$\Hom_{\hfpt\msF\clQ/R}(X,Y)=\Hom_{\hfpt\msF\clQ}(X,Y)/R_{X,Y}$.
A presentation of a category $\clC$ by means of $\clQ$ and $R$ is an isofunctor $F:\clC\rightarrow\msF\clQ/R$. 
\end{defi}

We introduce next the key notion of free strict monoidal category as a refinement of that of free category.

\begin{defi} \label{def:monquiv}
A monoidal quiver is a quiver $\clQ$ whose vertex collection $\Ver_{\hfpt\clQ}$ is 
the free monoid generated by a subcollection $\Ver^0{}_{\hfpt\clQ}\subset\Ver_{\hfpt\clQ}$. 
\end{defi}

\noindent
We denote by $\otimes_{\hfpt\clQ}$, $1_{\hfpt\clQ}$ the monoid multiplication and unit of $\Ver_{\hfpt\clQ}$. 

\begin{defi} \label{def:monquivmor}
A morphism $L:\clQ\rightarrow\clR$ of monoidal quivers is a morphism of the underlying quivers
with the property that $L:\Ver_{\hfpt\clQ}\rightarrow\Ver_{\hfpt\clR}$ is a monoid morphism.
\end{defi}

Small monoidal quivers and monoidal quiver morphisms are the constituents of the 
small monoidal quiver category $\ul{\rm MQuiv}$. 
The composition of morphisms and the identity morphisms are defined as in the
small quiver category $\ul{\rm Quiv}$, making $\ul{\rm MQuiv}$ a subcategory of $\ul{\rm Quiv}$. 



The object set of a strict monoidal category is a monoid, but this monoid is not necessarily free.
For this reason, there is no monoidal counterpart of the quiver forgetful functor $\msU$ defined in 
\cref{def:forgfunct}. However, $\msU$ one can still be applied to strict monoidal categories and functors,
since these latter are special cases of categories and functors. 

Conversely, analogously to the ordinary category case, it is possible to associate a small free strict monoidal
category with any small monoidal quiver and a strict monoidal functor of small free strict monoidal
categories with any morphism of small monoidal quivers through an apposite functor by an appropriate
monoidal generalization of the free category construction detailed above as described next. 

The free strict monoidal categorical data of a monoidal quiver are defined analogously to the 
free categorical data of a quiver. 
Let $\clQ$ be a monoidal quiver. 
The free strict monoidal categorical data of $\clQ$
consist of a set of elementary free symbols and a set of
regular free expressions built with those symbols according to certain rules. 
The elementary free symbols are essentially those of the non monoidal case, viz 
the object and morphism symbols $X$ and $e$
with $X\in\Ver^0{}_{\hfpt\clQ}$ and $e\in\Edg_{\hfpt\clQ}$, the composition symbol $\circ_{\hfpt\clQ}$
and the identity assigning symbol $\id_{\hfpt\clQ}$, 
and a further monoidal multiplication symbol $\otimes_{\hfpt\clQ}$.
The regular free expressions are the bracketed expressions built with all the symbols
$X$, $e$, $\circ_{\hfpt\clQ}$, $\id_{\hfpt\clQ}$, $\otimes_{\hfpt\clQ}$ according to the same
rules which would hold if they  were objects,
morphisms, composition, identity assigning and monoidal multiplication
operations of a strict monoidal category (cf. def. \cref{def:moncat}). 
Further, expressions with the $X$ and $e$ which can be turned into one another by using
the associativity and unitality relations \eqref{catdef1}, \eqref{catdef2}, 
the bifunctoriality relations \ceqref{moncat0}, \ceqref{moncat-1}
and the monoidal associativity and unitality relations \ceqref{moncat1}--\ceqref{moncat4}
are identified.

The free strict monoidal functorial datum of a morphism $L:\clQ\rightarrow\clR$
of two monoidal quivers $\clQ$, $\clR$ is a map
also denoted by $L$ of the free strict monoidal categorical data of $\clQ$ to those of $\clR$ determined 
in the following manner. $L$ maps the free elementary object and morphism
symbols $X$ and $e$ of $\clQ$ into the corresponding free elementary symbols
$LX$ and $Le$ of $\clR$ and turns each regular free expression $f$ of $\clQ$ into
one $Lf$ over $\clR$ by acting as if it was a strict monoidal functor of two strict monoidal categories
(cf. def. \cref{def:monfun}). 


\begin{defi} \label{def:fremoncat}
The free strict monoidal category functor $\msF\msM:\ul{\rm MQuiv}\rightarrow\ul{\rm MCat}$ is the functor
of the small monoidal quiver category $\ul{\rm MQuiv}$ to the small strict monoidal category
category $\ul{\rm MCat}$ so defined. $\msF$ associates with a small monoidal quiver $\clQ$ a small strict monoidal
category $\msF\msM\clQ$ specified as follows. 

\begin{enumerate} [label=\alph*)]

\item \label{it:fremoncat1} The terms of items \cref{it:frecat1}--\cref{it:frecat7} of def. \cref{def:frecat} hold
with $\msF\clQ$ replaced by $\msF\msM\clQ$.

\item \label{it:fremoncat2}
For every morphism pair $f,g\in\Hom_{\hfpt\msF\msM\clQ}$, the monoidal product morphism of $f$ and $g$ is given by
$f\otimes_{\hfpt\msF\msM\clQ}f=f\otimes_{\hfpt\clQ}g\in\Hom_{\hfpt\msF\msM\clQ}$, 
the right hand side being a regular free expression of $\clQ$.  

\end{enumerate}

\noindent
Furthermore, $\msF$ associates with a morphism $L:\clQ\rightarrow\clR$ of small monoidal quivers
the functor $\msF\msM L:\msF\msM\clQ\rightarrow\msF\msM\clR$ of small strict monoidal categories
by requiring the following. 

\begin{enumerate} [label=\alph*),start=3]
  
\item The terms of items \cref{it:frecat8}, \cref{it:frecat7} of def. \cref{def:frecat}
hold  with $\msF$, $\msF L$, $\msF\clR$ replaced by $\msF\msM$, $\msF\msM L$, $\msF\msM\clR$. 

\end{enumerate}

\end{defi}

\noindent
With a few evident modifications and additions, the comments below def. \cref{def:frecat} apply also here.
The associativity and unitality relations
\eqref{catdef1}, \eqref{catdef2}, the bifunctoriality relations \ceqref{moncat0}, \ceqref{moncat-1}
and the monoidal associativity and unitality relations \ceqref{moncat1}--\ceqref{moncat4} 
are satisfied by construction. 

As its ordinary counterpart, the free strict monoidal category construction enjoys a universality property.
To formulate it, one more notion is required. 
A quasi monoidal quiver morphism is a quiver morphism $L:\clQ\rightarrow\clR$, where
$\clQ$ is a monoidal quiver, $\clR$ is a quiver such that $\Ver_{\hfpt\clR}$
is a monoid (not necessarily free) and $L:\Ver_{\hfpt\clQ}\rightarrow\Ver_{\hfpt\clR}$
is a monoid morphism. Clearly, monoidal quiver morphisms are quasi monoidal. 

\begin{theor} \label{theor:fremonuniv}
Let $\clQ$ be a small monoidal quiver. Then, there exists a quasi monoidal small quiver morphism
$I:\clQ\rightarrow\msU\,\msF\msM\clQ$ such that
for any small strict monoidal category $\clC$ and any quasi monoidal small quiver morphism
$L:\clQ\rightarrow\msU\clC$, there is a unique strict monoidal functor
$F:\msF\msM\clQ\rightarrow\clC$ obeying $\msU F\circ I=L$. 
\end{theor}

Presentations constitute a useful way of displaying also a strict monoidal category through the imposition 
on a free strict monoidal category of a set of congruences compatible with morphism composition, see def. \cref{def:pres}. 

\begin{defi}  \label{def:presmon}
Let $\clQ$ be a small monoidal quiver and let $R$ be a mapping assigning to each object pair
$X,Y\in\Obj_{\hfpt\msF\msM\clQ}$ of the free strict monoidal category $\msF\msM\clQ$ a congruence $R_{X,Y}$
on the morphism set $\Hom_{\hfpt\msF\msM\clQ}(X,Y)$.
The strict monoidal category generated by $\clQ$ with relations $R$ is the quotient category $\msF\msM\clQ/R$.
A presentation of a strict monoidal category $\clC$ by means of $\clQ$ and $R$ is a
strict monoidal isofunctor $F:\clC\rightarrow\msF\msM\clQ/R$. 
\end{defi}

\subsection{\textcolor{blue}{\sffamily Graded monads}}\label{subsec:gradmon}

A monad is a monoid internal to the category of endofunctors
of a given category. The monad is graded when it is equipped with a grading provided
by a further monoidal category  \ccite{Mellies:2012mea,Katsumata:2014pms,Mellies:2017pcm,Fujii:2019gim}.

For a strict monoidal category $\clC$ and a category $\clF$, 
a $\clC$--graded monad over $\clF$ is a lax monoidal functor $T:\clC\rightarrow\End(\clF)$. 
In more explicit terms, this notion can be defined as follows.

\begin{defi} \label{def:gradmon}

A graded monad $\clX$ is a sextuple consisting of the following elements:

\begin{enumerate}[label=\alph*)]

\item a grading strict monoidal category $\clC_{\hfpt\clX}$, the monad's grading category;

\item a category $\clF_{\hfpt\clX}$, the monad's target category;
 
\item a monad's endofunctor $T_{\hfpt\clX X}:\clF_{\hfpt\clX}\rightarrow\clF_{\hfpt\clX}$
for each object $X\in\Obj_{\hfpt\clC_{\hfpt\clX}}$;

\item a monad's natural transformation $T_{\hfpt\clX f}:T_{\hfpt\clX X}\Rightarrow T_{\hfpt\clX Y}$
for every pair of objects $X,Y\in\Obj_{\hfpt\clC_{\hfpt\clX}}$
and morphism $f\in\Hom_{\hfpt\clC_{\hfpt\clX}}(X,Y)$;

\item a natural transformation
$\mu_{\hfpt\clX X,Y}:T_{\hfpt\clX X}\circ T_{\hfpt\clX Y} \Rightarrow T_{\hfpt\clX X\otimes_{\hfpt\clC_{\hfpt\clX}}Y}$
for every pair of objects $X,Y\in\Obj_{\hfpt\clC_{\hfpt\clX}}$, the monad's multiplication;

\item a natural transformation $\eta_{\hfpt\clX}:\id_{\hfpt\clF_{\hfpt\clX}}\Rightarrow T_{\hfpt\clX 1_{\hfpt\clC_{\hfpt\clX}}}$,
the monad's unit.
  
\end{enumerate}

\noindent
These are required to satisfy the following conditions for all $A\in\Obj_{\hfpt\clF_{\hfpt\clX}}$.

\begin{enumerate} 

\item For $X\in\Obj_{\hfpt\clC_{\hfpt\clX}}$, it holds that \hphantom{xxxxxxxxx}
\begin{equation}
\label{gradmon1}
T_{\hfpt\clX \hfpt\id_{\clC_{\hfpt\clX} X}\! A}=\id_{\hfpt\clF_{\hfpt\clX} T_{\hfpt\clX X}A}.
\end{equation}

\item For composable $f,g\in\Hom_{\hfpt\clC_{\hfpt\clX}}$, \hphantom{xxxxxxxxx}
\begin{equation}
\label{gradmon2}
T_{\hfpt\clX g\hfpt\circ_{\hfpt\clC_{\hfpt\clX}} fA}=T_{\hfpt\clX gA}\circ_{\hfpt\clF_{\hfpt\clX}} T_{\hfpt\clX fA}.
\end{equation}

\item For $X,Y,Z\in\Obj_{\hfpt\clC_{\hfpt\clX}}$, one has \hphantom{xxxxxxxxx}
\begin{equation}
\label{gradmon4}
\mu_{\hfpt\clX X,Y\hfpt\otimes_{\hfpt\clC_{\hfpt\clX}}\hfpt ZA}\circ_{\hfpt\clF_{\hfpt\clX}} T_{\hfpt\clX X}\mu_{\hfpt\clX Y,ZA}
=\mu_{\hfpt\clX X\hfpt\otimes_{\hfpt\clC_{\hfpt\clX}}\hfpt Y,ZA}\circ_{\hfpt\clF_{\hfpt\clX}} \mu_{\hfpt\clX X,YT_{\hfpt\clX Z}A}.
\end{equation}

\item For $X\in\Obj_{\hfpt\clC_{\hfpt\clX}}$, one has \hphantom{xxxxxxxxxxxx}
\begin{equation}
\label{gradmon5}
\mu_{\hfpt\clX 1_{\hfpt\clC_{\hfpt\clX}},XA}\circ_{\hfpt\clF_{\hfpt\clX}} \eta_{\hfpt\clX T_{\hfpt\clX X}A}
=\mu_{\hfpt\clX X,1_{\hfpt\clC_{\hfpt\clX}}A}\circ_{\hfpt\clF_{\hfpt\clX}} T_{\hfpt\clX X}\eta_{\hfpt\clX A}.
\end{equation}

\item For $X,Y,U,V\in\Obj_{\hfpt\clC_{\hfpt\clX}}$, $f\in\Hom_{\hfpt\clC_{\hfpt\clX}}(X,U)$, $g\in\Hom_{\hfpt\clC_{\hfpt\clX}}(Y,V)$,
\begin{equation}
\label{gradmon3}
T_{\hfpt\clX f\hfpt\otimes_{\hfpt\clC_{\hfpt\clX}}\hfpt gA}\circ_{\hfpt\clF_{\hfpt\clX}} \mu_{\hfpt\clX X,YA}
=\mu_{\hfpt\clX U,V A}\circ_{\hfpt\clF_{\hfpt\clX}} T_{\hfpt\clX U}T_{\hfpt\clX gA}\circ_{\hfpt\clF_{\hfpt\clX}} T_{\hfpt\clX fT_{\hfpt\clX Y}A}. 
\end{equation}

\end{enumerate}
\end{defi}

There exists a notion of morphism of graded monads.
Since a graded monad is ultimately a lax monoidal functor $T:\clC\rightarrow\End(\clF)$,
it would seem natural to define a morphism of two
graded monads $T:\clC\rightarrow\End(\clF)$, $U:\clC\rightarrow\End(\clF)$ 
as a monoidal natural transformation $\phi:T\Rightarrow Y$.
A notion of this kind, albeit natural, however constrains the intervening monads to share the
underlying grading and target categories $\clC$, $\clF$ and for this reason is not sufficient to cover the multiple 
instances of maps of graded monads treated in the present work.
We propose as a substitute the following more general definition. 

\begin{defi} \label{def:gradmonmor} 
A morphism $F:\clX\rightarrow\clY$
of graded monads is a triple consisting of the following elements:

\begin{enumerate}[label=\alph*)]

\item a strict monoidal functor $\theta_F:\clC_{\hfpt\clX}\rightarrow\clC_{\hfpt\clY}$
  of the grading categories;

\item a functor $\varPhi_F:\clF_{\hfpt\clX}\rightarrow\clF_{\hfpt\clY}$ of the target categories;

\item a natural transformation $\phi_{FX}:\varPhi_F\circ T_{\hfpt\clX X}\Rightarrow T_{\hfpt\clY \theta_F X}\circ \varPhi_F$
associated with each object $X\in\Obj_{\hfpt\clC_{\hfpt\clX}}$, called entwining map at $X$. 
  
\end{enumerate} 

\noindent
These are required to satisfy the following conditions for all $A\in\Obj_{\hfpt\clF_{\hfpt\clX}}$.

\begin{enumerate}

\item For $X,Y\in\Obj_{\hfpt\clC_{\hfpt\clX}}$, $f\in\Hom_{\hfpt\clC_{\hfpt\clX}}(X,Y)$,
\begin{equation}
\label{gradmon6}
T_{\hfpt\clY \theta_F f\hfpt\varPhi_F\mhfpt A}\circ_{\clF_{\hfpt\clY}}\phi_{FXA}=\phi_{FYA}\circ_{\clF_{\hfpt\clY}}\varPhi_F T_{\hfpt\clX fA}.
\end{equation}

\item For $X,Y\in\Obj_{\hfpt\clC_{\hfpt\clX}}$, one has 
\begin{equation}
  \label{gradmon7}
  \hspace{-1mm}
\phi_{FX\hfpt\otimes_{\hfpt\clC_{\hfpt\clX}}\hfpt YA}\circ_{\clF_{\hfpt\clY}}\varPhi_F\mu_{\hfpt\clX X,YA}
=\mu_{\hfpt\clY\theta_F X,\theta_F Y\hfpt\varPhi_F\mhfpt A}\circ_{\clF_{\hfpt\clY}}T_{\hfpt\clY\theta_FX}\phi_{FYA}
\circ_{\clF_{\hfpt\clY}}\phi_{FXT_{\hfpt\clX Y}A}.
\end{equation}

\item It holds that \hphantom{xxxxxxxxxxxxxxxxx}
\begin{equation}
\label{gradmon8}
\eta_{\hfpt\clY\varPhi_F A}=\phi_{F1_{\hfpt\clC_{\hfpt\clX}}A}\circ_{\clF_{\hfpt\clY}}\varPhi_F \eta_{\hfpt\clX A}.
\end{equation}

\end{enumerate}
\end{defi}

\noindent
As one might expect, graded monads and graded monad morphisms constitute a category.

\begin{theor} \label{theor:catgrmnd}
Graded monads and graded monad morphisms form a distinguished category
$\ul{\rm GM}$. The composite of the morphism pair 
$F:\clX\rightarrow\clY$, $G:\clY\rightarrow\clZ$ is the morphism 
$G\circ F:\clX\rightarrow\clZ$ specified by the two functor composites 
$\theta_{G\circ F}=\theta_G\circ\theta_F$, $\varPhi_{G\circ F}=\varPhi_G\circ\varPhi_F$ 
and the component morphisms \hphantom{xxxxxxxxxxxxxxxxx}
\begin{equation}
\label{gradmon9}
\phi_{G\circ FXA}=\phi_{G\theta_F X\varPhi_F A}\circ_{\clF_{\hfpt\clY}'}\varPhi_G\phi_{FXA}
\end{equation}
with $X\in\Obj_{\hfpt\clC_{\hfpt\clX}}$, $A\in\Obj_{\hfpt\clF_{\hfpt\clX}}$.
The identity morphism $\id_{\hfpt\clX}$ of a graded monad $\clX$ is the morphism
with $\theta_{\id_{\hfpt\clX}}=\id_{\hfpt\clC_{\hfpt\clX}}$,
$\varPhi_{\id_{\hfpt\clX}}=\id_{\hfpt\clF_{\hfpt\clX}}$ identity functors and component morphisms
\begin{equation}
\label{gradmon10}
\phi_{\id_{\hfpt\clX}XA}=\id_{\hfpt\clF_{\hfpt\clX} T_{\hfpt\clX X}A}
\end{equation}
with $X\in\Obj_{\hfpt\clC_{\hfpt\clX}}$, $A\in\Obj_{\hfpt\clF_{\hfpt\clX}}$ again.
\end{theor}

\begin{proof}
The proof is lengthy but straightforward. We limit ourselves to provide a sketch of it.
To begin with, one shows that the composite of two graded monad morphisms as defined in \ceqref{gradmon9}
is again a morphism. Next, one checks that the identity of a graded monad defined in \ceqref{gradmon10}
is also a morphism. Finally, one proves the associativity and unitality of composition 
verifying that relations \ceqref{catdef1}, \ceqref{catdef2} are satisfied.
\end{proof}

Paramonoidal graded monads and monad morphisms, which we introduce next, are 
data for the construction of the kind of
graded monads and monad morphisms relevant in the present work. They are not to be confused
with monoidal monads \ccite{Koch:1972msm}.

\begin{defi} \label{def:gradmonqvr}
A paramonoidal graded monad $\clK$ is a quadruple comprising: 

\begin{enumerate}[label=\alph*)]

\item a grading strict monoidal category $\clC_{\hfpt\clK}$;

\item a target strict monoidal category $\clF_{\hfpt\clK}$ such that $\clC_{\hfpt\clK}$ is a 
subcategory of $\clF_{\hfpt\clK}$ if the monoidal structures of both categories are disregarded;

  
\item a multiplication morphism $\mu_{\clK X,Y}\in\Hom_{\hfpt\clF_{\hfpt\clK}}(X\otimes_{\hfpt\clF_{\hfpt\clK}} Y,X\otimes_{\hfpt\clC_{\hfpt\clK}}Y)$
for every object pair $X,Y\in\Obj_{\hfpt\clC_{\hfpt\clK}}$;

\item a unit morphism $\eta_{\hfpt\clK}\in\Hom_{\hfpt\clF_{\hfpt\clK}}(1_{\hfpt\clF_{\hfpt\clK}},1_{\hfpt\clC_{\hfpt\clK}})$.
  
\end{enumerate}

\noindent
These are required to satisfy the following relations.

\begin{enumerate}

\item For $X,Y,Z\in\Obj_{\hfpt\clC_{\hfpt\clK}}$, \hphantom{xxxxxxxxx}
\begin{align}
\label{gradmon12}
&\mu_{\hfpt\clK X,Y\hfpt\otimes_{\hfpt\clC_{\hfpt\clK}}\hfpt Z}\circ_{\hfpt\clF_{\hfpt\clK}}
(\id_{\hfpt\clF_{\hfpt\clK} X}\otimes_{\hfpt\clF_{\hfpt\clK}}\mu_{\hfpt\clK Y,Z})
\\
\nonumber
&\hspace{4.5cm}=\mu_{\hfpt\clK X\hfpt\otimes_{\hfpt\clC_{\hfpt\clK}}\hfpt Y,Z}\circ_{\hfpt\clF_{\hfpt\clK}}
(\mu_{\hfpt\clK X,Y}\otimes_{\hfpt\clF_{\hfpt\clK}}\id_{\hfpt\clF_{\hfpt\clK} Z}).
\end{align}

\item For $X\in\Obj_{\hfpt\clC_{\hfpt\clK}}$, it holds that \hphantom{xxxxxxxxxxxx}
\begin{equation}
\label{gradmon13}
\mu_{\hfpt\clK 1_{\hfpt\clC_{\hfpt\clK}},X}\circ_{\hfpt\clF_{\hfpt\clK}}(\eta_{\hfpt\clK} \otimes_{\hfpt\clF_{\hfpt\clK}}\id_{\hfpt\clF_{\hfpt\clK} X})
=\mu_{\hfpt\clK X,1_{\hfpt\clC_{\hfpt\clK}}}\circ_{\hfpt\clF_{\hfpt\clK}}(\id_{\hfpt\clF_{\hfpt\clK} X}\otimes_{\hfpt\clF_{\hfpt\clK}}\eta_{\hfpt\clK })
=\id_{\hfpt\clF_{\hfpt\clK} X}.
\end{equation}

\item For $X,Y,U,V\in\Obj_{\hfpt\clC_{\hfpt\clK}}$, $f\in\Hom_{\hfpt\clC_{\hfpt\clK}}(X,U)$, $g\in\Hom_{\hfpt\clC_{\hfpt\clK}}(Y,V)$, one has 
\begin{equation}
\label{gradmon11}
(f\otimes_{\hfpt\clC_{\hfpt\clK}}g)\circ_{\hfpt\clF_{\hfpt\clK}} \mu_{\hfpt\clK X,Y}
=\mu_{\hfpt\clK U,V}\circ_{\hfpt\clF_{\hfpt\clK}}(f\otimes_{\hfpt\clF_{\hfpt\clK}} g).
\end{equation}

\end{enumerate}

\end{defi}

\begin{defi} \label{def:gradmonmap}
A morphism $P:\clK\rightarrow\clL$ 
of paramonoidal graded monads is a triple consisting of the following elements:

\begin{enumerate}[label=\alph*)]

\item a strict monoidal functor $\theta_P:\clC_{\hfpt\clK}\rightarrow\clC_{\hfpt\clL}$ of the grading categories;

\item a strict monoidal functor $\varPhi_P:\clF_{\hfpt\clK}\rightarrow\clF_{\hfpt\clL}$ of the target categories;

\item for each object $X\in\Obj_{\hfpt\clC_{\hfpt\clK}}$, an entwining morphism
  $\phi_{PX}\in\Hom_{\hfpt\clF_{\hfpt\clL}}(\varPhi_PX,\theta_PX)$.
  
\end{enumerate} 

\noindent
These are required to satisfy the following conditions. 

\begin{enumerate}

\item For $X,Y\in\Obj_{\hfpt\clC_{\hfpt\clK}}$, $f\in\Hom_{\hfpt\clC_{\hfpt\clK}}(X,Y)$,
\begin{equation}
\label{gradmon18}
\phi_{PY}\circ_{\hfpt\clF_{\hfpt\clL}}\varPhi_Pf=\theta_Pf\circ_{\hfpt\clF_{\hfpt\clL}}\phi_{PX}.
\end{equation}

\item For $X,Y\in\Obj_{\hfpt\clC_{\hfpt\clK}}$, one has 
\begin{equation}
\label{gradmon19}
\phi_{PX\hfpt\otimes_{\hfpt\clC_{\hfpt\clK}}\hfpt Y}\circ_{\clF_{\hfpt\clL}}\varPhi_P\mu_{\hfpt\clK X,Y}
=\mu_{\hfpt\clL\theta_PX,\theta_PY}\circ_{\clF_{\hfpt\clL}}(\phi_{PX}\otimes_{\hfpt\clF_{\hfpt\clL}}\phi_{PY}).
\end{equation}

\item It holds that \hphantom{xxxxxxxxxxxxxxxxxxxxx}
\begin{equation}
\label{gradmon20}
\eta_{\hfpt\clL}=\phi_{P1_{\hfpt\clC_{\hfpt\clK}}}\circ_{\clF_{\hfpt\clL}}\varPhi_P\eta_{\hfpt\clK}.
\end{equation}

\end{enumerate}

\end{defi}

\noindent
Paramonoidal graded monads and paramonoidal graded monad morphisms also make up a category. 

\begin{theor} \label{theor:catgrmndset}
Paramonoidal graded monads and paramonoidal graded monad morphisms are the constituents of a category 
$\ul{\rm pMGM}$. The composite of two morphisms $P:\clK\rightarrow\clL$, $Q:\clL\rightarrow\clM$ is the morphism
$Q\circ P:\clK\rightarrow\clM$ specified by the functor composites 
$\theta_{Q\circ P}=\theta_Q\circ\theta_P$, $\varPhi_{Q\circ P}=\varPhi_Q\circ\varPhi_p$ 
and the component morphisms 
\begin{equation}
\label{gradmon22}
\phi_{Q\circ PX}=\phi_{Q\theta_P X}\circ_{\hfpt\clF_{\hfpt\clL}}\varPhi_Q\phi_{PX}
\end{equation}
for $X\in\Obj_{\hfpt\clC_{\hfpt\clX}}$. The identity morphism $\id_{\hfpt\clK}$
of a paramonoidal graded monad $\clK$ is the morphism with $\theta_{\id_{\hfpt\clK}}=\id_{\hfpt\clC_{\hfpt\clK}}$,
$\varPhi_{\id_{\hfpt\clK}}=\id_{\hfpt\clF_{\hfpt\clK}}$ identity functors and component morphisms
\begin{equation}
\label{gradmon23}
\phi_{\id_{\hfpt\clK}X}=\id_{\hfpt\clF_{\hfpt\clK} X}
\end{equation}
for $X\in\Obj_{\hfpt\clC_{\hfpt\clK}}$.
\end{theor}

\begin{proof}
The proof proceeds along the lines of that of theor. \cref{theor:catgrmnd}.
Again, we restrict ourselves to furnish a sketch of it.
To begin with, one shows that the composite of two paramonoidal graded monad 
morphisms as defined in \ceqref{gradmon22}
is again a morphism. Next, one verifies that the identity of a paramonoidal graded monad 
defined in \ceqref{gradmon23}
is also a morphism. Finally, one proves the associativity and unitality of composition 
by checking that properties \ceqref{catdef1}, \ceqref{catdef2} hold.
\end{proof}

The construction of graded monads and graded monad morphisms from 
paramonoidal graded monads and paramonoidal graded monad morphisms mentioned earlier
proceeds consistently through a distinguished functor relating the corresponding categories
$\ul{\rm pMGM}$, $\ul{\rm GM}$.

\begin{theor} \label{theor:gradmonqvr}
There is a functor $\msG\msM:\ul{\rm pMGM}\rightarrow\ul{\rm GM}$ of the categories
$\ul{\rm pMGM}$, $\ul{\rm GM}$ injective on objects and morphisms described thusly. 
$\msG\msM$ associates with a paramonoidal graded monad $\clK$ the graded monad
$\msG\msM\clK$ constructed as follows.

\begin{enumerate}[label=\alph*)]

\item \label{it:gradmonqvr1} Set $\clC_{\hfpt\msG\msM\clK}=\clC_{\hfpt\clK}$ and $\clF_{\hfpt\msG\msM\clK}=\clF_{\hfpt\clK}$.

\item \label{it:gradmonqvr2} For any object $X\in\Obj_{\hfpt\clC_{\hfpt\clK}}$, put 
{\allowdisplaybreaks
\begin{align}
\label{gradmon14}
T_{\hfpt\msG\msM\clK X}A&=X\otimes_{\hfpt\clF_{\hfpt\clK}}A,
\\
\nonumber
T_{\hfpt\msG\msM\clK X}\chi&=\id_{\hfpt\clF_{\hfpt\clK} X}\otimes_{\hfpt\clF_{\hfpt\clK}}\chi
\end{align}
}
\!\!with $A\in\Obj_{\hfpt\clF_{\hfpt\clX}}$ and $\chi\in\Hom_{\hfpt\clF_{\hfpt\clK}}$.

\item \label{it:gradmonqvr3} For any morphism $f\in\Hom_{\hfpt\clC_{\hfpt\clK}}$, set 
\begin{equation}
\label{gradmon15}
T_{\hfpt\msG\msM\clK fA}=f\otimes_{\hfpt\clF_{\hfpt\clK}}\id_{\hfpt\clF_{\hfpt\clK} A}
\end{equation}
with $A\in\Obj_{\hfpt\clF_{\hfpt\clK}}$.

\item \label{it:gradmonqvr4} For any two objects $X,Y\in\Obj_{\hfpt\clC_{\hfpt\clK}}$, define 
\begin{equation}
\label{gradmon16}
\mu_{\hfpt\msG\msM\clK X,YA}=\mu_{\hfpt\clK X,Y}\otimes_{\hfpt\clF_{\hfpt\clK}}\id_{\hfpt\clF_{\hfpt\clK} A}
\end{equation}
with $A\in\Obj_{\hfpt\clF_{\hfpt\clK}}$. 

\item \label{it:gradmonqvr5} Set furthermore \hphantom{xxxxxxxxxxxxxxxx}
\begin{equation}
\label{gradmon17}
\eta_{\hfpt\msG\msM\clK A}=\eta_{\hfpt\clK}\otimes_{\hfpt\clF_{\hfpt\clK}}\id_{\hfpt\clF_{\hfpt\clK} A}
\end{equation}
with $A\in\Obj_{\hfpt\clF_{\hfpt\clK}}$. 

\end{enumerate}
The graded monad morphism $\msG\msM P:\msG\msM\clK\rightarrow\msG\msM\clL$ associated  by $\msG\msM$
with a paramonoidal graded monad morphism
$P:\clK\rightarrow\clL$ is specified as follows. 

\begin{enumerate}[label=\alph*),start=6]

\item \label{it:gradmonqvr6} Set $\theta_{\msG\msM P}=\theta_P$, $\varPhi_{\msG\msM P}=\varPhi_P$. 

\item \label{it:gradmonqvr7} For any object $X\in\Obj_{\hfpt\clC_{\hfpt\clK}}$, put \hphantom{xxxxxxxxxxxxx}
\begin{equation}
\label{gradmon21}
\phi_{\msG\msM PXA}=\phi_{PX}\otimes_{\hfpt\clF_{\hfpt\clL}}\id_{\hfpt\clF_{\hfpt\clL}\varPhi_PA}
\end{equation}
with $A\in\Obj_{\hfpt\clF_{\hfpt\clK}}$.
  
\end{enumerate}
\end{theor}

\begin{proof}
The proof involves a long series of verifications. \pagebreak We provide only a summary below.
First, one shows that the data list $(\clC_{\hfpt\msG\msM\clK},\clF_{\hfpt\msG\msM\clK},T_{\hfpt\msG\msM\clK},
\mu_{\hfpt\msG\msM\clK},\eta_{\hfpt\msG\msM\clK})$ given in items {\it \cref{it:gradmonqvr1}--\cref{it:gradmonqvr5}}
is indeed a graded monad by checking that relations \ceqref{gradmon1}--\ceqref{gradmon3} are satisfied
by virtue of \ceqref{gradmon12}--\ceqref{gradmon11}.
Second, one proves that the data list
$(\theta_{\msG\msM P},\varPhi_{\msG\msM P},\phi_{\msG\msM P})$ furnished in items  {\it \cref{it:gradmonqvr6}, \cref{it:gradmonqvr7}}
is a graded monad morphism, as required, by verifying that relations
\ceqref{gradmon6}--\ceqref{gradmon8} are also met. Third, one demonstrates that $\msG\msM$
enjoys the essential functor properties \ceqref{functdef1}, \ceqref{functdef2} by checking that
\ceqref{gradmon9}, \ceqref{gradmon10} follow from \ceqref{gradmon22}, \ceqref{gradmon23}.
The injectivity on objects and morphisms of $\msG\msM$ follows from \ceqref{gradmon16},  \ceqref{gradmon17}
and  \ceqref{gradmon18} by setting $A=1_{\hfpt\clF_{\hfpt\clK}}$. 
\end{proof}

\noindent
The upshot of the above analysis is that although paramonoidal graded monads and monad morphisms are not
strictly speaking graded monads and monad morphisms, the category of the former, $\ul{\rm pMGM}$,
is isomorphic to the subcategory $\msG\msM\ul{\rm pMGM}$
of the category of the latter, $\ul{\rm GM}$, by virtue of theor. \cref{theor:gradmonqvr},
where $\msG\msM\ul{\rm pMGM}$ is the range category of the functor $\msG\msM$. 
The graded monads and monad morphisms relevant in the present work are precisely those
from $\msG\msM\ul{\rm pMGM}$.


We explain next how the results expounded above are applied to the theory of concrete
graded $\varOmega$ monads and morphisms thereof
worked out in subsect. \cref{subsec:dcatjoint}.
The set--up considered there can be modelled abstractly 
within the above categorical framework as a certain quiver $\fkQ$ whose
vertex collection $\Ver_\fkQ$ is a set of paramonoidal graded monads
and whose edge set $\Edg_\fkQ$ consists of morphisms
between those monads with the following property. 
There is an encompassing strict monoidal category $\clU$ such that
for each monad $\clK\in\Ver_\fkQ$, $\clF_\clK=\clU$ is the target category of $\clK$
and for each monad pair $\clK,\clL\in\Ver_\fkQ$ and morphism $P\in\Edg_\fkQ(\clK,\clL)$,
$\varPhi_P=\id_{\hfpt\clU}$ is the target category functor of $P$.
The careful construction of the quiver $\fkQ$ requires a number of steps.

{\it
\begin{enumerate}[label=\Alph*)]

\item \label{it:prot1} Assignment of the grading categorical and functorial data.

\end{enumerate}
}

\noindent
These data are organized in a quiver $\Delta$ whose vertex collection $\Ver_{\hfpt\Delta}$ is a set of small strict
monoidal categories and whose edge set $\Edg_{\hfpt\Delta}$ is formed by strict monoidal functors between those
categories. 

{\it
\begin{enumerate}[label=\Alph*),start=2]

\item \label{it:prot2} Assignment of multiplication, unit and entwining morphism data compatible with the grading data. 

\end{enumerate}
}

\noindent
These data are organized in a monoidal quiver $\clQ$. 
Let $\clA_{\hfpt\clQ}$  be the free monoid generated by the
disjoint union $\clA^0{}_{\hfpt\clQ}=\bigsqcup_{\hfpt\clC\in\Ver_{\hfpt\Delta}}\Obj_{\hfpt\clC}$
of the object sets of the categories of $\Ver_{\hfpt\Delta}$
and $\otimes_{\hfpt\clQ}$ and $1_{\hfpt\clQ}$ the monoidal product and unit of $\clA_{\hfpt\clQ}$.
Then, $\clQ$ is specified in the items set forth  below. 

{\it
\begin{enumerate}[label=\alph*)]  

\item \label{it:minqvr1} {\rm The vertex set of $\clQ$ is $\Ver_{\hfpt\clQ}=\clA_{\hfpt\clQ}$.}

\item \label{it:minqvr2} {\rm The edge set $\Edg_{\hfpt\clQ}$ of $\clQ$ contains all and only the edges
listed in the next points.} 

\item \label{it:minqvr3} {\rm There is a distinct edge $f\in\Edg_{\hfpt\clQ}(X,Y)$ for every category
$\clC\in\Ver_{\hfpt\Delta}$, pair of objects $X,Y\in\Obj_{\hfpt\clC}$ and morphisms $f\in\Hom_{\hfpt\clC}(X,Y)$,
except for when $X=Y$ and $f=\id_{\hfpt\clC X}$.}

\item \label{it:minqvr4} {\rm There is a multiplication edge
$\mu_{\hfpt\clC X,Y}\in\Edg_{\hfpt\clQ}(X\otimes_{\hfpt\clQ}Y,X\otimes_{\hfpt\clC}Y)$
for each category $\clC\in\Ver_{\hfpt\Delta}$ and object pair $X,Y\in\Obj_{\hfpt\clC}$.}

\item \label{it:minqvr5} {\rm There is a unit edge
$\eta_{\hfpt\clC}\in\Edg_{\hfpt\clQ}(1_{\hfpt\clQ},1_{\hfpt\clC})$ for every category $\clC\in\Ver_{\hfpt\Delta}$.}

\item  \label{it:minqvr6} {\rm There is a entwining edge $\phi_{\theta X}\in\Edg_{\hfpt\clQ}(X,\theta X)$
for every functor $\theta\in\Edg_{\hfpt\Delta}$ and  object $X\in\Obj_{\hfpt\clC}$,
where $\clC\in\Ver_{\hfpt\Delta}$ is the source category of $\theta$.}

\end{enumerate}
}

\noindent
As explained in app. \!\cref{subsec:frecat} (cf. defs. \cref{def:frecat}, \cref{def:fremoncat}),
there exists a free strict monoidal category $\msF\msM\clQ$ associated with
the monoidal quiver $\clQ$. This category is however too loose for our purposes.
A refinement of the free strict monoidal category construction is required here. 
The elementary free symbols and regular free expressions as well as the
rules used to assemble the former into the latter 
are the same as those of the free strict monoidal category construction. 
The elementary free symbols are thus the object symbols 
$X\in\clA^0{}_{\hfpt\clQ}$ and morphism symbols
$f$, $\mu_{\hfpt\clC X,Y}$, $\eta_{\hfpt\clC}$, $\phi_{\theta X}\in\Edg_{\hfpt\clQ}$, the composition symbol $\circ_{\hfpt\clQ}$,
the identity assigning symbol $\id_{\hfpt\clQ}$, 
and the monoidal multiplication symbol $\otimes_{\hfpt\clQ}$; 
the regular free expressions are the bracketed expressions built with the symbols $X$, $f$, $\mu_{\hfpt\clC X,Y}$,
$\eta_{\hfpt\clC}$, $\phi_{\theta X}$, $\circ_{\hfpt\clQ}$, $\id_{\hfpt\clQ}$, $\otimes_{\hfpt\clQ}$
according to the same rules which would hold if they were objects,
morphisms, composition, identity assigning and monoidal multiplication
operations of a strict monoidal category, respectively. 
The novelty is that now we identify expressions built with the $X$, 
$f$, $\mu_{\hfpt\clC X,Y}$, $\eta_{\hfpt\clC}$, $\phi_{\theta X}$
which can be turned into one another not only by using
the associativity and unitality relations \eqref{catdef1}, \eqref{catdef2}, 
the bifunctoriality relations \ceqref{moncat0}, \ceqref{moncat-1}
and the monoidal associativity and unitality relations \ceqref{moncat1}--\ceqref{moncat4}
with $\clC$ replaced by $\clQ$ but also the following further specific ones: 
relations \ceqref{gradmon12}--\ceqref{gradmon11} with 
$\clC_{\hfpt\clK}$, $\clF_{\hfpt\clK}$, $\mu_{\hfpt\clK X,Y}$ and $\eta_{\hfpt\clK}$
replaced by $\clC$, $\clQ$, $\mu_{\hfpt\clC X,Y}$ and $\eta_{\hfpt\clC}$ 
for each category $\clC\in\Ver_{\hfpt\Delta}$;
relations \ceqref{gradmon18}--\ceqref{gradmon20}
with $\clC_{\hfpt\clK}$, $\clC_{\hfpt\clL}$ and $\theta_P$ replaced by $\clC$, $\clD$ and
$\theta$, further $\clF_{\hfpt\clK}$, $\clF_{\hfpt\clL}$ replaced by $\clQ$ with $\varPhi_P$ deleted
and finally $\mu_{\hfpt\clK X,Y}$, $\eta_{\hfpt\clK}$, $\mu_{\hfpt\clL U,V}$,
$\eta_{\hfpt\clL}$ and $\phi_{P X}$ substituted by $\mu_{\hfpt\clC X,Y}$, $\eta_{\hfpt\clC}$, $\mu_{\hfpt\clD U,V}$,
$\eta_{\hfpt\clD}$ and $\phi_{\theta X}$ for every functor $\theta\in\Edg_{\hfpt\Delta}$
with source, target categories $\clC,\clD\in\Ver_{\hfpt\Delta}$.
A category $\msF\msM_{GM}\clQ$ is then defined by implementing 
the terms of items {\it \cref{it:frecat1}--\cref{it:frecat7}} of def. \cref{def:frecat} 
with $\msF\clQ$ exchanged with $\msF\msM_{GM}\clQ$ and those of item {\it \cref{it:fremoncat2}}
of def. \cref{def:fremoncat} with $\msF\msM\clQ$ replaced by $\msF\msM_{GM}\clQ$.
$\msF\msM_{GM}\clQ$ is strict monoidal
and has all the requirements which the $\mu_{\hfpt\clC X,Y}$, $\eta_{\hfpt\clC}$
$\phi_{\theta X}$ would obey if they were the multiplication, unit and entwining morphisms
in a paramonoidal graded monad and monad morphism already satisfied by construction. 

{\it
\begin{enumerate}[label=\Alph*),start=3]

\item \label{it:prot3} Assignment of a background category for the morphism data. 

\end{enumerate}
}

\noindent
A background category $\clB$ for the morphism data is a category with the following properties. 

{\it
\begin{enumerate}[label=\alph*)] 

\item \label{it:bkgr1} {\rm Each object $A\in\Obj_{\msF\msM_{GM}\clQ}$ of $\msF\msM_{GM}\clQ$
  can be identified with an object $A\in\Obj_{\hfpt\clB}$ of $\clB$ 
  denoted by the same symbol.}

\item \label{it:bkgr2} {\rm Each morphism $h\in\Hom_{\msF\msM_{GM}\clQ}(A,B)$ of $\msF\msM_{GM}\clQ$
 can be identified with a morphism $h\in\Hom_{\hfpt\clB}(A,B)$ of $\clB$ 
 denoted again by the same symbol. In particular, each identity morphism $\id_{\msF\msM_{GM}\clQ A}$ of $\msF\msM_{GM}\clQ$
is identified with its counterpart $\id_{\hfpt\clB A}$ of $\clB$.}

\item \label{it:bkgr3} {\rm Each category $\clC\in\Ver_\Delta$ is a subcategory of $\clB$.}

\item \label{it:bkgr4} {\rm The $f$, $\mu_{\hfpt\clC X,Y}$, $\eta_{\hfpt\clC}$, $\phi_{\theta X}$
as morphisms $\clB$ obey the relations \ceqref{gradmon12}--\ceqref{gradmon11} and \ceqref{gradmon18}--\ceqref{gradmon20}
with the same notational changes mentioned earlier except for replacing
$\clF_{\hfpt\clK}$, $\clF_{\hfpt\clL}$ with $\clB$. }
  
\end{enumerate}
}

\noindent
The encompassing category $\clU$ mentioned earlier is the strict monoidal category yielded by imposing on
$\msF\msM_{GM}\clQ$ the congruences further on.

{\it
\begin{enumerate}[label=\roman*)] 

\item \label{it:bkgr5} {\rm For every composable pair $h,k\in\Hom_{\msF\msM_{GM}\clQ}$,
  $k\circ_{\hfpt\msF\msM_{GM}\clQ}h=k\circ_{\hfpt\clB}h$.}

\item \label{it:bkgr6} {\rm For every object $A\in\Obj_{\msF\msM_{GM}\clQ}$,
  $\id_{\hfpt\msF\msM_{GM}\clQ A}=\id_{\hfpt\clB A}$.}

\end{enumerate}
}

\noindent
$\clU$ is evidently a non monoidal subcategory of $\clB$. Further, for each category $\clC\in\Ver_\Delta$, the quadruple 
$\clS_{\hfpt\clC}=(\clC,\clU,\mu_{\hfpt\clC},\eta_{\hfpt\clC})$ is a paramonoidal graded monad
and for each functor $\theta\in\Edg_\Delta(\clC,\clD)$, the triple $P_\theta=(\theta,\id_{\hfpt\clB},\phi_\theta)$
is a paramonoidal graded monad morphisms $P_\theta:\clS_{\hfpt\clC}\rightarrow\clS_{\hfpt\clD}$. The
monads $\clS_{\hfpt\clC}$ and monad morphisms $P_\theta$
constitute respectively the edges and the vertices of quiver $\fkQ$ mentioned at the beginning.

\vfil\eject

\vspace{.5cm}

\noindent
\markright{\textcolor{blue}{\sffamily Acknowledgements}}

\noindent
\textcolor{blue}{\sffamily Acknowledgements.} 
The author acknowledges financial support from INFN Research Agency
under the provisions of the agreement between Alma Mater Studiorum University of Bologna and INFN. 
He is grateful to the organizers of the Conference PAFT24 - Quantum Gravity and Information, where part
of this work was done, for hospitality and support. Most of the algebraic calculations presented in this paper
have been carried out employing the WolframAlpha computational platform.

\vspace{.5cm}

\noindent
\textcolor{blue}{\sffamily Conflict of interest statement.}
The author declares that the results of the current study do not involve any conflict of interest.

\vspace{.5cm}

\noindent
\textcolor{blue}{\sffamily Data availability statement.}
The data that support the findings of this study are openly available on scientific journals or in the arxiv repository.

\vfill\eject

\noindent
\textcolor{blue}{\bf\sffamily References}


\begin{thebibliography}{99}


\begin{small}


  
\bibitem{Berge:1973gth}
C.~Berge, 
{\it Graphs and hypergraphs}, \newline
\textcolor{blue}
{\href{https://www.sciencedirect.com/bookseries/north-holland-mathematical-library/vol/6}
{North-Holland Mathematical Library {\bf 6} (1973)}}.

\bibitem{Ouvrard:2020hir}
X.~Ouvrard, 
{\it Hypergraphs: an introduction and review}, \newline
\textcolor{blue}
{\href{https://arxiv.org/abs/2002.05014}
  {\sffamily arXiv:2002.05014 [cs.DM]}}.





\bibitem{Schlingemann:2003cag}
D.~Schlingemann, 
{\it Cluster states, algorithms and graphs}, \newline
\textcolor{blue}
{\href{https://www.rintonpress.com/journals/doi/QIC4.4-4.html}
{Quant. Inf. Comput. {\bf 4} (4) (2004), 287}},\newline
[\textcolor{blue}
{\href{https://arxiv.org/abs/quant-ph/0305170}
{\sffamily arXiv:quant-ph/0305170}}].

\bibitem{Hein:2003meg}
M.~Hein, J.~Eisert and H.~J.~Briegel. 
{\it Multi-party entanglement in graph states}, \newline
\textcolor{blue}
{\href{}
{Phys. Rev. A {\bf 69} (2004), 062311}},\newline
[\textcolor{blue}
{\href{arXiv:quant-ph/0307130}
  {\sffamily arXiv:quant-ph/0307130}}].




\bibitem{Schlingemann:2000ecg}
D.~Schlingemann and R.~F.~Werner, 
{\it Quantum error-correcting codes associated with graphs}, \newline
\textcolor{blue}
{\href{https://journals.aps.org/pra/abstract/10.1103/PhysRevA.65.012308}
{Phys.~Rev. A {\bf 65} (2001), 012308}},\newline
[\textcolor{blue}
{\href{https://arxiv.org/abs/quant-ph/0012111}
{\sffamily arXiv:quant-ph/0012111}}].

\bibitem{Bausch:2019tpn}
J.~Bausch and F.~Leditzky, 
{\it Error thresholds for arbitrary Pauli noise}, \newline
\textcolor{blue}
{\href{https://epubs.siam.org/doi/10.1137/20M1337375}
{SIAM J. Comput. {\bf 50} (4) (2021), 1410}},\newline
[\textcolor{blue}
{\href{https://arxiv.org/abs/1910.00471}
  {\sffamily arXiv:1910.00471 [quant-ph]}}]. 




\bibitem{Raussendorf:2001oqc}
R.~Raussendorf and H.~J.~Briegel,
{\it A one-way quantum computer}, \newline
\textcolor{blue}
{\href{https://journals.aps.org/prl/abstract/10.1103/PhysRevLett.86.5188}
{Phys. Rev. Lett. {\bf 86} (2001), 5188}}. 




\bibitem{Markham:2008gss}
D.~Markham and B.~C.~Sanders, 
{\it Graph states for quantum secret sharing}, \newline
\textcolor{blue}
{\href{https://journals.aps.org/pra/abstract/10.1103/PhysRevA.78.042309}
{Phys. Rev. A {\bf 78} (2008), 042309}},\newline
[\textcolor{blue}
{\href{https://arxiv.org/abs/0808.1532}
{\sffamily arXiv:0808.1532 [quant-ph]}}]. 




\bibitem{Scarani:2005ncs}
V.~Scarani, A.~Acín, E.~Schenck and M.~Aspelmeyer, 
{\it Nonlocality of cluster states of qubits}, \newline
\textcolor{blue}
{\href{https://journals.aps.org/pra/abstract/10.1103/PhysRevA.71.042325}
{Phys. Rev. A {\bf 71} (2005), 042325}},\newline
[\textcolor{blue}
{\href{https://arxiv.org/abs/quant-ph/0405119}
{\sffamily arXiv:quant-ph/0405119}}]. 

\bibitem{Guehne:2004big}
O.~G\"uhne, G.~T\'oth, P.~Hyllus, H.~J.~Briegel, 
{\it Bell inequalities for graph states}, \newline \pagebreak 
\textcolor{blue}
{\href{https://journals.aps.org/prl/abstract/10.1103/PhysRevLett.95.120405}
{Phys. Rev. Lett. {\bf 95} (2005), 120405}},\newline
[\textcolor{blue}
{\href{https://arxiv.org/abs/quant-ph/0410059}
  {\sffamily arXiv:quant-ph/0410059}}]. 

\bibitem{Baccari:2020bst}
F.~Baccari, R.~Augusiak, I.~\v{S}upi\'{c}, J.~Tura and A.~Ac\'in, 
{\it Scalable Bell inequalities for qubit graph states and robust self-testing}, \newline
\textcolor{blue}
{\href{https://journals.aps.org/prl/abstract/10.1103/PhysRevLett.124.020402}
{Phys. Rev. Lett. {\bf 124} (2020), 020402}},\newline
[\textcolor{blue}
{\href{https://arxiv.org/abs/1812.10428}
{\sffamily arXiv:1812.10428 [quant-ph]}}]. 




\bibitem{Kruszynska:2006epp}
C.~Kruszynska, A.~Miyake, H.~J.~Briegel and W.~D\'ur
{\it Entanglement purification protocols for all graph states}, \newline
\textcolor{blue}
{\href{https://journals.aps.org/pra/abstract/10.1103/PhysRevA.74.052316}
{Phys. Rev. A {\bf 74} (2006), 052316}},\newline
[\textcolor{blue}
{\href{https://arxiv.org/abs/quant-ph/0606090}
{\sffamily arXiv:quant-ph/0606090}}]. 

\bibitem{Toth:2005dsf}
G.~T\'oth, O.~G\"uhne, 
{\it Entanglement detection in the stabilizer formalism}, \newline
\textcolor{blue}
{\href{https://journals.aps.org/pra/abstract/10.1103/PhysRevA.72.022340}
{Phys. Rev. A {\bf 72} (2008), 022340}},\newline
[\textcolor{blue}
{\href{https://arxiv.org/abs/quant-ph/0501020}
{\sffamily arXiv:quant-ph/0501020}}]. 


\bibitem{Jungnitsch:2011ewg}
B.~Jungnitsch, T.~ Moroder and O.~G\"uhne,
{\it Entanglement witnesses for graph states: general theory and examples}, \newline
\textcolor{blue}
{\href{https://journals.aps.org/pra/abstract/10.1103/PhysRevA.84.032310}
{Phys. Rev. A {\bf 84} (2011), 032310}},\newline
[\textcolor{blue}
{\href{https://arxiv.org/abs/1106.1114}
{\sffamily arXiv:1106.1114 [quant-ph]}}]. 




\bibitem{Hein:2006ega}
M.~Hein, W.~D\"ur, J.~Eisert, R.~Raussendorf, M.~Van~den~Nest, H.-J.~Briegel, 
{\it Entanglement in graph states and its applications}, \newline
\textcolor{blue}
{\href{https://ebooks.iospress.nl/DOI/10.3254/978-1-61499-018-5-115}
{Proc. Internat. School Phys. Enrico Fermi (2008), 115}},\newline
[\textcolor{blue}
{\href{https://arxiv.org/abs/quant-ph/0602096}
{\sffamily arXiv:quant-ph/0602096}}]. 




\bibitem{Qu:2012eqs}
R.~Qu, J.~Wang, Z.-S.~Li and Y.-R.~Bao, 
{\it Encoding hypergraphs into quantum states}, \newline
\textcolor{blue}
{\href{https://journals.aps.org/pra/abstract/10.1103/PhysRevA.87.022311}
{Phys.~Rev. A {\bf 87} (2013), 022311}},\newline
[\textcolor{blue}
{\href{https://arxiv.org/abs/1211.3911}
{\sffamily arXiv:1211.3911 [quant-ph]}}]. 

\bibitem{Rossi:2012qhs}
M.~Rossi, M.~Huber, D.~Bru\ss\ and C.~Macchiavello
{\it Quantum hypergraph states}, \newline
\textcolor{blue}
{\href{https://iopscience.iop.org/article/10.1088/1367-2630/15/11/113022}
{New J. Phys. {\bf 15} (2013), 113022}},\newline
[\textcolor{blue}
{\href{https://arxiv.org/abs/1211.5554}
{\sffamily arXiv:1211.5554 [quant-ph]}}]. 




\bibitem{Wagner:2017shs}
T.~Wagner, H.~Kampermann and D.~Bru\ss, 
{\it Analysis of quantum error correction with symmetric hypergraph states}, \newline
\textcolor{blue}
{\href{https://iopscience.iop.org/article/10.1088/1751-8121/aaad6e}
{J.~Phys.~A:~Math.~Theor. {\bf 51} (2018), 125302}},\newline
[\textcolor{blue}
{\href{https://arxiv.org/abs/1711.00295}
{\sffamily arXiv:1711.00295 [quant-ph]}}]. 




\bibitem{Gachechiladze:2019mbh}
M.~Gachechiladze, O.~G\"uhne, and A.~Miyake, 
{\it Changing the circuit-depth complexity of measurement-based quantum computation with hypergraph states}, \newline
\textcolor{blue}
{\href{https://journals.aps.org/pra/abstract/10.1103/PhysRevA.99.052304}
{Phys. Rev. A {\bf 99} (2019), 052304}},\newline
[\textcolor{blue}
{\href{https://arxiv.org/abs/1805.12093}
  {\sffamily arXiv:1805.12093 [quant-ph]}}]. 

\bibitem{Takeuchi:2018qcu}
Y.~Takeuchi, T.~Morimae and M.~Hayashi, 
{\it Quantum computational universality of hypergraph states with Pauli-X and Z basis measurements}, \newline
\textcolor{blue}
{\href{https://www.nature.com/articles/s41598-019-49968-3}
{Sci. Rep. {\bf 9} (2019), 13585}},\newline
[\textcolor{blue}
{\href{https://arxiv.org/abs/1809.07552}
{\sffamily arXiv:1809.07552 [quant-ph]}}]. 




\bibitem{Morimae:2017vmq}
T.~Morimae, Y.~Takeuchi and M.~Hayashi, 
{\it Verification of hypergraph states}, \newline
\textcolor{blue}
{\href{https://journals.aps.org/pra/abstract/10.1103/PhysRevA.96.062321}
{Phys. Rev. A {\bf 96} (2017), 062321}},\newline
[\textcolor{blue}
{\href{https://arxiv.org/abs/1701.05688}
  {\sffamily arXiv:1701.05688 [quant-ph]}}]. 


\bibitem{Zhu:2018evh}
H.-J.~Zhu and M.~Hayashi, 
{\it Efficient verification of hypergraph states}, \newline
\textcolor{blue}
{\href{https://journals.aps.org/prapplied/abstract/10.1103/PhysRevApplied.12.054047}
{Phys.~Rev.~Applied {\bf 12} (2019), 054047}},\newline
[\textcolor{blue}
{\href{https://arxiv.org/abs/1806.05565}
{\sffamily arXiv:1806.05565 [quant-ph]}}]. 




\bibitem{Gachechiladze:2015evr}
M.~Gachechiladze, C.~Budroni and O.~G\"uhne, 
{\it Extreme violation of local realism in quantum hypergraph states}, \newline
\textcolor{blue}
{\href{https://journals.aps.org/prl/abstract/10.1103/PhysRevLett.116.070401}
{Phys. Rev. Lett. {\bf 116} (2016), 070401}},\newline
[\textcolor{blue}
{\href{https://arxiv.org/abs/1507.03570}
{\sffamily arXiv:1507.03570 [quant-ph]}}]. 





\bibitem{Guehne:2014nch}
O.~G\"uhne, M.~Cuquet, F.~E.~S.~Steinhoff, T.~Moroder, M.~Rossi, D.~Bru\ss, B.~Kraus and C.~Macchiavello,
{\it Entanglement and nonclassical properties of hypergraph states}, \newline
\textcolor{blue}
{\href{https://iopscience.iop.org/article/10.1088/1751-8113/47/33/335303}
{J.~Phys.~A:~Math.~Theor. {\bf 47} (2014), 335303}},\newline
[\textcolor{blue}
{\href{https://arxiv.org/abs/1404.6492}
  {\sffamily arXiv:1404.6492 [quant-ph]}}]. 

\bibitem{Ghio:2017med}
M.~Ghio, D.~Malpetti, M.~Rossi, D.~Bru\ss\ and C.~Macchiavello,
{\it Multipartite entanglement detection for hypergraph states}, \newline 
\textcolor{blue}
{\href{https://iopscience.iop.org/article/10.1088/1751-8121/aa99c9}
{J. Phys. A: Math. Theor. {\bf 51} (2017), 045302}}, \newline
[\textcolor{blue}
{\href{https://arxiv.org/abs/1703.00429}
{\sffamily arXiv:1703.00429 [quant-ph]}}]. 




\bibitem{Gachechiladze:2019phd}
M.~Gachechiladze,
{\it Quantum hypergraph states and the theory of multiparticle entanglement}, \newline 
Ph. D. Thesis, University of Siegen 
\textcolor{blue}
{\href{https://dspace.ub.uni-siegen.de/bitstream/ubsi/1509/2/Dissertation_Mariami_Gachechiladze.pdf}
{University of Siegen Thesis Archive (2019)}}.




\bibitem{Ashikhmin:2000nbc}
A.~Ashikhmin and E.~Knill, 
{\it Nonbinary quantum \pagebreak stabilizer codes}, \newline 
\textcolor{blue}
{\href{https://ieeexplore.ieee.org/document/959288}
{IEEE Trans. Inf. Theory {\bf 47} (7) (2001), 3065}},\newline
[\textcolor{blue}
{\href{https://arxiv.org/abs/quant-ph/0005008}
{\sffamily arXiv:quant-ph/0005008}}].  

\bibitem{Gheorghiu:2011qsg}
V.~Gheorghiu, 
{\it Standard form of qudit stabilizer groups}, \newline
\textcolor{blue}
{\href{https://www.sciencedirect.com/science/article/abs/pii/S0375960113011080?via\%3Dihub}
{Phys.~Lett. A {\bf 378} (2014), 505}},\newline
[\textcolor{blue}
{\href{https://arxiv.org/abs/1101.1519}
{\sffamily arXiv:1101.1519 [quant-ph]}}]. 

\bibitem{Wang:2020qhd}
Y.~Wang, Z.~Hu, B.~C.~Sanders and S.~Kais, 
{\it Qudits and high-dimensional quantum computing}, \newline
\textcolor{blue}
{\href{https://www.frontiersin.org/journals/physics/articles/10.3389/fphy.2020.589504/full}
{Front. Phys. {\bf 8} (2020), 589504}},\newline
[\textcolor{blue}
{\href{https://arxiv.org/abs/2008.00959}
{\sffamily arXiv:2008.00959 [quant-ph]}}]. 





\bibitem{Helwig:2013amq}
W.~Helwig, 
{\it Absolutely maximally entangled qudit graph states}, \newline
\textcolor{blue}
{\href{https://arxiv.org/abs/1306.2879}
{\sffamily arXiv:1306.2879 [quant-ph]}}. 

\bibitem{Keet:2010qss}
A.~Keet, B.~Fortescue, D.~Markham and B.~C.~Sanders, 
{\it Quantum secret sharing with qudit graph states}, \newline
\textcolor{blue}
{\href{https://journals.aps.org/pra/abstract/10.1103/PhysRevA.82.062315}
{Phys.~Rev. A {\bf 82} (2010), 062315}},\newline
[\textcolor{blue}
{\href{https://arxiv.org/abs/1004.4619}
{\sffamily arXiv:1004.4619 [quant-ph]}}]. 




\bibitem{Steinhoff:2016:qhs}
F.~E.~S.~Steinhoff, C.~Ritz, N.~Miklin and O.~Gühne, 
{\it Qudit hypergraph states}, \newline
\textcolor{blue}
{\href{https://journals.aps.org/pra/abstract/10.1103/PhysRevA.95.052340}
{Phys.~Rev. A {\bf 95} (2017), 052340}},\newline
[\textcolor{blue}
{\href{https://arxiv.org/abs/1612.06418}
{\sffamily arXiv:1612.06418 [quant-ph]}}]. 

\bibitem{Xiong:2017qhp}
F.-L.~Xiong, Y.-Z.~Zhen, W.-F.~Cao, K.~Chen and Z.-B.~Chen, 
{\it Qudit hypergraph states and their properties}, \newline
\textcolor{blue}
{\href{https://journals.aps.org/pra/abstract/10.1103/PhysRevA.97.012323}
{Phys.~Rev. A {\bf 97} (2018), 012323}},\newline
[\textcolor{blue}
{\href{https://arxiv.org/abs/1701.07733}
  {\sffamily arXiv:1701.07733 [quant-ph]}}].  




\bibitem{Looi:2007ecq}
S.~Y.~Looi, L.~Yu, V.~Gheorghiu and R.~B.~Griffiths, 
{\it Quantum error correcting codes using qudit graph states}, \newline
\textcolor{blue}
{\href{https://journals.aps.org/pra/abstract/10.1103/PhysRevA.78.042303}
{Phys.~Rev. A {\bf 78} (2008), 042303}},\newline
[\textcolor{blue}
{\href{https://arxiv.org/abs/0712.1979}
{\sffamily arXiv:0712.1979 [quant-ph]}}]. 



\bibitem{Gottesman:1997scq}
D.~Gottesman, 
{\it Stabilizer codes and quantum error correction}, \newline
Ph. D. thesis, California Institute of Technology, 
\textcolor{blue}
{\href{http://xxx. lanl. gov/abs/quant-ph/9705052}
{\sffamily arXiv:quant-ph/9705052}}.


\bibitem{Garcia:2014gss}
H.~J.~García, I.~L.~Markov and A.~W.~Cross, \pagebreak 
{\it On the geometry of stabilizer states}, \newline
\textcolor{blue}
{\href{https://documentsdelivered.com/source/000/018/000018086/2014/052741577.php}
{Quantum Inf. Comput. {\bf 14} no. 7-8 (2014), 683}},\newline
[\textcolor{blue}
{\href{https://web3.arxiv.org/abs/1711.07848}
{\sffamily arXiv:1711.07848 [quant-ph]}}]. 




\bibitem{Kruszynska:2008lem}
C.~Kruszynska and B.~Kraus,
{\it Local entanglability and multipartite entanglement}, \newline
\textcolor{blue}
{\href{https://journals.aps.org/pra/abstract/10.1103/PhysRevA.79.052304}
{Phys. Rev. A [\bf 79} (2009), 052304},\newline
[\textcolor{blue}
{\href{https://arxiv.org/abs/0808.3862}
{\sffamily arXiv:0808.3862 [quant-ph]}}].  

 



\bibitem{Nest:2003lct}
M.~Van~den~Nest, J.~Dehaene and B.~De~Moor, 
{\it Graphical description of the action of local Clifford transformations on graph states}, \newline
\textcolor{blue}
{\href{https://journals.aps.org/pra/abstract/10.1103/PhysRevA.69.022316}
{Phys.~Rev. A {\bf 69} (2004), 022316}},\newline
[\textcolor{blue}
{\href{https://arxiv.org/abs/quant-ph/0308151}
{\sffamily arXiv:quant-ph/0308151}}]. 

\bibitem{Nest:2004lce}
M.~Van~den~Nest, J.~Dehaene and B.~De~Moor, 
{\it An efficient algorithm to recognize local Clifford equivalence of graph states}, \newline
\textcolor{blue}
{\href{https://journals.aps.org/pra/abstract/10.1103/PhysRevA.70.034302}
{Phys. Rev. A {\bf 70} (2004), 034302}},\newline
[\textcolor{blue}
{\href{https://arxiv.org/abs/quant-ph/0405023}
{\sffamily arXiv:quant-ph/0405023}}]. 

\bibitem{Nest:2004uce}
M.~Van~den~Nest, J.~Dehaene and B.~De~Moor, 
{\it Local unitary versus local Clifford equivalence of stabilizer states}, \newline
\textcolor{blue}
{\href{https://journals.aps.org/pra/abstract/10.1103/PhysRevA.71.062323}
{Phys. Rev. A {\bf 71} (2005), 062323 }},\newline
[\textcolor{blue}
{\href{https://arxiv.org/abs/quant-ph/0411115}
{\sffamily arXiv:quant-ph/0411115}}]. 

\bibitem{Bravyi:2005ghz}
S.~Bravyi, D.~Fattal and D.~Gottesman, 
{\it GHZ extraction yield for multipartite stabilizer states}, \newline
\textcolor{blue}
{\href{https://pubs.aip.org/aip/jmp/article-abstract/47/6/062106/897669/GHZ-extraction-yield-for-multipartite-stabilizer?redirectedFrom=fulltext}
{J.~Math.~Phys. {\bf 47} (2006), 062106}},\newline
[\textcolor{blue}
{\href{https://arxiv.org/abs/quant-ph/0504208}
  {\sffamily arXiv:quant-ph/0504208}}]. 




\bibitem{Lyons:2014luh}
D.~W.~Lyons, D.~J.~Upchurch, S.~N.~Walck and C~D.~Yetter, 
{\it Local unitary symmetries of hypergraph states}, \newline
\textcolor{blue}
{\href{https://iopscience.iop.org/article/10.1088/1751-8113/48/9/095301}
{J. Phys. A: Math. Theor. {\bf 48} (2015), 095301}},\newline
[\textcolor{blue}
{\href{https://arxiv.org/abs/1410.3904}
{\sffamily arXiv:1410.3904 [quant-ph]}}].  


\bibitem{Hostens:2004qma}
E.~Hostens, J.~Dehaene and B.~De~Moor, 
{\it Stabilizer states and Clifford operations for systems of arbitrary dimensions, and modular arithmetic}, \newline
\textcolor{blue}
{\href{https://journals.aps.org/pra/abstract/10.1103/PhysRevA.71.042315}
{Phys.~Rev. A {\bf 71} (2005), 042315}},\newline
[\textcolor{blue}
{\href{https://arxiv.org/abs/quant-ph/0408190}
{\sffamily arXiv:quant-ph/0408190}}]. 


\bibitem{Bahramgiri:2006cnb}
M.~Bahramgiri and S.~Beigi, \pagebreak 
{\it Graph states under the action of local Clifford group in non-binary case}, \newline
\textcolor{blue}
{\href{https://arxiv.org/abs/quant-ph/0610267}
{\sffamily arXiv:quant-ph/0610267}}. 




\bibitem{Ionicioiu:2012egq}
R.~Ionicioiu and T.~P.~Spiller, 
{\it Encoding graphs into quantum states: an axiomatic approach}, \newline
\textcolor{blue}
{\href{https://journals.aps.org/pra/abstract/10.1103/PhysRevA.85.062313}
{Phys.~Rev. A {\bf 85} (2012), 062313}},\newline
[\textcolor{blue}
{\href{https://arxiv.org/abs/1110.5681}
{\sffamily arXiv:1110.5681 [quant-ph]}}]. 




\bibitem{MacLane:1978cwm}
S.~Mac~Lane, 
{\it Categories for the working mathematician}, \newline
\textcolor{blue}
{\href{https://link.springer.com/book/10.1007/978-1-4757-4721-8}
  {Grad. Texts in Math. {\bf 5}, Springer (1978)}}.

\bibitem{Fong:2019act}
B.~Fong and D.~I.~Spivak, 
{\it An Invitation to applied category theory}, \newline
\textcolor{blue}
{\href{https://www.cambridge.org/core/books/an-invitation-to-applied-category-theory/D4C5E5C2B019B2F9B8CE9A4E9E84D6BC}
{Cambridge University Press (2019)}}.



\bibitem{Maclane:1965cta}
S.~MacLane
{\it Categorical algebra}, \newline
\textcolor{blue}
{\href{https://www.ams.org/journals/bull/1965-71-01/S0002-9904-1965-11234-4/}
{Bull. Amer. Math. Soc. {\bf 71} (1965), 40}}.


\bibitem{Boardman:1968hes}
J.~M.~Boardman and R.~M.~
{\it Homotopy-everything H-spaces}, \newline
\textcolor{blue}
{\href{https://www.ams.org/journals/bull/1968-74-06/S0002-9904-1968-12070-1/}
{Bull.~Amer.~Math.~Soc.~{\bf 74} (1968), 1117}}.

\bibitem{May:1972ils}
J.~P.~May, 
{\it The geometry of iterated loop space}, \newline
\textcolor{blue}
{\href{https://link.springer.com/book/10.1007/BFb0067491}
  {Lect.~Notes~Math.~{\bf 271}, Springer (1972)}}.

\bibitem{Markl:2002otp}
M.~Markl, S.~Shnider and J.~Stasheff
{\it Operads in algebra, topology and physics}, \newline
\textcolor{blue}
{\href{https://bookstore.ams.org/view?ProductCode=SURV/96}
{Math.~Surveys~Monogr.~Amer.~Math.~Soc.~{\bf 96} (2002)}}.



\bibitem{Mellies:2012mea}
P.-A.~Melli\`es, 
{\it Parametric monads and enriched adjunctions}, \newline
\textcolor{blue}
{\href{https://scholar.google.it/scholar?q=\%22Parametric+monads+and+enriched+adjunctions.+\%22&hl=it&as_sdt=0&as_vis=1&oi=scholart}
{Preprint available at the author's homepage (2012)}}.

\bibitem{Katsumata:2014pms}
S.~Katsumata, 
{\it Parametric effect monads and semantics of effect systems}, \newline
\textcolor{blue}
{\href{https://dl.acm.org/doi/10.1145/2535838.2535846}
{Proc.~POPL 14, ACM (2014), 633}}.

\bibitem{Mellies:2017pcm}
P.-A.~Melli\`es, 
{\it The parametric continuation monad},\newline
\textcolor{blue}
{\href{https://www.cambridge.org/core/journals/mathematical-structures-in-computer-science/article/abs/parametric-continuation-monad/E05CB5E33697E8AA1EDB49091901C6EB}
{Math.~Struct.~in~Comp.~Science {\bf 27} (5) (2017), 651}}.


\bibitem{Fujii:2019gim}
S.~Fujii,
{\it A 2-categorical study of graded and indexed monads}, \newline
\textcolor{blue}
{\href{https://arxiv.org/abs/1904.08083}
  {\sffamily arXiv:1904.08083 [math.CT]}}.



\bibitem{Wan:2011gfr}
Z.-X.~Wan, 
{\it Finite fields and Galois rings}, \newline
\textcolor{blue}
{\href{https://www.worldscientific.com/worldscibooks/10.1142/8250\#t=aboutBook}
{World Scientific (2011)}}. \pagebreak 


\bibitem{Bini:2002gfr}
G.~Bini and F.~Flamini, 
{\it Finite commutative rings and their applications}, \newline
\textcolor{blue}
{\href{https://link.springer.com/book/10.1007/978-1-4615-0957-8}
  {SECS {\bf 680}, Springer (2002)}}.

\bibitem{Kibler:2017gre}
M.~R.~Kibler, 
{\it Galois fields and Galois rings made easy}, \newline
\textcolor{blue}
{\href{https://www.sciencedirect.com/book/9781785482359/galois-fields-and-galois-rings-made-easy}
{Elsevier (2017)}}.



\bibitem{Spivak:2009hdn}
D.~I.~Spivak, 
{\it Higher-dimensional models of networks}, \newline
\textcolor{blue}
{\href{https://arxiv.org/abs/0909.4314}
{\sffamily arXiv:0909.4314 [cs.NI]}}.  






\bibitem{Gabriel:1972uds}
P.~Gabriel, 
{\it Unzerlegbare Darstellungen}, \newline
\textcolor{blue}
{\href{https://link.springer.com/article/10.1007/BF01298413}
{Manuscripta~Math.~{\bf 6} (1972), 71}}.



\bibitem{Barr:2012ccs}
M.~Barr and  C.~Wells,
{\it Category theory for computing science}, \newline
\textcolor{blue}
{\href{http://www.tac.mta.ca/tac/reprints/articles/22/tr22abs.html}
{TAC reprints {\bf 22} (2012), 1}}.



\bibitem{Koch:1972msm}
A.~Kock,
{\it Monads on symmetric monoidal closed categories,}, \newline
\textcolor{blue}
{\href{https://link.springer.com/article/10.1007/BF01220868}
{Arch. Math. {\bf 21} (1970), 1}}.





\end{small}

\end{thebibliography}
\end{document}